\documentclass[11pt]{article}%
\pdfoutput=1
\usepackage[hidelinks]{hyperref}
\usepackage{bbm,url,microtype,amssymb,amsthm,mathtools,mathrsfs,cleveref,graphicx,float,color}
\usepackage{algorithmic}
\usepackage{fullpage}
\usepackage{lmodern}
\usepackage[T1]{fontenc}
\usepackage{authblk}
\usepackage{mathdots}
\usepackage{subcaption}

\usepackage[USenglish]{babel}

\usepackage{longtable}
\usepackage{booktabs} 
\usepackage{caption}
\usepackage{subcaption}
\usepackage{multirow}
\usepackage{enumitem}

\usepackage{tikz} %
\usetikzlibrary{decorations.pathmorphing,decorations.pathreplacing,arrows.meta,positioning,fit,calc,backgrounds,shapes.arrows,matrix,intersections}
\pgfdeclarelayer{background}
\pgfsetlayers{background,main}

\tikzset{%
  gate/.style={draw, minimum size=14, fill=gray!15, inner sep=2},
  control/.style={circle, fill=black, minimum size=3, inner sep=0},
  target/.style={circle, draw, minimum size=8, inner sep=0},
  cross/.style={cross out, draw, minimum size=2.5, inner sep=0},
  wire/classical/.style={double, double distance=1pt},
  wire/solid/.style={on background layer},
  wire/inbox/.style={dashed, line cap=round},
  cell/.style={rectangle, draw=black, minimum width=1cm, minimum height=1cm,
      text height=1.5ex, text depth=.25ex, font=\scriptsize, align=center},
  celllabel/.style={font=\scriptsize, align=center},
  labelsty/.style={font=\scriptsize, align=center},
  pics/tensorbox/.style n args={3}{
    code={
      \node[draw, minimum width=1cm, minimum height=2cm, fill=gray!15,
            rounded corners=2pt, align=center, font=\scriptsize, pic actions] (B) {#1};
      \foreach \i/\t in {0/0.11, 1/0.37, 2/0.63, 3/0.89}{
        \coordinate (-in\i) at ($(B.north west)!\t!(B.south west)$);
        \coordinate (-inpin\i) at ($(-in\i)+(-.2,0)$);
      }
      \draw[wire/classical] (-inpin0) -- (-in0);
      \node[font=\scriptsize, anchor=north, yshift=20pt, xshift=4pt] at (-inpin0) {#2};
      \foreach \i/\t in {0/0.11, 1/0.37, 2/0.63, 3/0.89}{
        \coordinate (-out\i) at ($(B.north east)!\t!(B.south east)$);
        \coordinate (-outpin\i) at ($(-out\i)+(.2,0)$);
      }
      \draw[wire/classical] (-out0) -- (-outpin0);
      \node[font=\scriptsize, anchor=north, yshift=20pt, xshift=2pt] at (-outpin0) {#3};
      \coordinate (-SW) at (B.south west);
      \coordinate (-NE) at (B.north east);
    }
  },
  pics/tensorboxin/.style n args={2}{
    code={
      \node[draw, minimum width=1cm, minimum height=1.5cm, fill=gray!15,
            rounded corners=2pt, align=center, font=\scriptsize, pic actions] (B) {#1};
      \foreach \i/\t in {0/0.11, 1/0.37, 2/0.63, 3/0.89}{
        \coordinate (-in\i) at ($(B.north west)!\t!(B.south west)$);
        \coordinate (-inpin\i) at ($(-in\i)+(-.2,0)$);
      }
      \draw[wire/classical] (-inpin0) -- (-in0);
      \node[font=\scriptsize, xshift=-8pt] at (-inpin0) {#2};
      \foreach \i/\t in {0/0.11, 1/0.37, 2/0.63, 3/0.89}{
        \coordinate (-out\i) at ($(B.north east)!\t!(B.south east)$);
        \coordinate (-outpin\i) at ($(-out\i)+(.2,0)$);
      }
    }
  },
  pics/tensorboxout/.style n args={2}{
    code={
      \node[draw, minimum width=1cm, minimum height=1.5cm, fill=gray!15,
            rounded corners=2pt, align=center, font=\scriptsize, pic actions] (B) {#1};
      \foreach \i/\t in {0/0.11, 1/0.37, 2/0.63, 3/0.89}{
        \coordinate (-in\i) at ($(B.north west)!\t!(B.south west)$);
        \coordinate (-inpin\i) at ($(-in\i)+(-.2,0)$);
      }
      \foreach \i/\t in {0/0.11, 1/0.37, 2/0.63, 3/0.89}{
        \coordinate (-out\i) at ($(B.north east)!\t!(B.south east)$);
        \coordinate (-outpin\i) at ($(-out\i)+(.2,0)$);
      }
      \draw[wire/classical] (-out0) -- (-outpin0);
      \node[font=\scriptsize, xshift=8pt] at (-outpin0) {#2};
    }
  },
  pics/tensorboxwide/.style n args={1}{
    code={
      \node[draw, minimum width=1.4cm, minimum height=1.8cm, fill=gray!15,
            rounded corners=2pt, align=center, font=\scriptsize, pic actions] (B) {#1};
      \foreach \i/\t in {0/0.11, 1/0.37, 2/0.63, 3/0.89}{
        \coordinate (-in\i) at ($(B.north west)!\t!(B.south west)$);
        \coordinate (-inpin\i) at ($(-in\i)+(-.2,0)$);
      }
      \foreach \i/\t in {0/0.11, 1/0.37, 2/0.63, 3/0.89}{
        \coordinate (-out\i) at ($(B.north east)!\t!(B.south east)$);
        \coordinate (-outpin\i) at ($(-out\i)+(.2,0)$);
      }
    }
  },
  pics/tensorboxmid/.style n args={1}{
    code={
      \node[draw, minimum width=1cm, minimum height=1.5cm, fill=gray!15,
            rounded corners=2pt, align=center, font=\scriptsize, pic actions] (B) {#1};
      \foreach \i/\t in {0/0.11, 1/0.37, 2/0.63, 3/0.89}{
        \coordinate (-in\i) at ($(B.north west)!\t!(B.south west)$);
        \coordinate (-inpin\i) at ($(-in\i)+(-.2,0)$);
      }
      \foreach \i/\t in {0/0.11, 1/0.37, 2/0.63, 3/0.89}{
        \coordinate (-out\i) at ($(B.north east)!\t!(B.south east)$);
        \coordinate (-outpin\i) at ($(-out\i)+(.2,0)$);
      }
    }
  }
}

\newcommand{\wire}[3][]{%
  \begin{pgfonlayer}{background}
    \draw[wire/solid,#1] #2;
  \end{pgfonlayer}
  \foreach \b in {#3}{
    \begin{scope}
      \clip (\b-SW) rectangle (\b-NE);
      \draw[wire/inbox,#1] #2;
    \end{scope}
  }
}

\newcommand{\clipwire}[5][]{%
  \begin{pgfonlayer}{background}
    \begin{scope}
      \clip (#4) rectangle (#5);
      \draw[wire/solid,#1] #2;
    \end{scope}
  \end{pgfonlayer}
  \foreach \b in {#3}{
    \begin{scope}
      \clip (\b-SW) rectangle (\b-NE);
      \draw[wire/inbox,#1] #2;
    \end{scope}
  }
}

\newcommand{\microspace}{\mspace{.5mu}} %
\newcommand{\bigket}[1]{\bigl\lvert\microspace#1%
  \microspace\bigr\rangle} %
\newcommand{\Bigket}[1]{\Bigl\lvert\microspace#1%
  \microspace\Bigr\rangle} %
\newcommand{\bigbra}[1]{\bigl\langle\microspace#1%
  \microspace\bigr\rvert} %
\newcommand{\Bigbra}[1]{\Bigl\langle\microspace#1%
  \microspace\Bigr\rvert} %

\providecommand{\U}[1]{\protect\rule{.1in}{.1in}}
\newtheorem{thm}{Theorem}\crefname{thm}{Theorem}{Theorems}
\newtheorem{lem}[thm]{Lemma}\crefname{lem}{Lemma}{Lemmas}
\newtheorem{prp}[thm]{Proposition}\crefname{prp}{Proposition}{Propositions}
\newtheorem{cor}[thm]{Corollary}\crefname{cor}{Corollary}{Corollaries}
\newtheorem{dfn}[thm]{Definition}\crefname{dfn}{Definition}{Definitions}
\crefname{section}{Section}{Sections}
\crefname{appendix}{Appendix}{Appendices}
\numberwithin{equation}{section}
\let\oldref\ref
\renewcommand{\ref}[1]{(\oldref{#1})}
\DeclareMathOperator{\tr}{tr}

\DeclareMathOperator{\argmax}{argmax}
\newcommand{\R}{\mathbb{R}}

\newcommand{\ket}[1]{\vert #1\rangle }
\newcommand{\bra}[1]{\langle #1 \vert}
\newcommand{\ketbra}[1]{\ket{#1} \bra{#1}}
\newcommand{\braket}[2]{\left\langle #1 \middle| #2 \right\rangle}
\newcommand{\abs}[1]{{\left\vert{#1}\right\vert}}

\AddToHook{cmd/bfseries/after}{\boldmath}
\AddToHook{cmd/normalfont/after}{\unboldmath}



\begin{document}
\title{Quantum Nonlocality under Latency Constraints}
\author[1,2,3,4]{Dawei Ding\thanks{daweiding@fudan.edu.cn}}
\author[5,6]{Zhengfeng Ji\thanks{jizhengfeng@tsinghua.edu.cn}}
\author[7,8]{Pierre Pocreau}
\author[9,10]{Mingze Xu}
\author[1,4,11]{Xinyu Xu}
\affil[1]{\small Yau Mathematical Sciences Center, Tsinghua University, Beijing 100084, China}
\affil[2]{\small Beijing Institute of Mathematical Sciences and Applications (BIMSA), Beijing 101408, China}
\affil[3]{\small Center for Mathematics and Interdisciplinary Sciences, Fudan University, Shanghai 200433, China}
\affil[4]{\small Shanghai Institute for Mathematics and Interdisciplinary Sciences (SIMIS), Shanghai 200433, China}
\affil[5]{\small Department of Computer Science and Technology, Tsinghua University, Beijing 100084, China}
\affil[6]{\small Zhongguancun Laboratory, Beijing 100094, China}
\affil[7]{\small Inria, Université Grenoble Alpes, 38000 Grenoble, France}
\affil[8]{\small CNRS, Grenoble INP, LIG, Université Grenoble Alpes, 38000 Grenoble, France}
\affil[9]{\small Zhili College, Tsinghua University, Beijing 100084, China}
\affil[10]{\small Department of Electrical and Computer Engineering, University of Illinois Urbana-Champaign, \newline
Urbana, IL 61801, USA}
\affil[11]{\small Research Institute of Intelligent Complex Systems, Fudan University, Shanghai 200433, China}

\date{}
\maketitle

\begin{abstract}
    Bell inequalities are bounds on the correlations between different parties obeying a local hidden variable theory. Here, ``local'' refers to spacetime locality: the parties cannot communicate their inputs because they must produce their outputs faster than the speed-of-light delay between them. In other words, the parties must satisfy a certain \emph{latency constraint}. In this work, we explicitly incorporate spacetime locality into the formulation of Bell inequalities by imposing such a latency constraint. When the latency constraint is sufficiently tight such that no parties can communicate, this becomes a standard Bell scenario. When the latency constraint is relaxed such that a subset of the parties can communicate, we no longer have a Bell scenario, but we can again find a divide between classical and quantum behaviors. Hence, we observe that the classical-quantum gap should actually be a function of time.
    To study these more general scenarios, we introduce the mathematical framework of latency-constrained games, which models time-evolving input and output processes for spatially separated parties subject to finite communication speeds. 
    This framework allows us to systematically study the weirdness of quantum mechanics in the ``low-latency regime'' where the speed-of-light delay is non-negligible.
    Latency-constrained games can describe real-time decision-making in real-world settings that are latency-sensitive, such as high-frequency trading and distributed systems, and can reveal the utility of quantum correlations in these settings.
\end{abstract}

\newpage 

\tableofcontents

\section{Introduction}
\label{sec:intro}

One of the most fundamental concepts in quantum mechanics is that of a Bell inequality~\cite{bell1964einstein}.
A Bell inequality is a bound on the correlations 
between classical parties that are unable to communicate and are therefore described by a local hidden variable theory.
This inability to communicate is usually justified by assuming the parties are spacelike separated~\cite{Bell1976LocalBeables}. More explicitly, we impose a 
\emph{latency constraint}: after receiving their inputs, the parties must 
produce their outputs within a time limit shorter than the speed-of-light delay between any pair of parties, making communication physically impossible.
That is, \emph{the finite speed of light decreed by the theory of relativity translates latency constraints into communication constraints}.
This is referred to as closing the locality loophole in the experimental
literature and has been addressed in early
experiments~\cite{aspect1982experimental,weihs1998violation}.
Thus, we have the following physical interpretation of Bell inequalities:
\begin{center}
   \emph{Bell inequalities are bounds on the correlations between classical parties \\ under very
  tight latency constraints}.
\end{center}
This reveals that Bell inequalities cannot be separated from the notion of \emph{spacetime locality}: one event can only influence another event if the second lies within the first event's future light cone.

In this paper, we make explicit the underlying communication constraints imposed by spacetime locality in the formulation of Bell inequalities.
Using this formulation, we derive a generalization of Bell inequalities that bounds classical correlations in more general physical scenarios. 
We then analyze quantum violations of these inequalities under the same locality conditions,
thereby formalizing the ``nonlocality'' in quantum nonlocality.
We primarily use the concept of a latency constraint: after receiving their inputs, the parties must produce their outputs within some time limit.
When the latency constraint is so tight that no parties can communicate, Bell inequalities describe the bounds on the correlations between classical parties.
But what happens when we relax this latency constraint?
In the case of two parties, if the latency constraint is relaxed enough to allow
communication between them, the parties can exhibit any correlation by simply exchanging
their inputs and choosing outputs accordingly.
This entire process is classical, and so the latency constraint does not impose 
any bounds on classical correlations.

However, the situation becomes nontrivial when we go beyond two parties.
For example, consider three parties $A,B,C$ arranged collinearly such that $A$
and $B$ are close together, while $C$ is farther away, as shown
in~\Cref{fig:ABC}. For simplicity, we assume that all three parties receive their inputs simultaneously at $t=0$ and must produce their outputs at a time $t < \tau$ (that is, strictly before $\tau$).
\begin{figure}[htbp!]
  \centering
  \def\R{8}          
  \def\H{5}          
  \def\angA{45}      
  \def\angB{135}     
  \begin{tikzpicture}[scale=1,
    labelsty/.style={align=left, font=\scriptsize}]
    \useasboundingbox (-2.2,-0.6) rectangle (11.5,5);
    \begin{scope}
      \clip (-1.3,-1) rectangle (7,4.5);
      \coordinate (A) at (0,0);
      \coordinate (B) at (1,0);
      \coordinate (C) at (3.5,0);
      \node[below] at ([xshift=-1.5pt]A) {$A$};
      \node[below] at ([xshift=-1pt]B) {$B$};
      \node[below] at ([xshift=-1pt]C) {$C$};

      \draw[name path=Ar] (A) -- ++(\angA:\R);
      \draw[name path=Al] (A) -- ++(\angB:\R);
      \draw[name path=Br] (B) -- ++(\angA:\R);
      \draw[name path=Bl] (B) -- ++(\angB:\R);
      \draw[name path=Cr] (C) -- ++(\angA:\R);
      \draw[name path=Cl] (C) -- ++(\angB:\R);

      \draw[name path=Au, dashed] (A) -- ++(0,\H);
      \draw[name path=Bu, dashed] (B) -- ++(0,\H);
      \draw[name path=Cu, dashed] (C) -- ++(0,\H);
      \path[name intersections={of=Au and Bl, by=A1}];
      \path[name intersections={of=Cu and Br, by=C1}];
      \path[name intersections={of=Cu and Ar, by=C2}];
      \draw[dashed] (A1) -- ++(11,0);
      \draw[name path=Av, dashed] (A1) -- ++(-11,0);
      \draw[dashed] (C1) -- ++(11,0);
      \draw[dashed] (C1) -- ++(-11,0);
      \draw[dashed] (C2) -- ++(11,0);
      \draw[dashed] (C2) -- ++(-11,0);
    \end{scope}

    \coordinate (tau1) at ([xshift=-1.3cm]A1);
    \coordinate (tau2) at ([xshift=-1.3cm]A1 |- C1);
    \coordinate (tau3) at ([xshift=-1.3cm]A1 |- C2);

    \node[left] at (tau1) {$\tau_1$};
    \node[left] at (tau2) {$\tau_2$};
    \node[left] at (tau3) {$\tau_3$};
    \coordinate (L1) at ($(0,0)!.5!(A1)$);
    \coordinate (L2) at ($(A1)!.5!(C1)$);
    \coordinate (L3) at ($(C2)!.5!(C1)$);
    \coordinate (L4) at ($(C2) + (0,.5)$);
    \node[label={[labelsty]right:Bell scenario}] at (7.2,0 |- L1) {};
    \node[label={[labelsty]right:Partially communicating scenario 1}] at (7.2,0 |- L2) {};
    \node[label={[labelsty]right:Partially communicating scenario 2}] at (7.2,0 |- L3) {};
    \node[label={[labelsty]right:Fully communicating scenario}] at (7.2,0 |- L4) {};

    \draw[->,] (-1.3,0) -- (7.3,0) node[below] {$x$};
    \draw[->,] (-1.3,0) -- (-1.3,4.4) node[left] {$t$};
  \end{tikzpicture}
  \caption{An illustration of the three parties $A,B,C$ that are collinear on the $x$-axis and their respective light
    cones.
    We assume the three parties receive their inputs at time $t=0$. They must produce their outputs at a time $t < \tau$.
    Different latency constraints $\tau$ lead to inequivalent non-communication
    constraints.}\label{fig:ABC}
\end{figure}

\noindent
Under the tightest latency constraint $t < \tau_1$, no pair of parties can communicate
and this becomes a regular Bell scenario with three parties. 
Now, consider slightly relaxing the latency constraint to $t < \tau_2$ such that $A$ and $B$ can
communicate with each other, but neither can communicate with $C$. This corresponds to the ``Partially communicating scenario 1'' in~\Cref{fig:ABC}. 
\emph{This is now an in-between scenario that is neither a Bell scenario where no parties can
communicate nor the fully communicating scenario with no bounds on correlations}.
Indeed, due to the latency constraint, there are still bounds on the
correlations exhibited by classical parties.
For example, the CHSH inequalities between $B$ and $C$ and between $A$ and $C$
still need to be satisfied.
However, bounds that hold in the no-communication Bell scenario, such as the
CHSH inequality between $A$ and $B$, can now be violated by classical parties.
By relaxing the latency constraint further to $t< \tau_3$, communication becomes possible between
$A$ and $B$ and between $B$ and $C$,
but not between $A$ and $C$.
This is again an in-between scenario, corresponding to ``Partially communicating scenario 2'' in~\Cref{fig:ABC}.
In this case, the CHSH inequality between $B$ and $C$ can now be violated, but the
CHSH inequality between $A$ and $C$ still holds for classical parties. Relaxing the latency constraint still further to $t < \tau$ for $\tau > \tau_3$ leads to the trivial, fully communicating scenario.

This simple example makes explicit how spacetime locality is inherent in the formulation of Bell inequalities and motivates a fundamental question:
\begin{center}
  1. What are the bounds on the correlations between classical parties \\
  under different latency constraints?
\end{center}
Furthermore, what is the maximum quantum violation of these bounds?
As of this writing, the theory of quantum mechanics is the most general theory for
describing physically realizable correlations between different parties.
Hence, we are actually asking 
\begin{center}
  2. What are the physically realizable correlations
    between multiple parties \\ under different latency constraints?
\end{center}
That is, giving the parties all possible physical resources, what are the
correlations they can exhibit under different latency constraints?
This is a question of fundamental physical limits.
Now, we emphasize that looser latency constraints allow for quantum strategies that are \emph{more powerful than those available in Bell scenarios}.
Just as classical parties can send each other classical information, quantum parties can send each other quantum
information and possibly realize stronger correlations.
We do not bound the amount of information exchanged during communication,\footnote{This is in contrast to what is assumed in quantum communication complexity~\cite{buhrman2009non}. } but
only impose latency constraints that dictate which pairs of parties can communicate.

This leads to a new mathematical framework that goes beyond that of the
current theory of Bell nonlocality.
In particular, we extend the existing mathematical framework
of \textbf{nonlocal games}. 
As our extension centers around the concept of latency constraints, we will use
the name \textbf{latency-constrained (LC) games}.
Our definition will mathematically formalize the concept of a latency constraint
and the communication latency\footnote{Note that the term ``latency'' is overloaded. Communication latency is the time required to transmit information from one party to another. A latency constraint, on the other hand, is the time limit within which parties must produce their outputs after receiving their inputs. } between different parties.
We will define classical and quantum strategies as the most general operations
that the parties can perform using classical and quantum resources,
respectively, within the latency constraint.

The framework of LC games has real-world applications.
As established in previous work~\cite{brandenburger2016team,szegedy2020systems,ding2024coordinating,arun2025faster}, there are real-world scenarios that can benefit from correlations between parties acting on timescales shorter than the communication latency. This amounts to finding ``real-world nonlocal games.''
One notable example is high-frequency trading (HFT), where trading servers issue or cancel orders on timescales of microseconds or even nanoseconds~\cite{haldane2010}. For stock exchanges separated by dozens of kilometers or more, these timescales are already shorter than the speed-of-light delay~\cite{ding2024coordinating}. Now, LC games are more faithful models of real-world scenarios, such as HFT, than nonlocal games.
In the real world, even at short time scales, different parties \emph{can} communicate, but doing so takes time. Furthermore, the framework of LC games can directly answer questions relevant to applications.
For example, in HFT, the physical limits on correlations under a certain latency constraint can be interpreted as the maximum value of a figure of merit for an HFT strategy, such as low risk or overall payoffs~\cite{brandenburger2016team,ding2024coordinating}, that can be achieved within a certain time limit. That is, we can answer the question
\begin{center}
    3. What is the lowest risk or highest payoff we can achieve within a certain time limit?
\end{center}
This is exactly the question being asked in HFT, where time is literally money. 
Taking a different angle, we can ask what is the minimum time needed to achieve a certain figure of merit using classical or quantum resources. If the minimum time is shorter with quantum resources than with classical resources, \emph{this constitutes a provable quantum advantage in the time required to achieve this figure of merit}. 
This applies to other real-world scenarios that are latency-sensitive, including (classical) distributed systems~\cite{hasanpour2017quantum,ding2024coordinating,da2025entanglement,arun2025faster}. Furthermore, as mentioned in~\cite{brandenburger2015quantum,szegedy2020systems,ding2024coordinating}, non-communication between parties can arise in real-world scenarios for reasons other than latency, such as privacy concerns or a lack of centralized control, characteristics prevalent in distributed control systems~\cite{hasanpour2017quantum,viola2024quantum,tucker2024quantum}. Our LC games framework naturally applies to these scenarios where non-communication is enforced only for a subset of the parties. 

Our paper is structured as follows. In~\Cref{sec:lc} we first introduce the simplest type of LC game that is not just a nonlocal game, namely one in which a subset of the parties can engage in a single round of communication. The communication constraints can then be described by a directed graph. We provide specific examples of these games and also consider special classes of quantum strategies. In particular, we show that quantum strategies for LC games constitute genuinely different mathematical objects compared to quantum strategies for nonlocal games.
In~\Cref{sec:tau_lc} we introduce multi-step LC games, the most general class of latency-constrained games, in which parties can communicate over multiple rounds, receive multiple inputs, and produce multiple outputs over time. We prove that if the parties only receive inputs in the initial time step and only produce outputs in the final time step, one round of communication is sufficient to realize all possible quantum correlations (up to closure), \emph{thereby revealing an interesting fact about the nature of quantum mechanics}.
We also prove that adding remote quantum registers located between the parties does not lead to new correlations (up to closure).
In~\Cref{sec:numerical} we introduce numerical techniques for computing bounds on the quantum value (the supremum of the winning probabilities attainable with quantum resources) of LC games. This culminates in~\Cref{fig:threeXor}, which plots the classical and quantum values of an LC game played by parties in a fixed spatial layout \emph{as a function of time}. Our plot shows how these values change as communication becomes possible between more pairs of parties. We conclude with a discussion in~\Cref{sec:discussion}. To make the paper self-contained, we give a short introduction to nonlocal games in~\Cref{app:prelim}.
In~\Cref{app:teleportation}, we describe how to use port-based teleportation~\cite{beigi2011simplified} in LC games to prove some of our results.
In~\Cref{app:extendedXOR} and~\Cref{appendix:3_party_xor}, we give numerical and analytic results on the quantum values of extended XOR (LC) games and three-party XOR (LC) games, respectively. 

\paragraph{Related work} 
Related literature typically focuses on generalizing local hidden-variable models by allowing limited communication and determining whether quantum strategies (with no communication) can still rule out these stronger models.
In particular, unlike our paper, \emph{spacetime locality is not a central concept in previous works}. 
Furthermore, we are not aware of any work that generalizes quantum strategies by allowing quantum communication between a subset of the parties. 
There is a body of work concerning causal relationships and the types of quantum operations that can be performed~\cite{lorenz2021causal,beckman2001causal}. These works concentrate on what types of unitaries can be implemented, while our paper concerns what types of correlations can be exhibited. 

Early work by Toner and Bacon \cite{PhysRevLett.90.157904} investigated classical strategies that allow communication. This scenario was further investigated in Ref.~\cite{PhysRevA.89.042108}. Both works focused on the bipartite case with limits on the number of communicated bits. Subsequently, this scenario was extended to the multipartite case and formalized using directed acyclic graphs (DAGs)~\cite{PhysRevA.101.052339}.
Refs.~\cite{PhysRevLett.114.140403,brask2017bell,vieira} analyzed one-way communication of inputs, outputs, or other classical messages in the bipartite setting using DAGs. It was shown that with three inputs and two outputs per party, one-way output communication alone is insufficient to reproduce all quantum correlations.
The multistage game where the outcomes are announced one by one was discussed in Ref.~\cite{PhysRevA.102.042412}.  

Additionally, there are many works on Bell nonlocality in networks \cite{tavakoli2022bell}. This literature typically focuses on setups with multiple independent quantum sources distributing entanglement across network nodes. Source-independence and no-signaling assumptions lead to constraints on classical and quantum strategies, which are used to derive network Bell inequalities~\cite{PhysRevLett.116.010403} and certify full network nonlocality~\cite{PhysRevLett.128.010403}. That is, this area of research is concerned with \emph{weaker} classical and quantum strategies.
In contrast, our work concerns \emph{stronger} classical and quantum strategies that are made possible by relaxing the latency constraint.

Communication between nodes is also allowed in genuine multipartite nonlocality (GMNL)~\cite{PhysRevLett.88.170405} and the associated Svetlichny inequalities~\cite{PhysRevD.35.3066}. GMNL is defined as the set of correlations that cannot be expressed as a convex combination of correlations that are local with respect to any bipartition of the parties. For each bipartition, parties can communicate within their side of the partition but not with parties on the other side. 
However, these scenarios are very specific: communication is allowed only within a group, and all parties within that group can communicate. Prior analyses are limited to deriving an inequality (a Svetlichny inequality) that cannot be violated with classical resources even if communication is allowed within the group, and then analyzing its quantum violations using conventional quantum strategies for nonlocal games. In contrast, we consider dynamic quantum strategies involving quantum measurements and quantum communication between parties that naturally arises from spacetime locality. Such dynamic strategies do appear in~\cite{gutoski2007toward}. However, our work introduces different definitions and performs a comprehensive analysis based on our definitions:
\begin{itemize}
    \item Our definitions are mainly motivated by latency constraints, and so there is a clear physical interpretation of our mathematical framework. In particular, we allow simultaneous two-way communication between a pair of parties.
    \item We consider the most general $n$-party scenario and mathematically describe a latency constraint using a directed graph and an integer-valued latency function.
    \item We consider specific LC games and explicitly construct optimal quantum strategies.
    \item We introduce new numerical optimization algorithms for bounding winning probabilities achievable by quantum strategies and conduct numerical experiments on specific LC games.
\end{itemize}
Latency constraints are also considered in position verification \cite{Buhrman_2014}, though in the adversarial setting of certifying a party's claimed location, rather than characterizing achievable correlations for parties within such constraints. Port-based teleportation \cite{beigi2011simplified}, a standard tool for attacks on position verification, will also be central in the study of LC games.

\section{Latency-Constrained (LC) Games}\label{sec:lc}

%

We begin by introducing the simplest form of a latency-constrained game.
This corresponds to a physical scenario in which the latency constraint
permits a single round of communication among a subset of the parties.
To specify which parties can communicate, we use a
directed graph $G$,\footnote{Throughout, directed graphs are assumed to have no self-loops. } where an edge $(i,j)$ indicates that party $i$ can send
information to party $j$.
For generality, we do not assume symmetry in communication—that is, if $i$ can
communicate with $j$, it does not imply that $j$ can communicate with $i$.
This is why we use directed graphs.
We now formally define this generalization of a nonlocal game.
For succinctness, we refer to this as a \textbf{latency-constrained game}, and
use the term \textbf{multi-step latency-constrained game} to denote more complex
settings where communication can occur over multiple rounds. To disambiguate, we will sometimes use the term \textbf{simple LC game} for the former.
See \Cref{sec:tau_lc} for more information.

We use $n$ and $[n]$ to denote the number of parties and the set
$\{1, 2, \ldots, n\}$ respectively.
\begin{dfn}\label{dfn:lc}
  Let $n \geq 2$ be an integer, and let $S_i$ and $A_i$ be finite sets for each
  $i \in [n]$.
  An $n$-party \textbf{latency-constrained (LC) game} with input sets $S_{i}$
  and output sets $A_{i}$ is defined by the tuple $(\mathcal{V}, \pi, G)$, where
  the probabilistic predicate $\mathcal{V}$ is a map
  \begin{equation*}
    \mathcal{V} : \prod_{i=1}^n A_i \times \prod_{i=1}^n S_i \to [0,1],
  \end{equation*}
  $\pi$ is the input distribution over $\prod_{i=1}^n S_i$, and $G$ is a
  directed graph with $n$ vertices, referred to as the \textbf{connectivity
    graph} of the game.
\end{dfn}

\noindent Note that a conventional nonlocal game is just an LC game with $G$
being the empty graph (no edges). Like nonlocal games, we can define a \textbf{behavior} as a conditional probability distribution
$$p(a_1, a_2, \cdots, a_n \vert s_1, s_2, \cdots, s_n)$$
with corresponding \textbf{winning probability}
$$p_\mathrm{win} \coloneqq \sum_{a_i \in A_i, s_i \in S_i} \pi(s_1, s_2,\cdots, s_n)  p(a_1, a_2, \cdots, a_n \vert s_1, s_2, \cdots, s_n) \mathcal V(a_1,a_2, \cdots ,a_n \vert s_1, s_2, \cdots ,s_n).$$
Here, to distinguish inputs and outputs, we write $\mathcal V(a\vert s)$ for $\mathcal V(a,s)$.
We will also use the term \textbf{maximum algebraic value} $\omega_a$ to mean the largest possible winning probability a behavior can attain:
\begin{align*}
    \omega_a \coloneqq \sum_{s_i \in S_i} \pi(s_1, s_2,\cdots ,s_n)  \max_{a_i \in A_i} \mathcal V(a_1,a_2, \cdots ,a_n \vert s_1, s_2, \cdots, s_n).
\end{align*}
We now proceed to define strategies for LC games.

\subsection{Classical strategies}
\label{subsec:classical_strategy}
We want to understand what is the most general strategy that uses classical
resources in the LC game setting.
We first make some preliminary definitions.
For a directed graph $G = (V,E)$, given $i \in V$, define
\begin{equation*}
  N_\text{out}(i) \coloneqq \{j \in V : (i,j) \in E\}
\end{equation*}
as the \textbf{out-neighborhood} of $i$ and
\begin{equation*}
  N_\text{in}(i) \coloneqq \{j \in V : (j,i) \in E\}
\end{equation*}
as the \textbf{in-neighborhood} of $i$. We also define
\begin{equation*}
  N_\text{out}[i] \coloneqq N_\text{out}(i) \cup \{i\},\,
  N_\text{in}[i] \coloneqq N_\text{in}(i) \cup \{i\}
\end{equation*}
as the \textbf{closed out-neighborhood} and \textbf{closed in-neighborhood}, respectively. 

Following the same logic as nonlocal games, we will define classical strategies by first considering the deterministic setting. 
When everything is deterministic, since classical information can be copied, the only nontrivial information
the parties can communicate is their inputs. We thus define \textbf{deterministic
  strategies} as each party taking their local information and producing an
output.
But unlike nonlocal games, for LC games each party also has access to the inputs
of his in-neighbors.
We therefore have the following definition.
\begin{dfn}
\label{def:classical_strat}
  Let $(\mathcal V, \pi, G)$ be an LC game with input and output sets $S_{i}$
  and $A_{i}$, respectively.
  Define the \textbf{past light cone of party $i$}, denoted $H_i$, as
  \begin{equation*}
    H_i \coloneqq \prod_{j \in N_\mathrm{in}[i]} S_j.
  \end{equation*}
  A \textbf{deterministic strategy} of the game is given by functions
  $\{f_i\}_{i=1}^n$, where 
  $$f_i : H_i \to A_i.$$
\end{dfn}

\noindent The behavior realized by a deterministic strategy is simply 
\begin{equation*}
  p(\mathbf a\vert \mathbf s) \coloneqq \prod_{i=1}^n \delta_{a_i, f_i(h_i)},
\end{equation*}
where $\textbf a \in \prod_{i=1}^n A_i, \textbf s \in \prod_{i=1}^n S_i$, and 
$h_i$ is the element of $H_i$ built from the corresponding elements in $\mathbf s$. 
We can compute the number of possible deterministic behaviors:
\begin{prp}
\label{prp:num_det}
    Given an LC game $(\mathcal V, \pi, G)$, the number of possible deterministic behaviors is 
    \begin{align}
    \label{eq:num_det}
    \prod_{i=1}^n \vert A_i\vert^{\vert H_i \vert }.
    \end{align}
\end{prp}
\begin{proof}
Party $i$ has $\vert H_i\vert$ possible input values, and for each input value, he can choose from $\vert A_i \vert$ possible output values, so the total number of possible deterministic behaviors is upper bounded by~\Cref{eq:num_det}. On the other hand, if the output differs for a certain input of a certain party, the resulting two behaviors will be different. Therefore, the number is exactly~\Cref{eq:num_det}.
\end{proof}
\noindent This result can be somewhat counter-intuitive since the parties' inputs must agree for shared in-neighbors. However, each party $i$ does indeed have $\abs{H_i}$ possible inputs. What we are counting is the number of different possible deterministic strategies, not the number of different possible inputs. 

Again following the same logic as nonlocal games, in the randomized setting, we have an additional classical resource: shared randomness.\footnote{This is the hidden variable $\lambda$ in local hidden variable theories. }
We thus define a \textbf{\boldmath classical strategy} as a probabilistic mixture of deterministic strategies. That is, the set of all classical behaviors is by definition the convex hull of the set of deterministic behaviors. It is therefore a polytope with the number of vertices given in~\Cref{eq:num_det}. \emph{The hyperplanes that define this polytope via the Minkowski-Weyl theorem are then the generalization of Bell inequalities for relaxed latency constraints, where a subset of parties can undergo a single round of communication}~\cite{brunner2014bell}. This answers the first question posed in~\Cref{sec:intro}. We shall call such inequalities \textbf{LC inequalities}. LC inequalities are Bell inequalities in the limit of very tight latency constraints. We define the largest winning probability attainable by a classical behavior for an LC game as $\omega_c$:
$$ \omega_c = \max_\text{classical behaviors} p_\text{win}.$$
We call $\omega_c$ the \textbf{classical value}.

Also mentioned in~\Cref{sec:intro} is that if every pair of parties can communicate with each other, then a classical strategy can realize any possible correlation. This is the ``Fully communicating scenario'' in~\Cref{fig:ABC}, which is defined by an LC game with connectivity graph being the complete directed graph, meaning every pair of distinct vertices are connected in both directions.
We state this result explicitly in the general case:
\begin{prp}
\label{prp:trivial}
    Let $(\mathcal V, \pi, G)$ be an LC game where the connectivity graph $G$ is a complete directed graph. Then, a classical strategy can realize all possible behaviors. In particular, it can achieve the maximum algebraic value $\omega_a$.
\end{prp}
\begin{proof}
    Let $p(\textbf{a}\vert \textbf{s})$ be a behavior.  For every $\textbf{s} \in \prod_{i=1}^n S_i$, the $n$ parties each have a copy of a random variable $A^s = (A_1^s, \cdots, A_n^s)$ distributed according to the probability distribution $p(\cdot \vert \textbf s)$. Since $G$ is a complete directed graph, every party has access to all of $\textbf s$. Each party $i$ then outputs $A_i^s$ from their copy. The resulting behavior is exactly $p(\textbf a \vert \textbf s)$.

    In particular, the classical strategy can realize the behavior given by
    \begin{align*}
        p_0(\mathbf a \vert \mathbf s) \coloneqq \delta_{\mathbf a,\mathbf a_0(\mathbf s)},
    \end{align*}
    where 
    $$\mathbf a_0(\mathbf s) \in \argmax_\mathbf{a} \mathcal V(\mathbf a \vert \mathbf s).$$
    This clearly achieves $\omega_a$.
\end{proof}
Now, we observe that the core mathematical object used to describe classical strategies for LC games is exactly the same as that of conventional nonlocal games. Both are simply a set of local discrete functions $f_i$, where for LC games the argument of $f_i$ includes the inputs of in-neighbors of $i$. However, when computing the winning probability for LC games, the arguments of $f_i$ and $f_j$ must agree on shared in-neighbors of $i,j$.

\subsection{Quantum strategies}
We next consider quantum strategies for LC games.
Like classical strategies, each party can transmit information to their out-neighbors with respect to the connectivity graph $G$. Unlike classical strategies, the parties can send \emph{quantum information} to each other. We can describe such a strategy with the following definition. 
\begin{dfn}
    \label{dfn:graph_quantum}
    Let $(\mathcal V, \pi, G)$ be an LC game of $n$ parties.
    Let $B_i$, $B_{i \to j}$ be finite-dimensional Hilbert spaces for $i\in [n]$ and $j \in N_\mathrm{out}[i]$.
    A \textbf{quantum strategy} is a tuple
    \begin{align*}
        \bigl(\vert \psi \rangle, \{W_i(s_i)\}_{i \in [n], s_i \in S_i}, \{M_i\}_{i=1}^n\bigr),
    \end{align*}
    where $\vert\psi\rangle \in \bigotimes_{i=1}^n B_i$ is a quantum state, 
    \begin{align*}
        W_i(s_i) : B_i \to 
        B_i^\mathrm{out} \coloneqq \bigotimes_{\substack{j \in N_\mathrm{out}[i]}} B_{i\to j}^{\mathrm{out}}
    \end{align*}
    is an isometry for all $s_i \in S_i$, and $M_i = \{\Pi_{i, a_i}\}_{a_i \in A_i}$ is a projective measurement on
    \begin{align*}
        B_i^\mathrm{in} \coloneqq 
        \bigotimes_{\substack{j \in N_\mathrm{in}[i]}} B_{j\to i}^{\mathrm{in}}.
    \end{align*}
\end{dfn}
\noindent In words, in a quantum strategy, the parties share a quantum state $\vert \psi\rangle$. Party $i$ uses the isometry $W_i(s_i)$ which depends on his input $s_i$. The isometry maps his share of the quantum state in the Hilbert space $B_i$ to a tensor product Hilbert space $B_i^\text{out} \coloneqq \bigotimes_{\substack{j \in N_\mathrm{out}[i]}} B_{i\to j}$, where tensor factor $B_{i \to j}$ is the quantum system transmitted to party $j$. $B_{i \to i}$ is the quantum system that is left behind.\footnote{In general, the Hilbert spaces $B_{i\to j}$ can be input-dependent as well, but we can always enlarge the Hilbert space (by taking the direct sum over all inputs for example) so that in the end the isometries map to a fixed Hilbert space. }
Hence, after the transmission, party $i$ has an overall Hilbert space $B_i^\text{in} \coloneqq \bigotimes_{\substack{j \in N_\mathrm{in}[i]}} B_{j\to i}$, on which he performs the measurement $M_i$.\footnote{Here we will use projective measurements, which is sufficiently general since any POVM can be turned into a projective measurement via Naimark dilation. }
We can write the behavior realized by a quantum strategy for an LC game as 
\begin{equation}\label{eq:g_network}
  p(\mathbf a \vert \mathbf s) =
  \Bigbra{\psi\,} \, \Biggl[\bigotimes_{i=1}^n W_i(s_i)^{\dagger}\Biggr] 
  \Biggl[\bigotimes_{j=1}^n \Pi_{j, a_j} \Biggr]  
  \Biggl[\bigotimes_{i=1}^n W_i(s_i)\Biggr] \, \Bigket{\,\psi}.
\end{equation}
Note that the notation is somewhat misleading as the tensor product structure of the measurements is \emph{not} the same as that of the isometries. More explicitly, the $\bigotimes_{i=1}^n W_i(s_i)$ tensor product acts on
$ \bigotimes_{i=1}^n B_i$
and outputs
$\bigotimes_{i=1}^n B_i^\mathrm{out},$
while the $\bigotimes_{j=1}^n \Pi_{j,a_j}$ acts on
$\bigotimes_{j=1}^n B_j^\mathrm{in}.$
This is better explained by a figure: we draw a quantum strategy in~\Cref{fig:g-quantum} for an LC game where the connectivity graph $G$ is the graph 
$$v_1 \leftrightarrows  v_2 \leftrightarrows v_3.$$

We denote by $\omega_q$ the supremum of the winning probabilities attainable by a quantum behavior for an LC game:
\begin{align}
\label{eq:qvalue}
  \omega_q \coloneqq \sup_\text{quantum behaviors} p_\text{win}.
\end{align}
We use supremum because the Hilbert spaces $B_i, B_{i \to j}$ can be of any finite dimension, and the supremum may not be attained at any finite dimension~\cite{slofstra2019set}.
We call $\omega_q$ the \textbf{quantum value}. Computing the quantum value of an LC game answers the second question posed in~\Cref{sec:intro}.
When $\omega_q > \omega_c$, the corresponding LC inequality can be violated by using quantum resources.
\begin{figure}[htbp!]
  \centering
  \begin{tikzpicture}[scale=.9]
    \pgfmathsetlengthmacro{\xstep}{3.6cm}
    \pgfmathsetlengthmacro{\ystep}{2.4cm}

    \pic (box11) at (0,2*\ystep) {tensorboxin={$W_{1}$}{$s_{1}$}};
    \pic (box12) at (0,1*\ystep) {tensorboxin={$W_{2}$}{$s_{2}$}};
    \pic (box13) at (0,0*\ystep) {tensorboxin={$W_{3}$}{$s_{3}$}};

    \pic (box21) at (\xstep,2*\ystep) {tensorboxout={$M_{1}$}{$a_{1}$}};
    \pic (box22) at (\xstep,1*\ystep) {tensorboxout={$M_{2}$}{$a_{2}$}};
    \pic (box23) at (\xstep,0*\ystep) {tensorboxout={$M_{3}$}{$a_{3}$}};

    \draw (box11-in1) -- ([xshift=-18pt]box11-inpin1);
    \draw (box12-in2) -- ([xshift=-18pt]box12-inpin2);
    \draw (box13-in3) -- ([xshift=-18pt]box13-inpin3);

    \draw (box11-out1) -- (box21-in1);
    \draw (box11-out2) to[out=0,in=180] (box22-in1);

    \draw (box12-out1) to[out=0,in=180] (box21-in2);
    \draw (box12-out2) -- (box22-in2);
    \draw (box12-out3) to[out=0,in=180] (box23-in2);

    \draw (box13-out2) to[out=0,in=180] (box22-in3);
    \draw (box13-out3) -- (box23-in3);

    \draw [decorate,decoration={brace,amplitude=5pt,mirror}]
    ($(box11-inpin1)+(-.8,0)$) -- ($(box13-inpin3)+(-.8,0)$)
    node[midway,xshift=-15pt] {$\ket{\psi}$};
  \end{tikzpicture}
  \caption{A quantum strategy for an LC game, where the connectivity graph $G$
    is the graph
    $v_1 \leftrightarrows v_2 \leftrightarrows v_3$. The $W_i$ operators are the input-dependent isometries and the $M_i$ operators are the measurements. The lines in between $W_i$ and $M_i$ denote the communication between the parties.
    }\label{fig:g-quantum}
\end{figure}

We here make two remarks about the quantum strategy of an LC game.
First, we emphasize that a quantum strategy for an LC game does not
  necessarily require the physical transmission of quantum systems in real
  time.
Instead, quantum teleportation—which only requires pre-established entanglement
and real-time classical communication—is sufficient to turn quantum
communication into classical communication.

Second, note that even though the measurement appears to occur at the end of the
quantum strategy, this definition does not preclude the possibility that the
parties perform a measurement first and then transmit the outcomes to one
another.
Moreover, this measurement can be made input-dependent, thereby
recovering quantum strategies for conventional nonlocal games.
Suppose party $i$ wishes to perform a measurement described by the
input-dependent projective measurement
${\{\Pi_{i, a_i}(s_i)\}}_{a_i \in A_i}$ on his share of the quantum state
$\ket{\psi}$ and transmit the measurement result to party $j$.
He can use the isometry $W_i(s_i)$ that satisfies
\begin{align}
    W_i(s_i) \vert\psi\rangle = \sum_{a_i \in A_i} [\Pi_{i, a_i}(s_i) \vert\psi\rangle \otimes\ket{a_i}]_{B_{i\to i}}\otimes\ket{a_i}_{B_{i\to j}}.
    \label{Eq:iso_measure_and_allocate}
\end{align}
where $\vert\psi\rangle$ is an arbitrary vector in the Hilbert space $B_i$ and $\vert a_i \rangle$ denotes classical information encoded in a fixed basis. It is easy to see that $W_i(s_i)$ is an isometry for all $s_i \in S_i$. Indeed, for any quantum states $\ket{\psi},\ket{\phi}\in B_i$,
\begin{align*}
        \bra{\psi}W_i(s_i)^\dagger W_i(s_i)\ket{\phi}&=\tr[W_i(s_i)\ket{\phi}\bra{\psi}W_i(s_i)^\dagger]\\
        &=\sum_{a_i \in A_i} \tr[\Pi_{i, a_i}(s_i) \ket{\phi}\bra{\psi} \Pi_{i, a_i}(s_i)]\\
        &=\braket{\psi}{\phi}.
\end{align*}
After this isometry, parties $i$ and $j$ receive system $B_{i\to i}$ and $B_{i \to j}$, respectively, and perform measurements in the fixed basis. 
In this way, party $i$ effectively performs an input-dependent measurement and transmits the result to party $j$ while keeping a copy for himself. 


Now, let $\mathcal Q$ be the set of all quantum behaviors. Note that this set only depends on the input sets $\{S_i\}_{i=1}^n$, the output sets $\{A_i\}_{i=1}^n$, as well as the connectivity graph $G$. We can prove that this is convex via a proof similar to~\cite{pitowsky1986range}.
\begin{prp}
    For any connectivity graph $G$, input sets $S_i$, and output sets $A_i$, the set $\mathcal Q$ is convex.
    \label{prp_convex_LC}
\end{prp}
\begin{proof}
The key idea of the proof is to use the direct sum of Hilbert spaces of the strategies that realize the original two behaviors. Note that in general, if we impose a dimension constraint on the Hilbert spaces used, the set of quantum behaviors is not always convex~\cite{nonconvexity}.

Consider a convex combination of two quantum behaviors in $\mathcal{Q}$:
\begin{equation*}
    \begin{split}
        &\xi(\mathbf{a}\vert\mathbf{s})=\lambda p(\mathbf{a}\vert\mathbf{s})+(1-\lambda)p'(\mathbf{a}\vert\mathbf{s}),\\
        \quad&p(\mathbf a \vert \mathbf s) = \langle \psi \vert \bigotimes_{i=1}^n W_i(s_i)^{\dagger}  \bigotimes_{i=1}^n \Pi_{i, a_i} \bigotimes_{i=1}^n W_i(s_i) \vert \psi\rangle,\\
        &p'(\mathbf a \vert \mathbf s) = \langle \psi' \vert \bigotimes_{i=1}^n W'_i(s_i)^{\dagger} \bigotimes_{i=1}^n \Pi'_{i, a_i}  \bigotimes_{i=1}^n W'_i(s_i) \vert \psi'\rangle.
    \end{split}
\end{equation*}
Denote the Hilbert spaces for $p(\textbf a \vert \textbf s), p'(\textbf a \vert \textbf s)$ as $B_i,B_{i \to j}$ and $B_i', B_{i \to j}'$, respectively.
By replacing the underlying quantum system with the direct sum of the original two systems:
\begin{align*}
    C_i \coloneqq B_i \oplus B_i', C_{i \to j} \coloneqq B_{i\to j} \oplus B_{i\to j}'
\end{align*}
and defining the corresponding $C_i^\mathrm{out}, C_i^\mathrm{in}$ systems as in~\Cref{dfn:graph_quantum},
we can reconstruct $\xi(\mathbf{a}\vert\mathbf{s})$ as:
\begin{equation*}
    \begin{split}
        &\xi(\mathbf{a}\vert\mathbf{s})=\langle \phi \vert \bigotimes_{i=1}^n V_i(s_i)^{\dagger} R_G^\dagger \bigotimes_{i=1}^n \Xi_{i, a_i} R_G \bigotimes_{i=1}^n V_i(s_i) \vert \phi\rangle,
    \end{split}
\end{equation*}
where
\begin{equation*}
    \begin{split}
        &\ket{\phi} \coloneqq \sqrt{\lambda}\ket{\psi}\oplus\sqrt{1-\lambda}\ket{\psi'} \in \bigotimes_{i=1}^n C_i,\\
        &V_i(s_i) \coloneqq W_i(s_i)\oplus W'_i(s_i): C_i \to C_i^\mathrm{out},\\
        \quad&\Xi_{i, a_i} \coloneqq \Pi_{i, a_i}\oplus\Pi'_{i, a_i} : C_i^\mathrm{in} \to C_i^\mathrm{in},\\
        & R_G : \bigotimes_{i=1}^n C_i^\mathrm{out} \to \bigotimes_{i=1}^n C_i^\mathrm{in}.
    \end{split}
\end{equation*}
The conclusion follows.
\end{proof}

We next prove a simple result regarding $\mathcal Q$ for different connectivity graphs.
\begin{prp}
    Let $\{S_i\}_{i=1}^n$, $\{A_i\}_{i=1}^n$ be input and output sets, respectively, for $n$ parties and $G= (V,E)$ a directed graph with $n$ vertices. Let $\mathcal Q$ be the set of all possible quantum behaviors. Then, if $G' = (V, E')$ where $E' \subseteq E$, then the set of possible quantum behaviors $\mathcal Q'$ for the connectivity graph $G'$ is a subset of $\mathcal Q$.
\end{prp}
\begin{proof}
    Consider a quantum strategy $(\vert \psi \rangle, \{W_i'(s_i)\}_{i \in [n], s_i \in S_i}, \{M_i'\}_{i=1}^n)$ where the connectivity graph is $G'$. Let $p'(\textbf a \vert \textbf s)$ be the behavior realized. Define additional Hilbert spaces $B_{i\to j}$ where $(i,j) \in E \setminus E'$ to be trivial (one-dimensional) and add them to the output spaces of the isometries $W_i'(s_i)$:
    \begin{align*}
        W_i(s_i) \vert \phi\rangle_{B_i} \coloneqq W_i'(s_i) \vert \phi\rangle_{B_i} \otimes \bigotimes_{j: (i,j) \in E \setminus E'}\vert \xi_{j}\rangle_{B_{i \to j}},
    \end{align*}
    where $\vert \xi_{j}\rangle_{B_{i \to j}}$ is a normalized state. 
    We also add them to the input spaces of the measurements $M_i' = \{\Pi_{i,a_i}'\}_{a_i \in A_i}$:
    \begin{align*}
        \Pi_{i,a_i} \Bigl(\vert \phi\rangle_{B_i^\mathrm{in}} \otimes \bigotimes_{j: (j,i) \in E \setminus E'} \vert \xi_{j}\rangle_{B_{j \to i}}\Bigr) \coloneqq \bigl (\Pi_{i,a_i}' \vert\phi\rangle_{B_i^\mathrm{in}} \bigr) \otimes \bigotimes_{j: (j,i) \in E \setminus E'} \vert \xi_{j}\rangle_{B_{j \to i}}.
    \end{align*}
    Hence, $(\vert \psi \rangle, \{W_i(s_i)\}_{i \in [n], s_i \in S_i}, \{M_i\}_{i=1}^n)$, where $M_i = \{\Pi_{i,a_i}\}_{a_i \in A_i}$, is a quantum strategy where the connectivity graph is $G$. Thus, $p'(\textbf a \vert \textbf s) \in \mathcal Q$. 
\end{proof}
\noindent That is, more correlations become possible as the latency constraint is relaxed. Hence, in particular, for a growing sequence of subgraphs
\begin{align*}
    G_1 \subseteq  G_2 \subseteq \cdots
\end{align*}
that have the same vertex set, the set of quantum behaviors satisfy 
\begin{align*}
    \mathcal Q_1 \subseteq \mathcal Q_2 \subseteq \cdots.
\end{align*}
In particular, when only one round of communication is possible between pairs of parties within the latency constraint, this implies that the quantum value $\omega_q$ is non-decreasing with time. This is clearly also true for the classical value $\omega_c$.

Lastly, we remark that unlike classical strategies, quantum strategies for LC games appear to be a different mathematical object compared to quantum strategies for nonlocal games. The main reason is the addition of the isometries $W$. We will make this statement more precise in~\Cref{subsec:forwarding_strategies}.

\subsection{Examples: distributed games}\label{subsec:distributed}

In this section, we present examples of LC games where classical and quantum
strategies have different winning probabilities.
The examples we consider are ``distributed'' variants of well-known nonlocal
games, specifically the CHSH game and the magic square game.
While the original versions involve only two parties, the distributed variants
feature three parties, with two of them collectively playing the role of a
single party from the original game.
This complicates the task, but we allow communication
between these two parties.
We show that the three parties together can achieve the quantum value of the original
games, thereby demonstrating the use of communication in a quantum strategy.

\subsubsection{Distributed CHSH game}
\label{subsubsec:dist_chsh}
We first give a distributed version of the CHSH game. The input and outputs sets are respectively $S_1,S_2,S_3 = \{0,1\}$ and $A_1, A_2, A_3 =\{0,1\}$. The predicate is given by
\begin{equation*}
    \mathcal{V}(a_1,a_2,a_3\vert s_1,s_2,s_3) = 
    \begin{cases}
        1 & a_1=a_2\text{ and } (s_1\oplus s_2)\wedge s_3 = a_1\oplus a_3,\\
        0 & \text{otherwise.}
    \end{cases}
\end{equation*}
That is, each of the first two parties is playing the CHSH game with the third party. The first two parties must produce the same output. Moreover, one input to the CHSH game is distributed in the sense that it is the parity of the inputs of the first two parties $s_1 \oplus s_2$. Neither of the first two parties has access to this parity information before communication. Lastly, let $\pi$ be the uniform distribution. The setup is shown in~\Cref{fig:dist_chsh}.
\begin{figure}
  \centering
  \begin{tikzpicture}[node distance=1.7cm,
    player/.style = {draw, circle, minimum width=.9cm, fill=gray!15},
    doublearrow/.style = {double arrow, draw, inner sep=0pt,
      minimum height=1.2cm,
      minimum width=4mm,
      double arrow head extend=.3pt,
    },
    singlearrow/.style = { single arrow, draw, inner sep=0pt,
      minimum height=.7cm,
      minimum width=4mm,
      single arrow head extend=.3pt,
      rotate=-90}]

    \node[player] (P1) at (0,0) {$1$};
    \node[player, right=of P1] (P2) {$2$};
    \node[player, right=of P2] (P3) {$3$};

    \node[doublearrow] at ($(P1)!.5!(P2)$) {};

    \node[singlearrow] at ([yshift=-1cm]P1) {};
    \node[singlearrow] at ([yshift=1cm]P1) {};
    \node[singlearrow] at ([yshift=-1cm]P2) {};
    \node[singlearrow] at ([yshift=1cm]P2) {};
    \node[singlearrow] at ([yshift=-1cm]P3) {};
    \node[singlearrow] at ([yshift=1cm]P3) {};

    \node at ([yshift=-1.7cm]P1) {$a_{1}$};
    \node at ([yshift=-1.7cm]P2) {$a_{2}$};
    \node at ([yshift=-1.7cm]P3) {$a_{3}$};

    \node at ([yshift=1.7cm]P1) {$s_{1}$};
    \node at ([yshift=1.7cm]P2) {$s_{2}$};
    \node at ([yshift=1.7cm]P3) {$s_{3}$};

    \node at ([yshift=-2.3cm]P2) {$a_{1} = a_{2}$ and $a_{1} \oplus a_{3} = (s_{1} \oplus s_{2}) \land s_{3}$};
  \end{tikzpicture}
  \caption{The distributed CHSH game with connectivity graph $v_1 \leftrightarrows v_2 \quad v_3$.}\label{fig:dist_chsh}
\end{figure}

We first claim that the nonlocal game (none of the parties can communicate) defined by the predicate $\mathcal{V}(a_1,a_2,a_3\vert s_1,s_2,s_3)$ does not have a quantum advantage. This would be the case if the latency constraint is too tight for communication to be possible. In this case, the connectivity graph is the empty graph and this becomes a conventional nonlocal game with three parties. We can write the winning probability as 
\begin{equation}
    \begin{aligned}
        p_\mathrm{win} =  \frac{1}{8}\Big[&p(000|000)+p(111|000) + p(000|010)+p(111|010) \\
        + & p(000|100)+p(111|100) + p(000|110)+p(111|110) \\
        + & p(000|001)+p(111|001) + p(110|011)+p(001|011)\\
        + & p(110|101)+p(001|101) + p(000|111)+p(111|111)\Big].
    \end{aligned}
    \label{eq:p_win_for_dist_CHSH}
\end{equation}
It will be useful to express this in correlator form~\cite{brunner2014bell}. Given the input $(s_1,s_2,s_3)$, the three parties measure binary observables $K_1^{s_1},K_2^{s_2},K_3^{s_3}$, respectively, with eigenvalues $+ 1,-1$. The measured eigenvalues $+1$ and $-1$ correspond to producing output bit values $0$ and $1$, respectively. We explicitly define correlators as
\begin{align*}
    \langle K_i^{s_i} \rangle & \coloneqq \sum_{a_i=0}^1 (-1)^{a_i} \cdot p(a_i \vert s_i)\\
    \langle K_i^{s_i} K_j^{s_j} \rangle & \coloneqq \sum_{a_i,a_j =0}^1 (-1)^{a_i+a_j} \cdot p(a_i, a_j \vert s_i, s_j)\\
    \langle K_i^{s_i} K_j^{s_j} K_k^{s_k}\rangle &\coloneqq \sum_{a_i,a_j,a_k=0}^1 (-1)^{a_i+a_j+a_k} \cdot p(a_i, a_j, a_k \vert s_i, s_j, s_k),
\end{align*}
where $p(a_i \vert s_i), p(a_i,a_j\vert s_i, s_j)$ are the marginal behaviors on party $i$ and a pair of parties $i \ne j$, which are well-defined due to the no-signaling condition. The equation
\begin{equation*}
    \begin{aligned}
        p(a_1,a_2,a_3|s_1,s_2,s_3) = & \frac{1}{8}\left[1+(-1)^{a_1}\langle K_1^{s_1}\rangle + (-1)^{a_2} \langle K_2^{s_2}\rangle + (-1)^{a_3}\langle K_3^{s_3}\rangle +(-1)^{a_1 \oplus a_2} \langle K_{1}^{s_1}K_{2}^{s_2}\rangle  \right. \\
        & \left. + (-1)^{a_1 \oplus a_3} \langle K_{1}^{s_1}K_{3}^{s_3}\rangle + (-1)^{a_2 \oplus a_3} \langle K_{2}^{s_2}K_{3}^{s_3}\rangle  + (-1)^{a_1 \oplus a_2 \oplus a_3} \langle K_{1}^{s_1}K_{2}^{s_2}K_{3}^{s_3}\rangle\right]
    \end{aligned}
\end{equation*}
holds for all $a_i, s_i \in \{0,1\}$.
For the outputs $(0,0,0),(1,1,1),(1,1,0),(0,0,1)$ that appear in $p_{\mathrm{win}}$, we have
\begin{equation*}
    \begin{aligned}
        & p(0,0,0|s_1,s_2,s_3)+p(1,1,1|s_1,s_2,s_3) = \frac{1}{4}[1+\langle K_1^{s_1}K_2^{s_2} \rangle + \langle K_1^{s_1}K_3^{s_3} \rangle + \langle K_2^{s_2}K_3^{s_3} \rangle ] \\
        & p(1,1,0|s_1,s_2,s_3)+p(0,0,1|s_1,s_2,s_3) = \frac{1}{4}[1+\langle K_1^{s_1}K_2^{s_2} \rangle - \langle K_1^{s_1}K_3^{s_3} \rangle - \langle K_2^{s_2}K_3^{s_3} \rangle ].
    \end{aligned}
\end{equation*}
Thus, the winning probability $p_\mathrm{win}$ in \Cref{eq:p_win_for_dist_CHSH} can be expressed in a correlator form as 
\begin{equation}
    \begin{aligned}
        p_\mathrm{win}  =& \frac{1}{4} + \frac{1}{16}[ \langle K_1^{0}K_2^{0}\rangle + \langle K_1^{0}K_2^{1}\rangle + \langle K_1^{1}K_2^{0}\rangle + \langle K_1^{1}K_2^{1}\rangle\\
        & + \langle K_1^{0}K_3^{0}\rangle + \langle K_1^{1}K_3^{0}\rangle + \langle K_2^{0}K_3^{0}\rangle + \langle K_2^{1}K_3^{0}\rangle]
    \end{aligned}
    \label{eq:distribited_CHSH_operator}
\end{equation}
Therefore, the RHS of~\Cref{eq:distribited_CHSH_operator} is at most $\frac{3}{4}$. Note that since we used the no-signaling constraints in deriving~\Cref{eq:distribited_CHSH_operator}, $\frac{3}{4}$ is the maximum value for no-signaling behaviors rather than the maximum algebraic value $\omega_a$, which is $1$. The winning probability of $\frac 3 4$ can be attained for a classical strategy where all parties have a constant output of $0$. Hence, the quantum value, which has to be between the classical and no-signaling value, also has to be $\frac 3 4$. The claim follows.

Now, we slightly relax the latency constraint so that the first two parties can communicate with each other, but no other communication is possible. That is, we have the connectivity graph $G$ given by
$$v_1 \leftrightarrows v_2~~~~ v_3.$$
This is possible if $v_3$ is slightly farther away, see for example the ``Partially communicating scenario 1'' shown in~\Cref{fig:ABC}. Here, we overload the letters $v_1, v_2, v_3$ to denote the parties as well as the vertices in $G$. 
We define the LC game, which we will call the \textbf{distributed CHSH game}, by the tuple $(\mathcal{V},\pi,G)$.  First of all, the classical value of this LC game is still $\frac 3 4$. This is because letting $v_1, v_2$ communicate effectively makes them one party\footnote{This will be stated more explicitly when we prove~\Cref{prp:agg_classical}. } as they can share inputs and locally compute outputs using their functions $f_1, f_2$ on their combined inputs. This ``aggregated'' party now has to play the regular CHSH game with $v_3$. Note that $s_1 \oplus s_2$ is uniformly distributed if $s_1$ and $s_2$ are. Hence, the winning probability is at most $\frac 3 4$ for a classical strategy. That is, even with the relaxed latency constraint, this is the maximum winning probability:
\begin{align*}
    p_\mathrm{win} \leq \frac 3 4.
\end{align*}
\emph{This is our first example of an LC inequality}.

However, we will find that in this relaxed latency constraint regime, there is a quantum violation. That is, although there was no quantum violation at first for very tight latency constraints, by waiting a little longer a quantum violation becomes possible! 
Moreover, the quantum strategy we find wins the distributed CHSH game with probability $\cos^2 \frac \pi 8$. This is the highest possible value for a quantum strategy. To see this, we aggregate $v_1,v_2$ and consider them as one party with input $(s_1,s_2)$ and output $(a_1,a_2)$. This can only increase the winning probability. Then, this aggregated party is effectively playing the usual CHSH game with $v_3$.
The CHSH game has a quantum value of $\cos^2 \frac \pi 8$. Hence, the quantum value of the distributed CHSH game is also $\cos^2 \frac \pi 8$. 

At first glance, it is not immediately clear how a quantum strategy for an LC game as defined in~\Cref{dfn:graph_quantum} can achieve this winning probability. In particular, for a quantum strategy, the first two parties can only communicate with each other once. If they can communicate over multiple rounds such that they can send each other their inputs, make a quantum measurement, and then send back their outputs
, the solution becomes trivial. We call this \textbf{back-and-forth communication}.
In this case, $v_1$ and $v_3$ share a maximally entangled state 
$$\vert \Phi^+\rangle \coloneqq (\ket{00}+\ket{11})/\sqrt{2}.$$
In the first round of communication, $v_2$ sends his input $s_2$ to $v_1$, who then computes $s_1 \oplus s_2$ and applies the corresponding measurement for an optimal quantum strategy of the original CHSH game on his share of $\vert \Phi^+\rangle$. $v_3$ applies his measurement according to $s_3$. Then, in the second round of communication, $v_1$ sends his output $a_1$ to $v_2$, who directly uses it as his own output. This clearly wins the distributed CHSH game with probability $\cos^2 \frac \pi 8$, the quantum value of the original CHSH game.

Surprisingly, we will find that a quantum strategy for an LC game, which only involves a single round of communication, also suffices to achieve this winning probability. From a physical perspective, this is interesting because it means that we can achieve the same winning probability while consuming only \emph{half} the time. We consider a strategy where the first two parties' initial quantum system is in a maximally entangled logical qubit state with the third party:
$$\ket\psi _{B_1B_2B_3} \coloneqq \frac{1}{\sqrt{2}}(\ket{0_L0}_{(B_1B_2)B_3}+\ket{1_L1}_{(B_1B_2)B_3}).$$
Here, the logical qubit states $\vert i_L\rangle_{B_1 B_2}$ are 
\begin{equation*}
    \begin{aligned}
        & \vert 0_L\rangle_{B_1 B_2} \coloneqq \frac{1}{\sqrt{2}}\left(\vert 0 \rangle_{B_1} \vert 1\rangle_{B_2} + \vert 1 \rangle_{B_1} \vert 0\rangle_{B_2}\right), \\
        & \vert 1_L\rangle_{B_1 B_2} \coloneqq \frac{1}{\sqrt{2}}\left(\vert 0 \rangle_{B_1} \vert 0\rangle_{B_2} - \vert 1 \rangle_{B_1} \vert 1\rangle_{B_2}\right), 
    \end{aligned}
\end{equation*}
which are two-qubit states residing in the $+1$ eigenspace of $Y_{B_1} \otimes Y_{B_2}$, with $v_1$ and $v_2$ each holding one qubit. We observe  
\begin{equation}
    \begin{aligned}
        &X_{B_1}\otimes X_{B_2}  \vert 0_L\rangle_{B_1 B_2} = \vert 0_L\rangle_{B_1 B_2},\quad X_{B_1}\otimes X_{B_2}  \vert 1_L\rangle_{B_1 B_2} = - \vert 1_L\rangle_{B_1 B_2}, \\
        &Z_{B_1}\otimes Z_{B_2}  \vert 0_L\rangle_{B_1 B_2} = -\vert 0_L\rangle_{B_1 B_2},\quad Z_{B_1}\otimes Z_{B_2}  \vert 1_L\rangle_{B_1 B_2} = \vert 1_L\rangle_{B_1 B_2}.
    \end{aligned}
    \label{eq:logical_Z}
\end{equation}
That is, both $X_{B_1}\otimes X_{B_2}$ and $-Z_{B_1}\otimes Z_{B_2}$ are logical $Z$ operators on the logical qubits. Similarly, as 
\begin{equation}
    \begin{aligned}
        &X_{B_1}\otimes Z_{B_2}  \vert 0_L\rangle_{B_1 B_2} = \vert 1_L\rangle_{B_1 B_2} \\
        &Z_{B_1}\otimes X_{B_2}  \vert 0_L\rangle_{B_1 B_2} = \vert 1_L\rangle_{B_1 B_2},
    \end{aligned}
    \label{eq:logical_X}
\end{equation}
both $X_{B_1}\otimes Z_{B_2}$ and $Z_{B_1}\otimes X_{B_2}$ are logical $X$ operators on the logical qubits. 
Based on the above observations, the following strategy allows both $v_1$ and $v_2$ to perform a CHSH game with $v_3$, while ensuring that $v_1$ and $v_2$ produce the same output:
\begin{itemize}
    \item $v_1$ measures $X_{B_1}$ on his qubit when $s_1=0$, and measures $Z_{B_1}$ on his qubit when $s_1=1$. The measurement result is denoted by $m_1 \in \{\pm 1\}$;
    \item $v_2$ measures $X_{B_2}$ on his qubit when $s_2=0$, and measures $Z_{B_2}$ on his qubit when $s_2=1$. The measurement result is denoted by $m_2\in \{\pm 1\}$;
    \item $v_1$ and $v_2$ share the measurement results $m_1,m_2$ and their inputs $s_1,s_2$ using a single round of classical communication. They subsequently output 
    \begin{align*}
        a_1=a_2 = \frac{1}{2}[1 - (-1)^{s_1s_2} m_1m_2].
    \end{align*}
    \item $v_3$ measures $\frac{Z_{B_{3}}+X_{B_3}}{\sqrt{2}}$ on his qubit when $s_3=0$, and measures $\frac{Z_{B_{3}}-X_{B_3}}{\sqrt{2}}$ on his qubit when $s_3=1$. Then $v_3$ uses his measurement result $m_3 \in \{\pm 1\}$ to give output $a_3 = \frac{1}{2}(1-m_3)$; 
\end{itemize}
With this strategy, both $v_1$ and $v_2$ effectively measure the observables $X_{B_1} \otimes X_{B_2}$, $X_{B_1} \otimes Z_{B_2}$, $Z_{B_1} \otimes X_{B_2}$, $-Z_{B_1} \otimes Z_{B_2}$ after communication for inputs $(s_1,s_2) \in \{(0,0),(0,1),(1,0),(1,1)\}$ respectively. These are exactly the logical $Z$ observable for $(s_1,s_2) \in \{(0,0),(1,1)\}$ and logical $X$ observable for $(s_1,s_2) \in \{(0,1),(1,0)\}$ on their logical qubit, as proved in \Cref{eq:logical_Z,eq:logical_X} above. Each of the first two parties $v_1$ and $v_2$, paired with the third party $v_3$, is therefore implementing an optimal quantum strategy for the usual CHSH game, and the outputs of the first two parties are guaranteed to be the same. 
We compute this explicitly.
According to \Cref{eq:p_win_for_dist_CHSH}, the winning probability of this strategy in the distributed CHSH game is 
\begin{equation*}
    \begin{aligned}
    p_\text{win} = &\frac{1}{8}\left[p(m_1m_2=m_3|000) +p(m_1m_2=m_3|001)+p(m_1m_2=m_3|010) + p(m_1m_2=-m_3|011)\right. \\
    &\left.+ p(m_1m_2=m_3|100) +p(m_1m_2=-m_3|101)+p(-m_1m_2=m_3|110) + p(-m_1m_2=m_3|111)\right] \\
    \end{aligned}
\end{equation*}
Here, we have\footnote{The subscripts $B_1,B_2,B_3$ for the state and measurements are omitted for simplicity.}
\begin{equation*}
    \begin{aligned}
        &p(m_1m_2=m_3|000) + p(m_1m_2=m_3|001)\\
        =\,& \frac{1 + \langle \psi|X\otimes X \otimes \frac{Z+X}{\sqrt{2}} |\psi \rangle }{2}+ \frac{1 + \langle \psi|X\otimes X \otimes \frac{Z-X}{\sqrt{2}} |\psi\rangle }{2} \\
        =\,& 1 + \frac{1}{\sqrt{2}} \langle \psi|X\otimes X \otimes Z |\psi \rangle \\
        =\,& 1 + \frac{1}{\sqrt{2}},
    \end{aligned}
\end{equation*}
and similarly, 
\begin{equation*}
    \begin{aligned}
        &p(m_1m_2=m_3|010) + p(m_1m_2=-m_3|011) = 1 + \frac{1}{\sqrt{2}}\langle \psi |X\otimes Z \otimes X | \psi \rangle = 1 + \frac{1}{\sqrt{2}}, \\
        &p(m_1m_2=m_3|100) + p(m_1m_2=-m_3|101) = 1 + \frac{1}{\sqrt{2}}\langle \psi |Z\otimes X \otimes X | \psi \rangle = 1 + \frac{1}{\sqrt{2}}, \\
        &p(-m_1m_2=m_3|110) + p(-m_1m_2=m_3|111) = 1 - \frac{1}{\sqrt{2}}\langle \psi |Z\otimes Z \otimes Z | \psi \rangle = 1 + \frac{1}{\sqrt{2}}. \\
    \end{aligned}
\end{equation*}
Therefore, this quantum strategy can achieve a winning probability of $\cos^2 \frac{\pi}{8}$
as claimed. 

We summarize our results for the distributed CHSH game in~\Cref{tab:dist_CHSH}.
\begin{table}[htp!]
\renewcommand{\arraystretch}{1.4}
    \centering
    \begin{tabular}{|c|c|c|c|}
    \hline 
         Connectivity graph &  $\omega_c$ & $\omega_q$ & Quantum violation?\\
    \hline
         $v_1 \quad v_2 \quad v_3$ & $\frac 3 4$ & $\frac 3 4$ & No \\
    \hline
         $v_1 \leftrightarrows v_2 \quad v_3$  & $\frac 3 4$ & $\cos^2 \frac{\pi}{8}$ & Yes\\
    \hline
    \end{tabular}
    \caption{Classical and quantum values for the distributed CHSH game with different connectivity graphs. The empty graph corresponds to a conventional nonlocal game, while the non-empty graph corresponds to an LC game. }
    \label{tab:dist_CHSH}
\end{table}

\noindent We state again the miraculous conclusion: while a quantum violation is impossible under very tight latency constraints where no parties can communicate, by slightly relaxing the latency constraint so that only $v_1,v_2$ can communicate with each other, a quantum violation becomes possible! 

Additionally, the three-party distributed CHSH game can be naturally generalized to an $(n+m)$-party distributed CHSH game in the following way. Each party receives input $s_i \in \{0,1\}$, which is chosen uniformly at random, and produces output $a_i \in \{0,1\}$. The probabilistic predicate is given by
\begin{equation*}
    \mathcal{V}(\mathbf{a}|\mathbf{s}) = \begin{cases}
        1 & a_1=a_2=\cdots= a_n,\,a_{n+1}=a_{n+2}=\cdots=a_{n+m},\, (\bigoplus_{i=1}^n s_i) \land (\bigoplus_{j=n+1}^{n+m} s_j) = a_1\oplus a_{n+1}, \\
        0 & \mathrm{otherwise}.
    \end{cases}
\end{equation*}
That is, there are two groups of parties $V_1=\{v_1,v_2,\cdots,v_n\}$ and $V_2=\{v_{n+1},v_{n+2},\cdots,v_{n+m}\}$. Each group collectively acts as a single party in the original CHSH game. The two inputs for this CHSH game correspond to the parity of the inputs from the parties within each group. Here the connectivity graph is $G=K_1 \cup K_2$, with $K_1,K_2$ being complete bidirected graphs on $V_1,V_2$, respectively. 

For such an $(n+m)$-party distributed CHSH game, we can consider a strategy where the parties share the $(n+m)$-qubit entangled state 
\begin{equation*}
    \vert \psi^{(n+m)} \rangle = \frac{1}{\sqrt{2}} \left[|0_L^{(n)}\rangle \left(\cos \frac{\pi}{8}|0_L^{(m)}\rangle + \sin \frac{\pi}{8} |1_L^{(m)}\rangle \right) + |1_L^{(n)}\rangle \left(\sin \frac{\pi}{8}|0_L^{(m)}\rangle - \cos \frac{\pi}{8} |1_L^{(m)}\rangle\right)  \right],
\end{equation*}
with each party holding one qubit.
Here, $\bigket{0^{(n)}_{L}}$ and $\bigket{1^{(n)}_{L}}$ are the logical $\ket{0}$ and $\ket{1}$ of the stabilizer code with stabilizers 
$\{Y_iY_{i+1}\}_{i=1}^{n-1}$ with $Z_{1}Z_{2} \cdots Z_{n}$ serving as the logical $Z$ operator.
Then each party in the first group measures $Z$ on his qubit when receives input $0$ and measures $X$ on his qubit when he receives input $1$. It is straightforward to check that, for a Pauli string consisting of $n$ $X$ and $Z$ operators, an odd number of $X$'s corresponds to the logical $\pm X$, while an even number corresponds to the logical $\pm Z$.
Similarly, 
for the $m$ parties in the second group,
$\bigket{0^{(m)}_{L}}$ and $\bigket{1^{(m)}_{L}}$ are the logical $\ket{0}$ and $\ket{1}$ of the code stabilized by $\{Y_iY_{i+1}\}_{i=n+1}^{n+m-1}$, with $Z_{n+1}Z_{n+2} \cdots Z_{n+m}$ serving as the logical $Z$ operator. Again, each of the $m$ parties  measures $Z$ when he receives input $0$ and measures $X$ when receives input $1$, so as to implement two anticommuting logical operators $\pm X$ and $\pm Z$ depending on the parity $\bigoplus_{j=n+1}^{n+m} s_j$. 

With the above shared entangled state and measurements, if the parties can have one round of communication within each group after their measurements in which they share the measurement results in order to correct for the possible sign flip $\pm 1$
, they can win this $(n+m)$-party distributed CHSH game with probability $\cos^2 \frac{\pi}{8}$.

\subsubsection{Distributed magic square game} 
We next consider a distributed version of the magic square game. In particular, we will use a variation of the magic square game \cite{arkhipov2012extending} where one party is given a row or a column to fill in while the other party fills in a specific entry in the magic square. The first party's answer must fulfill the parity requirements of a magic square while the second party's entry must match the first party's. The input distribution is constrained so that the row or column given contains the specified entry. We draw the setup of the magic square game in~\Cref{fig:magic_square} and show the optimal quantum strategy.
\begin{figure}[htbp!]
  \centering
  \begin{tikzpicture}
    \matrix (m) [matrix of nodes, nodes={cell},
    row sep=-\pgflinewidth,
    column sep=-\pgflinewidth]{
      $+1$ & $+1$ & $+1$ \\
      $-1$ & $+1$ & $-1$ \\
      $-1$ & $+1$ & \color{red}{$?$}\\
    };
    \node[celllabel] at ([xshift=1cm]m-1-3) {$+1$};
    \node[celllabel] at ([xshift=1cm]m-2-3) {$+1$};
    \node[celllabel] at ([xshift=1cm]m-3-3) {$+1$};
    \node[celllabel] at ([yshift=-1cm]m-3-1) {$+1$};
    \node[celllabel] at ([yshift=-1cm]m-3-2) {$+1$};
    \node[celllabel] at ([yshift=-1cm]m-3-3) {$-1$};

    \matrix (o) [matrix of nodes, nodes={cell}, xshift=5cm,
    row sep=-\pgflinewidth,
    column sep=-\pgflinewidth]{
      $XI$ & $IX$ & $XX$ \\
      $IZ$ & $ZI$ & $ZZ$ \\
      $XZ$ & $ZX$ & $YY$ \\
    };
  \end{tikzpicture}
  \caption{The magic square game. An optimal quantum strategy that has unit winning probability is to measure the observables on the right-hand square. }\label{fig:magic_square}
\end{figure}

In the \textbf{distributed magic square game}, $v_1$ and $v_2$ share the role of the party in the original game that has to fill in a specific entry. Instead of getting the full coordinates of that entry, the $x$-coordinate is given to $v_1$ and the $y$-coordinate to $v_2$. More explicitly, the LC game is defined as follows.
\begin{enumerate}
    \item The connectivity graph $G$ is again given by 
    $$v_1 \leftrightarrows v_2~~~~ v_3.$$
    \item $v_1$ receives input $s_1 \in [3]$, and $v_2$ receives input $s_2 \in [3]$. They each must output $\pm1$.
    \item $v_3$ receives either $s_1$ or $s_2+3$ as his input $s_3$ with equal probability. If $s_3\in\{1,2,3\}$, he must fill in the $s_3$-th row with $\pm1$. If instead $s_3\in\{4,5,6\}$, he must fill in the $(s_3-3)$-th column with $\pm1$. 
    \item The winning condition is: 
    \begin{itemize}
        \item $v_3$ fills the row or column with the correct parity according to the usual magic square game. The parity requirements are shown in~\Cref{fig:magic_square}. 
        \item $v_1$ and $v_2$ have the same output, which should also be the same as what $v_3$ filled in for the entry $(s_1, s_2)$.
    \end{itemize}
\end{enumerate}

There exists a quantum strategy with unit winning probability. In this strategy, two pairs of maximally entangled states 
$$\ket{\Phi^+}^{\otimes 2}=\frac{1}{2}(\ket{00}+\ket{11})^{\otimes 2}$$
are shared between $v_1$ and $v_3$. After receiving their inputs, $v_1$ performs measurements in the following bases depending on $s_1$:
\begin{equation*}
    \begin{aligned}
        &s_1=1: \{\ket{{+}{+}},\ket{{+}{-}},\ket{{-}{+}},\ket{{-}{-}}\} \\
        &s_1=2: \{\ket{00},\ket{10},\ket{01},\ket{11}\} \\
        &s_1=3: \{\frac{1}{\sqrt{2}}(\ket{0{+}}+\ket{1{-}}),\frac{1}{\sqrt{2}}(\ket{0{-}}+\ket{1{+}}),\frac{1}{\sqrt{2}}(\ket{0{+}}-\ket{1{-}}),\frac{1}{\sqrt{2}}(\ket{0{-}}-\ket{1{+}})\}.
    \end{aligned}
\end{equation*}
The four measurement outcomes, as ordered above, are labeled via $(m_1,m_1')\in\{\pm 1\}^{2}$. That is, the ordering is $(+1,+1), (+1,-1), (-1, +1), (-1,-1)$.
Meanwhile, the measurement bases of $v_3$ are chosen in the same way as in the magic square game in~\cite{arkhipov2012extending}:
\begin{equation*}
    \begin{aligned}
        &s_3=1: \{\ket{{+}{+}},\ket{{+}{-}},\ket{{-}{+}},\ket{{-}{-}}\} \\
        &s_3=2: \{\ket{00},\ket{10},\ket{01},\ket{11}\} \\
        &s_3=3: \Bigl\{\frac{1}{\sqrt{2}}(\ket{0{+}}+\ket{1{-}}),\frac{1}{\sqrt{2}}(\ket{0{-}}+\ket{1{+}}),\frac{1}{\sqrt{2}}(\ket{0{+}}-\ket{1{-}}),\frac{1}{\sqrt{2}}(\ket{0{-}}-\ket{1{+}})\Bigr\}\\
        &s_3=4: \{\ket{{+}0},\ket{{+}1},\ket{{-}0},\ket{{-}1}\} \\
        &s_3=5: \{\ket{0{+}},\ket{1{+}},\ket{0{-}},\ket{1{-}}\} \\
        &s_3=6: \Bigl\{\frac{1}{\sqrt{2}}(\ket{00}+\ket{11}),\frac{1}{\sqrt{2}}(\ket{01}+\ket{10}),\frac{1}{\sqrt{2}}(\ket{00}-\ket{11}),\frac{1}{\sqrt{2}}(\ket{01}-\ket{10})\Bigr\}.
    \end{aligned}
\end{equation*}
The four outcomes are also labeled by $(m_3,m_3')\in\{+1,-1\}^{ 2}$ in the same order.

After the measurements, $v_1$ sends his measurement outcome $(m_1,m_1')$ to $v_2$, and $v_2$ sends his input $s_2$ to $v_1$. Finally, both $v_1$ and $v_2$ output the result of the function 
\begin{equation*}
        f(s_2,m_1,m_1') \coloneqq \begin{cases}
            m_1 & \text{ if } s_2=1, \\
            m_1' & \text{ if } s_2=2, \\
            m_1m_1' & \text{ if } s_2=3.
            \end{cases}
\end{equation*}
Meanwhile, $v_3$ just fills the row or column specified by $s_3$ according to his measurement outcome exactly as in the quantum strategy for the variation of the magic square game. That is, if $s_3 \in \{1,2,3\}$, the $s_3$-th row is filled with $m_3,m_3',m_3m_3'$, respectively. If instead $s_3 \in \{4,5,6\}$, the $(s_3-3)$-th column is filled with $m_3,m_3',m_3m_3'$, respectively.

We now show that this strategy wins the distributed magic square game with unit probability. 
We observe that $v_1$ and $v_2$ together are effectively measuring the observable $M(s_1, s_2)$ consisting of $v_1$'s projective measurement followed by classical post-processing $f(s_2,m,m')$. For example, for $s_1=1$, the effective observables are 
\begin{equation*}
    \begin{aligned}
        & M(1,1)=\ket{{+}{+}}\bra{{+}{+}} + \ket{{+}{-}}\bra{{+}{-}} - \ket{{-}{+}}\bra{{-}{+}} - \ket{{-}{-}}\bra{{-}{-}} = X\otimes I, \\
        & M(1,2)=\ket{{+}{+}}\bra{{+}{+}} - \ket{{+}{-}}\bra{{+}{-}} + \ket{{-}{+}}\bra{{-}{+}} - \ket{{-}{-}}\bra{{-}{-}} = I\otimes X, \\
        & M(1,3)=\ket{{+}{+}}\bra{{+}{+}} - \ket{{+}{-}}\bra{{+}{-}} - \ket{{-}{+}}\bra{{-}{+}} + \ket{{-}{-}}\bra{{-}{-}} = X\otimes X.\\
    \end{aligned}
\end{equation*}
Similarly, one can check that for other possible inputs $(s_1,s_2)$, the effective measurements $M(s_1,s_2)$ are 
\begin{equation*}
    \begin{aligned}
        &M(2,1) = I\otimes Z,~M(2,2) = Z\otimes I,~M(2,3) = Z\otimes Z, \\
        &M(3,1) = X\otimes Z,~M(3,2) = Z\otimes X,~M(3,3) = Y\otimes Y. 
    \end{aligned}
\end{equation*}
This is exactly the same as the quantum strategy shown on the right side of \Cref{fig:magic_square}. Therefore, $v_1$ and $v_2$ always produce identical outputs, and each of them can be effectively regarded as playing the variation of the magic square game \cite{arkhipov2012extending} with $v_3$ using an optimal quantum strategy. Consequently, a quantum strategy with a single round of communication is sufficient to win the distributed magic square game with unit probability. In this case we were able to combine the inputs $s_1, s_2$ to measure the correct effective observable via classical post-processing.

\subsection{Aggregating parties}
\label{subsec:aggreg}
In latency-constrained settings, we often encounter situations where the parties can be split into groups. Parties within a group can communicate, while parties in different groups cannot.\footnote{This is also the setup of a Svetlichny scenario~\cite{PhysRevD.35.3066}. However, Svetlichny scenarios do not consider the possibility of communication in a quantum strategy, which is allowed in our case by the relaxed latency constraint. } To motivate this from a latency standpoint, we can imagine the parties can live in different cities. Parties in the same city can communicate with low latencies due to geographical proximity, while parties located in different cities communicate with much higher latencies. In this scenario, it is natural to ask when it is possible to ``aggregate'' multiple parties. For the distributed CHSH game in~\Cref{subsubsec:dist_chsh}, we found that when the first two parties can communicate, the quantum value is the same as that of the two-party nonlocal game obtained by aggregating the first two parties. This was achieved by implementing the optimal quantum strategy for the two-party CHSH game in a distributed manner with the help of a single round of communication between the first two parties. In general, for what kinds of LC games can we do this?

Define a \textbf{bidirected clique} of a directed graph as a subgraph that is itself a complete directed graph (every pair of vertices is connected in both directions). 
Then, we can prove that parties in a bidirected clique of $G$ that is isolated (vertices in the clique are not connected to vertices outside of it in either direction) can be aggregated into one party in the sense that the sets of classical and quantum behaviors for the aggregated game can be approximated arbitrarily well by those of the non-aggregated game.
\noindent For example, we can start with the connectivity graph
$$ G= v_1 \leftrightarrows v_2 \quad v_3.$$
The subgraph containing $v_1$ and $v_2$ is an isolated bidirected clique. After aggregating $v_1$ and $v_2$ we obtain the connectivity graph
$$G' = v_C \quad v_3,$$
where $v_C$ has the aggregated input set $S_1 \times S_2$ and the aggregated output set $A_1 \times A_2$. Our claim is that the possible classical and quantum behaviors for the three-party LC game with connectivity graph $G$ and the two-party nonlocal game obtained by aggregating $v_1$ and $v_2$ into one party $v_C$ are essentially the same.

To capture the notion of aggregation, we make a definition.
\begin{dfn}
    Let $\mathcal G = (\mathcal V, \pi, G)$ be an LC game.
    We call $\mathcal G$ \textbf{aggregable} if $G = (V,E)$ has a bidirected clique $C = (V_C,E_C) \subseteq G$ that is isolated.
    $G_\mathrm{agg}$ is the directed graph obtained from $G$ by replacing $C$ with a single isolated vertex $v_C$. We define the input and output sets for $v_C$ as
    \begin{align*}
        S_{v_C} \coloneqq \prod_{v \in V_C} S_{v}
    \end{align*}
    and 
    \begin{align*}
        A_{v_C} \coloneqq \prod_{v \in V_C} A_{v}.
    \end{align*}
    Let $\mathcal V_\mathrm{agg}, \pi_\mathrm{agg}$ correspond to $\mathcal V, \pi$ respectively after aggregating the inputs and outputs for parties in $C$. We say that $\mathcal G$, after \textbf{aggregation}, becomes the LC game $\mathcal G_\mathrm{agg} \coloneqq (\mathcal V_\mathrm{agg}, \pi_\mathrm{agg}, G_\mathrm{agg})$.
\end{dfn}
The proof for the classical case is straightforward.
\begin{prp}
\label{prp:agg_classical}
    Let $\mathcal G = (\mathcal V, \pi, G)$ be an aggregable LC game with isolated bidirected clique $C \subseteq G$ and let $\mathcal G_\mathrm{agg}$ be the corresponding LC game after aggregation. 
    Let $\mathcal L_C$ be the set of classical behaviors for $\mathcal G$ after aggregating inputs for parties in $C$ and $\mathcal L_\mathrm{agg}$ be the set of classical behaviors for $\mathcal G_\mathrm{agg}$.
    Then,
    $$\mathcal L_C = \mathcal L_\mathrm{agg}.$$
    In particular, 
    $$\omega_c(\mathcal G) = \omega_c( \mathcal G_\mathrm{agg}).$$
\end{prp}
\begin{proof}
    Consider the case of a deterministic strategy for $\mathcal G$. Every party $v \in V_C$ receives the inputs of all the other parties in $V_C$. Hence, each output $a_{v}$ is a function of an aggregated input in $S_{v_C}$. The collective output $\prod_{v \in V_C} a_v$ is in $A_{v_C}$. Thus the statement is true for deterministic strategies. We can then take convex combinations of deterministic strategies to obtain our conclusion.
\end{proof}

That the same conclusion essentially holds in the quantum case is somewhat counterintuitive. Just as in the case of the distributed CHSH game, we may think that back-and-forth communication is necessary. In this case, a ``central'' party can receive all the relevant inputs, make the corresponding quantum measurement, and then send back the corresponding outputs. However, just as we found for the distributed CHSH game, \emph{a single round of communication is sufficient}. The main technical tool we use to prove this in the general case is \textbf{port-based teleportation (PBT)}~\cite{PhysRevLett.101.240501,beigi2011simplified}. This is an alternative to quantum teleportation where the sender and receiver share many pairs of maximally entangled states, and the quantum state is teleported to the receiver's half of one of the entangled states. The receiver's halves of the entangled states are referred to as ``ports.'' The sender performs a quantum measurement and classically communicates to the receiver which of the ports the quantum state was teleported to.
This protocol can be used for implementing a bipartite unitary $U_{AB}$ between two remote parties using only local operations, classical communication, and shared entanglement~\cite{beigi2011simplified}. This fact will be central in proving our result. For more details on PBT, we refer the reader to~\Cref{app:pbt_dfn}.

We begin by extending this protocol to implement a multipartite measurement. 
Consider a measurement channel $\mathcal M$ (quantum-classical channel) on $n_{\mathrm{in}}$ input systems. In the following, when we say ``input party,'' we mean the party that holds the corresponding input quantum system.
We index the input parties by
\begin{equation*}
V_{\mathrm{in}}
=
\{i_1,\dots,i_{n_{\mathrm{in}}} \}.
\end{equation*}
For each input party \(i_a\), let
\(B_{i_a}\)
be a finite-dimensional Hilbert space for the corresponding input system.
The overall input to $\mathcal M$ is
\begin{equation*}
    B_{\mathrm{in}}
    \coloneqq
    B_{i_1} \otimes \dots \otimes B_{i_{n_{\mathrm{in}}}},
\end{equation*}
with dimension $d_{\mathrm{in}} = \dim(B_{\mathrm{in}})$.
Let \(\rho \in \mathcal L(B_{\mathrm{in}})\) be a quantum state, where $\mathcal L(\mathcal H)$ refers to the set of linear operators on a Hilbert space $\mathcal H$.
The measurement channel $\mathcal{ M }$ is 
of the form
\begin{equation*}
    \mathcal{M}(\rho)
    =
    \sum_{x\in X}
    \tr[M_x \rho]\,
    \lvert x\rangle\langle x\rvert,
\end{equation*}
where $X$ is the set of measurement results and $\{M_x\}_{x\in X}$ is a POVM.

Let $\rho_E \in \mathcal L(E)$ be a quantum state shared among the input parties, where \(E\) is a finite-dimensional Hilbert space, and each input party $i_a$ holds subsystem $E_{i_a}$:
\begin{equation*}
    E \coloneqq E_{i_1} \otimes \dots \otimes E_{i_{n_{\mathrm{in}}}}.
\end{equation*}
For $i_a \in V_\mathrm{in}$ let
\begin{equation*}
\{
M_{i_a, x_{i_a}}
\}_{x_{i_a}\in X_{i_a}}
\end{equation*}
be a POVM implemented by input party $i_a$ on $B_{i_a} \otimes E_{i_a}$
where $X_{i_a}$ is the set of possible measurement outcomes. 
Let $f$ be a classical post-processing function on the joint measurement results 
\begin{equation*}
   f: 
X_{i_1}\times\cdots\times X_{i_{n_{\mathrm{in}}}}
\rightarrow X.
\end{equation*}
We have the following result:
\begin{lem}
\label{lem:measurement_channel_reduction}
    Let $n_\mathrm{in} \in \mathbb{Z}^+$ and $n_\mathrm{in} \geq 2$. For any joint measurement channel $\mathcal{M}$ on $n_\mathrm{in}$ input parties as defined above, and any $\varepsilon >0$, there exists a quantum state $\rho_E$ defined on a Hilbert space of dimension $d_\varepsilon$, 
    with
\begin{equation*}
\begin{aligned}
&d_\varepsilon
\le d_{\mathrm{in}}^{2 + 2\sum_{a=1}^{n_\mathrm{in}-1} \left(\prod_{r=1}^{a} n_r \right) },\\
    &n_a = 2(n_{\mathrm{in}}-1)d_{\mathrm{in}}^{2\prod_{r=1}^{a-1} n_r }m_{\varepsilon}, 
\end{aligned}
\end{equation*}
where the empty product is equal to $1$ and $m_\varepsilon \coloneqq \lceil 1/\varepsilon\rceil$, as well as
POVMs $\{M_{i_a,x_{i_a}}\}_{x_{i_a}}$, and a function $f$, 
such that 
\begin{equation*}
    \| \mathcal{E}_M - \mathcal{M}\|_\diamond <\varepsilon, 
\end{equation*}
$\mathcal E_M$ being defined by
\begin{equation}
    \label{eq:op_single_round_meas}
        \mathcal{E}_M(\rho) \coloneqq  \sum_{\substack{x_{i_a} \\ i_a \in V_\mathrm{in}}}   
    \tr \left[\left(
    \bigotimes_{i_a\in V_\mathrm{in}}   M_{i_a,x_{i_a}}
    \right) (\rho \otimes \rho_E)\right]\vert  f(\vec x) \rangle\langle f(\vec x)\vert
\end{equation}
and $\vec x \coloneqq (x_{i_a})_{i_a\in V_{\mathrm{in}}}.$
\end{lem}
\begin{proof}
    The detailed proof is given in Appendix~\ref{app:box_reduction_proof}. 
\end{proof}

We can then use this result to show that an isolated bidirected clique can be aggregated in the quantum case as well. However, this will be an approximate result. To make this more precise, we make the following definition.
\begin{dfn}
    Let $T$ be a subset of behaviors. Then, the closure of $T$, denoted by $\overline{T}$, is defined by endowing the set of behaviors with a metric space structure, where the metric is induced by the 1-norm:
    \begin{align*}
        \Vert p-q\Vert_1 \coloneqq \sum_{\textbf{a},\textbf{s}}\vert p(\textbf a\vert \textbf s) - q(\textbf a \vert \textbf s)\vert.
    \end{align*}
\end{dfn}
\noindent We now state the theorem.
\begin{thm}
\label{thm:aggregable}
    Let $\mathcal G = (\mathcal V, \pi, G)$ be an aggregable LC game with isolated bidirected clique $C \subseteq G$ and let $\mathcal G_\mathrm{agg}$ be the corresponding LC game after aggregation.
    Then, let $\mathcal Q_C$ be the set of quantum behaviors for $\mathcal G$ after aggregating inputs and outputs in $C$ and $\mathcal Q_\mathrm{agg}$ be the set of quantum behaviors for $\mathcal G_\mathrm{agg}$. Then,
    \begin{align*}
        \mathcal Q_C \subseteq \mathcal Q_\mathrm{agg} \subseteq \overline{\mathcal Q_C}.
    \end{align*}
    In particular, 
    $$\omega_q(\mathcal G) = \omega_q( \mathcal G_\mathrm{agg}).$$
 \end{thm}
\begin{proof}
    Suppose $p(\mathbf a \vert \mathbf s)$ is a quantum behavior for $\mathcal G_\mathrm{agg}$. Let $\mathbf s_C = (s_i)_{v_i \in V_C}$ be the input of $v_C$. By definition, there exists a quantum strategy    
    $$
        (\vert \psi \rangle, \{W_i(s_i)\}, \{M_i\}),
    $$
    that realizes $p(\mathbf a \vert \mathbf s)$. 
    Now, since $v_C$ is isolated, its input-dependent isometry followed by a fixed projective measurement can be replaced by an ancillary state preparation followed by an input-dependent measurement, after enlarging the local Hilbert space if necessary (see~\Cref{Eq:iso_measure_and_allocate}).
    In this quantum strategy, party $v_C$ performs a measurement $M_{B_C}(\mathbf s_C) = \{M_k^{\mathbf s_C}\}_k$ on subsystem $B_C$ of the state $\vert \psi\rangle$ depending on the input $\mathbf{s}_C$. Equivalently, the aggregate party $C$ can first prepare state $|\mathbf s_C\rangle\langle \mathbf s_C|_{B_A} \coloneqq \bigotimes_{s_i\in \mathbf s_C} |s_i\rangle\langle s_i|_{B_{A_i}}$ in an ancillary system 
    $$B_A \coloneqq \bigotimes_{v_i \in V_C} B_{A_i},$$
    where $B_{A_i}$ are classical systems of dimension $\vert S_i\vert$.
    Then he performs the input-independent measurement 
    \begin{equation}
    \label{eq:joint_M}
        M_C = \{\sum_{\mathbf s_C \in S_{v_C}} |\mathbf s_C\rangle\langle \mathbf s_C|_{B_A} \otimes M_k^{\mathbf s_C} \}_{k}.
    \end{equation}
    Hence the quantum strategy 
    $$(\vert \psi \rangle , \{W_i(s_i)\}_{v_i \neq v_C} \cup\{ \mathrm{App}(\mathbf s_C)\}, \{M_i\}_{v_i \neq v_C} \cup \{M_C\}),$$
    where $\mathrm{App}(\mathbf s_C)$ appends the state $\vert \mathbf s_C\rangle$,
    realizes the same behavior. 
    
    For the original LC game $\mathcal G$, suppose that the parties outside $C$ execute the same strategy as in the aggregated game.
    Meanwhile, for the parties in $C$, each party $i$ prepares the subsystem $B_{A_i}$ of the state $\ketbra{\mathbf s_C}$ and some party $v_0 \in V_C$ also holds the subsystem $B_C$ of $\vert \psi\rangle$.     
    By~\Cref{lem:measurement_channel_reduction},
    for any $\varepsilon >0$, the parties in $C$ can simulate the measurement $M_C$ where the $n_\mathrm{in} = \abs{V_C}$ input parties hold the quantum systems $B_{A_i}$, with one also holding $B_C$,
    with a channel $\mathcal{E}_C$ of the form in~\Cref{eq:op_single_round_meas} such that
    $$
        \|    
        \mathcal{E}_C  - \mathcal{M}_C 
        \|_\diamond < \frac{\varepsilon }{\prod_{i=1}^{n} \vert S_i\vert},
    $$
    where $\mathcal M_C$ is the measurement channel corresponding to joint measurement $M_C$ in~\Cref{eq:joint_M}. 
    Now, $\vec x$ is the combination of the results of all the parties' measurements $M_{i_a, x_{i_a}}$ in $\mathcal E_C$. Since $C$ is a bidirected clique, every party in $C$ can access all of $\vec x$ during the communication round. Now, in our case the set of measurement outcomes of $M_C$ is 
    $$X = A_{v_C} = \prod_{v_i \in V_C} A_i.$$
    Each party can then output his corresponding part of $f(\vec x) \in X$.
    
    Let $p_{\mathcal{E}_C}(\mathbf{a}\vert\mathbf{s})$ be the behavior realized by this strategy for $\mathcal{G}$,    
    and $p_{\mathrm{agg}}(\mathbf{a}\vert\mathbf{s})$ be the behavior realized by the original strategy for $\mathcal{G}_{\mathrm{agg}}$.
    By the relationship between the diamond norm and the trace norm, we have  
    \begin{equation*}
    \begin{aligned}
            &\sum_{\mathbf a} |p_{\mathrm{agg}}(\mathbf{a}\vert\mathbf{s})- p_{\mathcal{E}_C}(\mathbf a \vert \mathbf s) | \\
           \le &\left   \| \left[ \mathcal I_{\{B_i: v_i \notin C\}} \otimes \left(   
        \mathcal{E}_C - \mathcal{M}_C \right) \right] \left( |\mathbf s_C\rangle\langle \mathbf s_C|\otimes |\psi\rangle\langle \psi|\right)
        \right \|_1  \\
        \le & \|\mathcal{E}_C - \mathcal{M}_C \|_\diamond \\
        <& \frac{\varepsilon }{\prod_{i=1}^n \vert S_i\vert},
    \end{aligned}
    \end{equation*}
    which holds for all inputs $\mathbf s$. Thus, 
    $$\Vert p_{\mathrm{agg}} - p_{\mathcal{E}_C} \Vert_1  = \sum_{\textbf a, \textbf s} \vert p_{\mathrm{agg}}(\textbf a \vert \textbf s) - p_{\mathcal{E}_C}(\textbf a \vert \textbf s) \vert < \varepsilon.$$
    Hence,
    $$ \mathcal Q_\mathrm{agg} \subseteq \overline{\mathcal Q_C}.$$

    Now, for any quantum strategy for $\mathcal G$, the quantum operations plus final measurement for the parties in $C$ can be expressed by a joint measurement. We let this joint measurement be part of the quantum strategy for $\mathcal G_\mathrm{agg}$, ceteris paribus. The realized quantum behavior is clearly exactly that of the quantum strategy for $\mathcal G$. Therefore,
    $$ \mathcal Q_C \subseteq \mathcal Q_\mathrm{agg}.$$
    
    Finally, since the quantum value as defined in~\Cref{eq:qvalue} is the supremum over all winning probabilities achievable by quantum behaviors, 
    \begin{align*}
        \mathcal Q_C \subseteq \mathcal Q_\mathrm{agg} \subseteq \overline{\mathcal Q_C}
    \end{align*}
    implies 
    \begin{align*}
        \omega_q(\mathcal G) = \omega_q(\mathcal G_\mathrm{agg})
    \end{align*}
    because of the continuity of the winning probability as a function of the behavior with respect to the 1-norm.
\end{proof}

Hence, using port-based teleportation, we can also aggregate parties in an isolated, bidirected clique in the quantum case just as in the classical case. However, this result is approximate and comes at a tremendous cost. As seen in the PBT protocol used in the proof of~\Cref{lem:measurement_channel_reduction}, the dimension of the entangled state used grows extremely quickly. Indeed, assuming
$$ \frac{\varepsilon}{2(n_\mathrm{in}-1)} <1,$$
taking $d_\mathrm{in}$ to be our asymptotic variable, we have 
\begin{align*}
    d_\varepsilon = \omega(^{n_\mathrm{in}} d_\mathrm{in} ),
\end{align*}
where $^n d$ denotes the $n$-fold tetration operation $d^{d^{\cdot^{\cdot^{d}}}}$ and we are using little $\omega$ notation.
Hence, in hardware implementations, it may be impractical to actually achieve the quantum value via this PBT-based approach.
Indeed, even if we do not use PBT, an amount of entanglement superlinear in the input state dimension is necessary to implement certain bipartite measurements via a measurement channel of the form in~\Cref{eq:op_single_round_meas}~\cite{beigi2011simplified}. 
Note however that a large amount of entanglement is not always necessary to achieve the quantum value of the aggregated game, as seen in the case of the distributed CHSH game.

\subsection{Forwarding strategies}
\label{subsec:forwarding_strategies}
We defined a quantum strategy for an LC game in~\Cref{dfn:graph_quantum}. Ostensibly, the core mathematical object that describes it is not the same as that of quantum strategies for conventional nonlocal games. However, it is worthwhile to check this statement more carefully. In particular, we want to check whether using the same mathematical object as in a quantum strategy for a nonlocal game, while allowing communication of inputs as in a classical strategy for an LC game, yields the same set of behaviors.

That is, we define a strategy where parties share inputs according to the connectivity graph and then make a measurement on a shared quantum state given the received information. This strategy is indeed consistent with~\Cref{dfn:graph_quantum}: if party $i$ wants to share his input $s_i$ to party $j$, he can simply use the isometry $W_i(s_i)$ which acts on a state $\vert \psi\rangle$ on $B_i$ by
\begin{align*}
    W_i(s_i) \vert \psi\rangle_{B_i} \coloneqq \vert \psi\rangle_{i\to i} \otimes \vert s_i\rangle_{B_{i\to j}},
\end{align*}
where $\vert s_i\rangle$ denotes classical information encoded in a fixed basis. Party $j$ performs a measurement on his quantum system conditioned on the classical information encoded in the fixed basis. This results in the strategy we described. In fact, we can give a new, simpler definition for this special class of quantum strategies.
\begin{dfn}
\label{dfn:forwarding}
    Let $(\mathcal V, \pi, G)$ be an LC game. Recall 
    \begin{align*}
        H_i \coloneqq \prod_{\substack{j\in N_\mathrm{in}[i]}} S_j
    \end{align*}
    is the past light cone of party $i$.
    Let $B_i$ be a finite-dimensional Hilbert space for $i \in [n]$. A \textbf{\boldmath forwarding strategy} is a tuple 
    \begin{align*}
        (\vert \psi \rangle, \{M_i(h_i)\}_{i \in [n], h_i \in H_i})
    \end{align*}
    where $\vert\psi\rangle \in \bigotimes_{i=1}^n B_i$ is a quantum state and for fixed $i \in [n]$ and $h_i \in H_i$, $M_i(h_i) \coloneqq \{\Pi_{i, a_i}(h_i)\}_{a_i \in A_i}$ is a projective measurement on $B_i$.
\end{dfn}
\noindent We choose the word ``forwarding'' because in such strategies, each party simply ``forwards'' his input to adjacent parties. 
Using the same notation as above, the realized behavior of a forwarding strategy, which we will call a \textbf{forwarding behavior}, is given by
\begin{align*}
    p(\mathbf a\vert \mathbf s) \coloneqq \Bigbra{\psi\,} \bigotimes_{i=1}^n \Pi_{i, a_i}(h_i) \Bigket{\,\psi},
\end{align*}
where 
\begin{align*}
    h_i \coloneqq \prod_{\substack{j \in N_\mathrm{in}[i]}} s_j.
\end{align*}
In~\Cref{fig:forwarding_strategy}, we show a forwarding strategy for the same graph as~\Cref{fig:g-quantum}: $v_1 \leftrightarrows v_2 \leftrightarrows v_3$.
\begin{figure}
  \centering

  \begin{tikzpicture}[scale=.9]
    \pgfmathsetlengthmacro{\xstep}{3.6cm}
    \pgfmathsetlengthmacro{\ystep}{2.7cm}

    \pic (box11) at (0,2*\ystep) {tensorbox={$M_{1}$}{}{}};
    \pic (box12) at (0,1*\ystep) {tensorbox={$M_{2}$}{}{}};
    \pic (box13) at (0,0*\ystep) {tensorbox={$M_{3}$}{}{}};

    \node[label=left:{\scriptsize $s_{1}, s_{2}$}] (s1) at (box11-inpin0) {};
    \node[label=left:{\scriptsize $s_{1}, s_{2}, s_{3}$}] (s2) at (box12-inpin0) {};
    \node[label=left:{\scriptsize $s_{2}, s_{3}$}] (s3) at (box13-inpin0) {};

    \node[label=right:{\scriptsize $a_{1}$}] (a1) at (box11-outpin0) {};
    \node[label=right:{\scriptsize $a_{2}$}] (a2) at (box12-outpin0) {};
    \node[label=right:{\scriptsize $a_{3}$}] (a3) at (box13-outpin0) {};

    \draw (box11-in1) -- ([xshift=-40pt]box11-inpin1);
    \draw (box12-in2) -- ([xshift=-40pt]box12-inpin2);
    \draw (box13-in3) -- ([xshift=-40pt]box13-inpin3);

    \draw [decorate,decoration={brace,amplitude=5pt,mirror}]
    ($(box11-inpin1)+(-1.8,0)$) -- ($(box13-inpin3)+(-1.8,0)$)
    node[midway,xshift=-15pt] {$\ket{\psi}$};
  \end{tikzpicture}
  \caption{A forwarding strategy, where $G$ is the graph
    $v_1 \leftrightarrows v_2 \leftrightarrows v_3$.}\label{fig:forwarding_strategy}
\end{figure}

Now, it is clear that the mathematical object describing a forwarding strategy is exactly that of a quantum strategy for a conventional nonlocal game. This is exactly analogous to classical strategies for LC games.
We give the following proposition to formalize this connection.
\begin{prp}
\label{prp:forwarding_is_q}
    Let $(\mathcal V, \pi, G)$ be an LC game and let 
    $$(\vert \psi \rangle, \{M_i(h_i)\}_{i \in [n], h_i \in H_i})$$
    be a forwarding strategy.
    Then, the winning probability realized by this strategy is equal to that of a quantum strategy described by the same tuple for the nonlocal game $(\mathcal V', \pi')$, where
    \begin{align*}
        \mathcal V' : & \prod_{i=1}^n A_i \times  \prod_{i=1}^n H_i  \to [0,1]\\
        \mathcal V'(\mathbf a \vert h_1, h_2, \ldots, h_n) & \coloneqq \mathcal V(\mathbf a \vert h_1(1), h_2(2), \ldots, h_n(n)),
    \end{align*}
    where for $j \in N_\mathrm{in}[i]$, $h_i(j)$ denotes the component of
    $$h_i \in H_i\coloneqq \prod_{\substack{k \in N_\mathrm{in}[i]}} S_k$$ corresponding to $S_j$, 
    and
    \begin{align*}
        &\pi' : \prod_{i=1}^n H_i  \to  [0,1] \\
        \pi'(h_1, h_2, \ldots, h_n) \coloneqq \pi&(h_1(1), h_2(2), \ldots, h_n(n))  \prod_{(i, j) \in E} \delta_{h_i(i), h_j(i)}.
    \end{align*}
\end{prp}
\noindent In words, the equivalent nonlocal game takes as inputs for each party from the set $H_i$ and the input distribution identifies all copies of $S_i$ across different $H_i$'s. We can see this pictorially in~\Cref{fig:forwarding_strategy} where we are now considering the multiple input components for one party as one large input.
\begin{proof}
    We can see this directly by writing down the two winning probabilities and finding them to be exactly the same.
\end{proof}

We now ask the question that motivated the definition of forwarding strategies. That is,
given an LC game, let $\mathcal Q$ be the set of all quantum behaviors and $\mathcal {FW}$ be the set of all forwarding behaviors. Is $\mathcal {FW}$ equal to $\mathcal Q$? It's obvious that $\mathcal {FW} \subseteq \mathcal Q$ since forwarding strategies are special cases of quantum strategies. Furthermore, each of these sets is convex (\cite{pitowsky1986range} and~\Cref{prp_convex_LC}) and only depend on the input sets $S_i$, output sets $A_i$, and connectivity graph $G$. Interestingly, we ultimately find that these two sets are not always equal, \emph{establishing that quantum strategies for LC games does indeed involve new mathematics}. 

Consider the case of three parties and consider again the connectivity graph $G$: 
\begin{align}
\label{eq:a-bc}
    v_1 \leftrightarrows v_2 \quad v_3.
\end{align}
At a high level, our counterexample is the following: Let $ v_1, v_3 $ play the CHSH game. Then, using a quantum strategy, $ v_2 $ can simply output the measurement results of $v_1$ using the isometry in~\Cref{Eq:iso_measure_and_allocate}. A forwarding strategy cannot do this. This is because the behavior of $ v_1, v_3 $ self-tests the quantum state that $ v_1, v_3 $ are using, and therefore by monogamy of entanglement, $ v_2 $ cannot possibly know the measurement result of $v_1$. We state our result as a theorem and give a formal proof.
\begin{thm}
\label{thm:q_neq_f}
    There exist input sets $S_i$, output sets $A_i$, and a connectivity graph $G$ for which the corresponding sets of forwarding and quantum behaviors satisfy
    \begin{align*}
        \mathcal {FW} \subsetneq \mathcal Q.
    \end{align*}
\end{thm}
\begin{proof}
Let $G$ be the graph in~\Cref{eq:a-bc} with 3 vertices. Let 
$$S_1, S_3 = \{0,1\}$$
be the possible inputs of $v_1, v_3$ respectively, while $S_2$ is a singleton. Since there is only one possible input, we will suppress $s_2$ in our notation for conciseness.
Also, let 
$$A_1, A_2 , A_3 = \{0,1\} $$
be the possible outputs of $v_1, v_2 , v_3$, respectively. Consider the behavior 
$$ 
q(a_1, a_2, a_3 \vert s_1, s_3) \coloneqq 
\begin{cases}
p_\text{CHSH}(a_1, a_3 \vert s_1, s_3) & a_2 = a_1\\
0 & \text{otherwise}
\end{cases},
$$
where $ p_\text{CHSH}(a_1, a_3 \vert s_1, s_3) $ is the behavior of $ v_1, v_3 $ implementing an optimal quantum strategy for the CHSH game.

Then, a quantum strategy for an LC game can achieve this behavior where $ v_1, v_3 $ implement an optimal quantum strategy for the CHSH game, and $v_1$ sends his measurement result
$a_1$ to $ v_2 $. $ v_1, v_3 $ output their measurement results $ a_1, a_3 $, while $ v_2 $ outputs $a_1$. This clearly realizes the desired behavior.

We next show that no forwarding strategy can realize $ q(a_1, a_2, a_3 \vert s_1, s_3) $. We see that the marginal distribution $ q(a_1, a_3 \vert s_1, s_3) $ is simply $ p_\text{CHSH}(a_1, a_3\vert s_1, s_3) $. Now, this behavior self-tests for the maximally entangled state of two qubits $ \vert \Phi^+\rangle \coloneqq \frac{1}{\sqrt{2}}(\vert 00\rangle +\vert 11\rangle) $ as well as the corresponding measurements $ P $ in an optimal quantum strategy for the CHSH game. That is, if $q$ is a forwarding behavior:
$$ q(a_1, a_2, a_3 \vert s_1, s_3) = \langle\psi\vert_{B_1 B_2 B_3} \Pi_{1, a_1}(s_1) \otimes \Pi_{2, a_2}(s_1)\otimes \Pi_{3, a_3}(s_3) \vert \psi \rangle_{B_1 B_2 B_3}, $$
there exists local isometries $ V_{1}: B_1 \to B_1'\otimes \bar B_1 $ , $ V_{3}: B_3 \to B_3' \otimes \bar B_3  $ such that
\begin{align}
\label{eq:self-test}
    &(V_{1} \otimes I_{B_2} \otimes V_{3}) (\Pi_{1, a_1}(s_1) \otimes I_{B_2} \otimes \Pi_{3, a_3}(s_3)\vert\psi\rangle_{B_1 B_2 B_3}) \nonumber \\
    &= (P_{1,a_1}(s_1) \otimes P_{3,a_3}(s_3)\vert \Phi^+\rangle_{B_1'B_3'}) \otimes \vert\xi\rangle_{\bar B_1 B_2 \bar B_3},
\end{align} 
where $ \xi $ is some state. Hence,
\begin{align*}
& \langle\psi\vert \Pi_{1, a_1}(s_1) \otimes \Pi_{2, a_2}(s_1)\otimes \Pi_{3, a_3}(s_3)\vert \psi \rangle \\
&=
\langle\psi\vert (V_{1}^{\dagger} \otimes I_{B_2} \otimes  V_{3}^{\dagger})(V_{1} \otimes I_{B_2} \otimes V_{3})  (\Pi_{1, a_1}(s_1) \otimes \Pi_{2, a_2}(s_1)\otimes \Pi_{3, a_3}(s_3))\vert \psi \rangle\\
& = \langle\Phi^+ \vert_{B_1' B_3'}\langle \xi\vert_{\bar B_1 B_2 \bar B_3} (P_{1,a_1}(s_1) \otimes P_{3,a_3}(s_3)\vert \Phi^+\rangle_{B_1'B_3'}) \otimes \Pi_{2, a_2}(s_1)\vert\xi\rangle_{\bar B_1 B_2 \bar B_3}\\
& = \langle\Phi^+ \vert_{B_1'B_3'} P_{1,a_1}(s_1) \otimes P_{3,a_3}(s_3)\vert \Phi^+\rangle_{B_1'B_3'} \langle \xi\vert_{\bar B_1 B_2 \bar B_3} \Pi_{2, a_2}(s_1)\vert\xi\rangle_{\bar B_1 B_2 \bar B_3},
\end{align*}
where $V^\dagger$ is the adjoint of $V$ and the second equality follows from~\Cref{eq:self-test} and from summing it over $a_1, a_3$. This implies that conditioned on $ s_1, s_3 $, the output of $ v_2
 $ is independent of the output of $v_1$, a contradiction.
\end{proof}
\noindent In fact, for any connectivity graph $G = (V,E)$ that is not weakly connected and not empty, $\mathcal {FW} \subsetneq \mathcal Q$ for an appropriate choice of sets $A_i, S_i$. This is because we can always find a triple of vertices $v_1, v_2, v_3$ such that $(v_1, v_2)\in E$ but $v_3$ is not weakly connected to either $v_1$ or $v_2$. We can then define the LC game where all parties that are not $v_1, v_3$ have singleton input sets, $v_1, v_3$ have 1-bit inputs, $v_1, v_2, v_3$ have 1-bit outputs, and all other parties have singleton output sets. That is, effectively all the parties other than $v_1,v_2,v_3$ are irrelevant. 
By the same argument as above, $\mathcal {FW} \subsetneq \mathcal Q$.

The gap in realizable behaviors proved in~\Cref{thm:q_neq_f} also translates to a gap in achievable winning probabilities for the corresponding LC game. We consider the LC game $\mathcal G=(\mathcal V, \pi, G)$ where
\begin{align*}
\mathcal V(a_1, a_2, a_3 \vert s_1, s_3) \coloneqq \mathcal V_\text{CHSH}(a_1, a_3 \vert s_1, s_3) \cdot [a_2 = a_1],
\end{align*}
where $\mathcal V_\mathrm{CHSH}$ is the winning condition for the CHSH game, $[\cdot = \cdot]$ is an indicator function,
and $G$ is given by~\Cref{eq:a-bc}. We also assume a uniform input distribution $\pi(s_1,s_3) \coloneqq \frac 1 4$. We will call this the \textbf{extended CHSH game} following~\cite{toner2009monogamy}, but as an LC game with partial communication allowed. We show the extended CHSH game schematically in~\Cref{fig:ext_chsh}.
\begin{figure}[htbp!]
  \centering
  \begin{tikzpicture}[node distance=1.7cm,
    player/.style = {draw, circle, minimum width=.9cm, fill=gray!15},
    doublearrow/.style = {double arrow, draw, inner sep=0pt,
      minimum height=1.2cm,
      minimum width=4mm,
      double arrow head extend=.3pt,
    },
    singlearrow/.style = { single arrow, draw, inner sep=0pt,
      minimum height=.7cm,
      minimum width=4mm,
      single arrow head extend=.3pt,
      rotate=-90}]

    \node[player] (P1) at (0,0) {$1$};
    \node[player, right=of P1] (P2) {$2$};
    \node[player, right=of P2] (P3) {$3$};

    \node[doublearrow] at ($(P1)!.5!(P2)$) {};

    \node[singlearrow] at ([yshift=-1cm]P1) {};
    \node[singlearrow] at ([yshift=1cm]P1) {};
    \node[singlearrow] at ([yshift=-1cm]P2) {};
    \node[singlearrow] at ([yshift=-1cm]P3) {};
    \node[singlearrow] at ([yshift=1cm]P3) {};

    \node at ([yshift=-1.7cm]P1) {$a_{1}$};
    \node at ([yshift=-1.7cm]P2) {$a_{2}$};
    \node at ([yshift=-1.7cm]P3) {$a_{3}$};

    \node at ([yshift=1.7cm]P1) {$s_{1}$};
    \node at ([yshift=1.7cm]P3) {$s_{3}$};

    \node at ([yshift=-2.3cm]P2) {$a_{1} = a_{2}$ and $a_{1} \oplus a_{3} = s_{1} \land s_{3}$.};
  \end{tikzpicture}
  \caption{The extended CHSH game.}\label{fig:ext_chsh}
\end{figure}
We denote the maximum achievable winning probabilities of classical, quantum, and forwarding strategies as $\omega_c, \omega_q, \omega_f$, respectively.
Then, it is clear that $\omega_q \geq \cos^2(\frac \pi  8)$ by simply letting $ v_1, v_3 $ implement an optimal quantum strategy for the CHSH game and then $v_1$ sends his output to $ v_2 $, who simply outputs exactly that. This actually achieves the quantum value by the same argument as for the distributed CHSH game. That is, 
$$\omega_q= \cos^2 (\frac \pi 8).$$

However, it is easy to see that $\omega_f < \cos^2(\frac \pi 8)$ again via self-testing. 
In fact, we can actually prove a stronger result: 
$$\omega_f = \frac 3 4,$$
the classical value of the CHSH game. To see this, let
\begin{align*}
    p(a_1, a_2, a_3 \vert s_1, s_3) = \langle\psi\vert\Pi_{1, a_1}(s_1) \otimes \Pi_{2, a_2}(s_1)\otimes \Pi_{3, a_3}(s_3)\vert \psi \rangle
\end{align*}
be a forwarding behavior. 
By~\Cref{prp:forwarding_is_q},
$$\omega_f(\mathcal G) = \omega_q(\mathcal G_1),$$
where $\mathcal G_1$ is a nonlocal game, namely the $1$-extension~\cite{toner2009monogamy} of the CHSH game. By Theorem 2.1 in~\cite{toner2009monogamy}, the classical and no-signaling values $\omega_c(\mathcal G_1), \omega_\mathrm{ns}(\mathcal G_1)$ are equal to the classical value of the CHSH game, $\frac 3 4$. Thus, $\omega_q(\mathcal G_1) = \frac 3 4$. We conclude $\omega_f(\mathcal G) = \frac 3 4$ as desired. We can summarize this result with a proposition. 
\begin{prp}
\label{prop:sep_forwarding_quantum}
    Let the LC game $\mathcal G$ be defined as above. Then, 
    $$\omega_q(\mathcal G) = \cos^2(\frac \pi 8) > \omega_f(\mathcal G) = \frac 3 4.$$
\end{prp}
\noindent In \Cref{app:extendedXOR} we observe similar separations for extended versions of other XOR games.

\subsection{Broadcasting strategies}
\label{sec:broadcasting}
A natural dual to forwarding strategies is a class of strategies in which parties perform quantum measurements prior to the transmission instead of after the transmission. This is indeed a quantum strategy, the $W_i(s_i)$ isometries being given by~\Cref{Eq:iso_measure_and_allocate}.
Like forwarding strategies, we make a new and simpler definition for this special class of quantum strategies:
\begin{dfn}
    Let $(\mathcal V, \pi, G)$ be an LC game. Let $A_i'$ be a finite set for $i \in [n]$ and define 
    \begin{align*}
        L_i \coloneqq \prod_{j \in N_\mathrm{in}[i]} A_j'.
    \end{align*}
    Let $B_i$ be finite-dimensional Hilbert spaces for $i \in [n]$.
    Then, a \textbf{broadcasting strategy} is a tuple 
    \begin{align*}
        (\vert \psi \rangle, \{M_i(s_i)\}_{i \in [n], s_i \in S_i}, \{f_i\}_{i \in [n]}),
    \end{align*}
    where $\vert \psi \rangle \in \bigotimes_{i=1}^n B_i$ is a quantum state, $M_i(s_i) \coloneqq \{\Pi_{a_i'}(s_i)\}_{a_i' \in A_i'}$ is a projective measurement on $B_i$ with outcomes in $A_i'$ for all $s_i \in S_i$, and 
    \begin{align*}
        f_i: L_i \to A_i
    \end{align*}
    is a function.
\end{dfn}
\noindent In words, each party receives his own input, performs a corresponding measurement on his share of the quantum state, shares the measurement outcome during communication, and produces his output based on the shared measurement outcomes. That is, the behavior of a broadcasting strategy is given by
\begin{align*}
    p(\mathbf a \vert \mathbf s) \coloneqq \sum_{a_i' \in A_i'} \langle \psi \vert \bigotimes_{i=1}^n \Pi_{a_i'}(s_i) \vert \psi \rangle \cdot \prod_{i=1}^n \delta_{f_i(l_i), a_i},
\end{align*}
where 
$$l_i \coloneqq \prod_{j \in N_\mathrm{in}[i]} a_j'.$$
If the connectivity graph $G$ is empty, a broadcasting strategy where $A_i' =A_i$ and $f_i$ is the identity function is just a quantum strategy for a conventional nonlocal game. In~\Cref{fig:measurement-to-broadcasting} we show schematically a broadcasting strategy for the connectivity graph 
$v_1 \leftrightarrows v_2 \quad v_3.$
    \begin{figure}[htbp!]
      \centering
      \begin{tikzpicture}[scale=.75]
        \pgfmathsetlengthmacro{\xstep}{3.8cm}
        \pgfmathsetlengthmacro{\ystep}{2.75cm}
        \pgfmathsetlengthmacro{\arrowx}{1.75*\xstep}
        \pgfmathsetlengthmacro{\rightshift}{3.0*\xstep}




        \begin{scope}[xshift=0]
          \pic (box31) at (0,2*\ystep) {tensorboxin={$M_{1}(s_{1})$}{$s_{1}$}};
          \pic (box32) at (0,1*\ystep) {tensorboxin={$M_{2}(s_{2})$}{$s_{2}$}};
          \pic (box33) at (0,0*\ystep) {tensorboxin={$M_{3}(s_{3})$}{$s_{3}$}};
          \pic (box41) at (\xstep,2*\ystep) {tensorboxout={$f_{1}$}{$a_{1}$}};
          \pic (box42) at (\xstep,1*\ystep) {tensorboxout={$f_{2}$}{$a_{2}$}};
          \pic (box43) at (\xstep,0*\ystep) {tensorboxout={$f_{3}$}{$a_{3}$}};

          \draw (box31-in1) -- ([xshift=-18pt]box31-inpin1);
          \draw (box32-in2) -- ([xshift=-18pt]box32-inpin2);
          \draw (box33-in3) -- ([xshift=-18pt]box33-inpin3);
          \draw[wire/classical] (box31-out1) -- (box41-in1);
          \draw[wire/classical] (box31-out2) to[out=0,in=180] (box42-in1);
          \draw[wire/classical] (box32-out1) to[out=0,in=180] (box41-in2);
          \draw[wire/classical] (box32-out2) -- (box42-in2);
          \draw[wire/classical] (box33-out3) -- (box43-in3);
          \draw[decorate,decoration={brace,amplitude=5pt,mirror}]
            ($(box31-inpin1)+(-.8,0)$) -- ($(box33-inpin3)+(-.8,0)$)
            node[midway,xshift=-15pt] {$\ket{\psi}$};
          \node at (2,-1.3) {\scriptsize $\textbf{a}'$};
        \end{scope}
      \end{tikzpicture}
      \caption{A broadcasting strategy for the connectivity graph $v_1 \leftrightarrows v_2 \quad v_3$. Each party makes a local input-dependent measurement $M_i(s_i)$ and sends its outcome to the receiver parties that implement post-processing functions $f_j$ on the combined received message. Double wires denote the classical messages $\textbf{a}' \coloneqq (a_1',a_2',a_3')$ that are transmitted. }
      \label{fig:measurement-to-broadcasting}
    \end{figure}

Now, we notice that the counterexample given in~\Cref{thm:q_neq_f} is actually a broadcasting strategy. We then ask, defining $\mathcal {BC}$ as the set of all broadcasting behaviors, is $\mathcal {BC} = \mathcal Q$? Using a proof similar to that of~\Cref{thm:aggregable}, we can actually prove the following.
\begin{thm}
\label{thm:LC_to_broadcasting}
    Let $(\mathcal V, \pi, G)$ be an LC game with $n$ parties. Let $\mathcal {BC}$ be the set of all broadcasting behaviors and $\mathcal Q$ the set of all quantum behaviors. Then,
    \begin{align*}
        \mathcal Q \subseteq \overline{\mathcal{BC}}.
    \end{align*}
\end{thm}
\begin{proof}
    The result in~\Cref{lem:measurement_channel_reduction} allows us to reduce quantum strategies to broadcasting strategies by replacing joint measurements on transmitted quantum systems by local measurements followed by classical processing of the measurement results. 
    Consider a quantum strategy 
    \begin{equation*}
        (\vert \psi \rangle, \{W_i(s_i)\}_{i \in [n], s_i \in S_i}, \{M_i\}_{i=1}^n)
    \end{equation*}
    for the LC game. 
    Let 
    \begin{equation*}
        \rho(\mathbf{s}) \coloneqq \left( \bigotimes_{i=1}^n W_i(s_i)\right)\ket{\psi}\bra{\psi} \left( \bigotimes_{i=1}^n W_i^\dagger(s_i)\right) 
    \end{equation*}
    denote the state obtained after the parties apply their input-dependent local isometries.

    The parties need to apply the measurement 
    \begin{equation*}
             \bigotimes_{j=1}^n  M_j
    \end{equation*} 
    to obtain their outputs. 
    Let $\mathcal{M}_j$ be the measurement channel associated with the joint measurement $M_j$ of party $j$ with input parties $i \in N_\mathrm{in}[j]$ with input system $B_{i \to j}$, and define the overall measurement channel by 
    \begin{equation*} 
        \mathcal{M} \coloneqq \bigotimes_{j=1}^n \mathcal{M}_j.
    \end{equation*}
    The realized quantum behavior $p_q(a\vert s)$ is given by the diagonal coefficients of $\mathcal M(\rho(\mathbf s))$.

    Let $\varepsilon > 0$. According to Lemma~\ref{lem:measurement_channel_reduction}, for every $j\in [n]$,
    the measurement channel $\mathcal{M}_j$ can be approximated by $\mathcal{E}_{M_j}$ of the form in~\Cref{eq:op_single_round_meas}, satisfying
    \begin{equation*}
        \| \mathcal{E}_{M_j} - \mathcal M_j \|_\diamond < \frac{\varepsilon}{n \prod_{i=1}^n \vert S_i \vert},
    \end{equation*}
    The channel $\mathcal{E}_{M_j}$ consists of POVMs
    $\{M'_{i\rightarrow j}\}_{i\in N_{\mathrm{in}}[j]}$
    performed by the parties $i\in N_{\mathrm{in}}[j]$,
    followed by classical communication of the corresponding
    measurement outcomes
    $\{x_{i\rightarrow j}\}_{i\in N_{\mathrm{in}}[j]}$
    to party $j$, who then applies a post-processing function
    $f_j$ to these outcomes to determine the final output.

    In this way, the overall measurement channel $\mathcal M$ is approximated by $\mathcal{E}_M \coloneqq \bigotimes_{j=1}^n \mathcal{E}_{M_j}$. 
    The quantum channel $\mathcal{E}_M$ can also be decomposed as 
    \begin{equation*}
        \mathcal{E}_M = \left(\bigotimes_{j=1}^n \mathcal F_j\right) \circ \left(\bigotimes_{i=1}^n \mathcal{M}_i'\right).
    \end{equation*}
    where $\mathcal{M}_i'$ is the measurement channel associated with local POVM
    \begin{equation*}
        M_i' \coloneqq \bigotimes_{j\in N_{\mathrm{out}}[i]} M_{i\rightarrow j}'
    \end{equation*}
    performed by party $i$, and $\mathcal{F}_j$ is the classical channel corresponding to the post-processing function $f_j$. 

    When we replace $\mathcal M$ in a quantum strategy with $\mathcal{E}_M$, 
    the result of this replacement is 
    a broadcasting strategy given by the tuple
    \begin{equation*}
        (|\psi\rangle, \{M_i(s_i)\}_{i\in [n], s_i \in S_i}, \{f_i\circ g_i\}_{i\in [n]} ),
    \end{equation*}
    where $M_i(s_i)$ is implemented by first applying the local isometry $W_i(s_i)$, followed by the measurement channel $\mathcal{M}_i'$,    
    and $g_i$ is a function that only keeps the measurement results $x_{j \rightarrow i}$ from $M_{j \to i}'$ for $j \in N_\mathrm{in}[i]$.
    This strategy realizes a broadcasting behavior $p_b$. 
    We have by the contractivity of the diamond norm under composition
    \begin{equation*}
        \| \mathcal{E}_M- \mathcal M \|_\diamond < n \frac{\varepsilon}{n\prod_{i=1}^n \vert S_i\vert} = \frac{\varepsilon}{\prod_{i=1}^n \vert S_i\vert}.
    \end{equation*}
    
    Consequently, for every input tuple $\mathbf{s}$, we have 
    \begin{equation*}
        \sum_{\mathbf{a}} |p_b(\mathbf{a}\vert \mathbf{s}) - p_q(\mathbf{a}\vert \mathbf{s}) | = \| \mathcal E_M (\rho(\mathbf{s}))-  \mathcal M (\rho(\mathbf{s})) \|_1 \le  \| \mathcal{E}_M- \mathcal M \|_\diamond <\frac{\varepsilon}{\prod_{i=1}^n \vert S_i\vert}.
    \end{equation*}
    Hence, the behaviors $p_q \in \mathcal Q$ and $p_b \in \mathcal {BC}$ induced by these two channels satisfy 
    \begin{equation*}
        \Vert p_b - p_q \Vert_1 = \sum_{\mathbf a,\mathbf s} |p_b(\mathbf{a}|\mathbf{s})- p_q(\mathbf{a}|\mathbf{s}) | < \prod_{i=1}^n \vert S_i\vert \frac{\varepsilon}{\prod_{i=1}^n \vert S_i\vert} = \varepsilon.
    \end{equation*}
    That is, the quantum behavior $p_q$ can be approximated arbitrarily well in 1-norm by broadcasting behaviors $p_b$ of this LC game. Consequently, 
    \begin{align*}
        \mathcal Q \subseteq \overline{\mathcal{BC}}.
    \end{align*}
\end{proof}

\noindent Therefore, it is sufficient to use a broadcasting strategy to realize any quantum behavior to arbitrary accuracy. The key ingredient is port-based teleportation, so again, an unbounded amount of entanglement may be necessary to achieve vanishing error. It may be better in practice to use a non-broadcasting strategy to reduce the amount of entanglement used to implement the strategy.

\section{Multi-Step Latency-Constrained (LC) Games}\label{sec:tau_lc}
A latency-constrained game as defined in~\Cref{dfn:lc} is the simplest extension of a nonlocal game that arises from latency constraints. It allows for only a single round of communication between a subset of the parties. To disambiguate, we will refer to this type of LC game as a \textbf{simple latency-constrained game}. Now, in real-world scenarios, the parties can have multiple rounds of communication, and may even receive inputs and produce outputs over different rounds. An example is in HFT where trading servers at various stock exchanges around the globe are monitoring local market data, communicating, and making trades over an extended period of time. We would therefore like to define a multi-round extension of LC games. Now, communication between different parties may incur different latencies, for example due to different geographical distances. Hence, we need to introduce \emph{a notion of time}. That is, we want a formulation of LC games involving multiple time steps.
We first define for a non-negative integer $\tau$,  $[0:\tau] \coloneqq \{0,1, \cdots, \tau\}$.
\begin{dfn}
\label{dfn:tau_lc}
    Let $n \geq 2$, $\tau \geq 0$ be integers. Let $S_i^{(t)}, A_i^{(t)}$ be finite sets for $i \in [n], t \in [0:\tau]$, and
    \begin{align*}
        S^{(t)} \coloneqq\prod_{i =1}^n S_i^{(t)}, A^{(t)} \coloneqq \prod_{i=1}^n A_i^{(t)}.
    \end{align*}
    We define a probabilistic predicate $\mathcal V$ which is a map
    \begin{align*}
        \mathcal V: \prod_{t = 0}^\tau A^{(t)} \times \prod_{t=0}^\tau  S^{(t)} \to [0,1].
    \end{align*}
    Furthermore, we define the input distribution $\pi$ which is a probability distribution on $\prod_{t=0}^\tau S^{(t)}$. 
    Finally, let 
    \begin{align*}
        \ell: [n]\times [n] \to \mathbb{Z}^+
    \end{align*}
    be a function such that 
    \begin{align*}
        \forall i \in [n], \, \ell(i,i)=1.
    \end{align*}
    This will be called the \textbf{latency function}.
    Then, a \textbf{$\tau$-step latency-constrained (LC) game} is the tuple $(\mathcal V, \pi, \ell)$. 
\end{dfn} 
\noindent We will refer to this class of games as~\textbf{multi-step latency-constrained (LC) games}. Here, $t \in [0:\tau]$ specifies the time step. Each party is to receive an input and produce an output at each time step from $t=0$ to $t =\tau$. $S_i^{(t)}$ is the input set for party $i$ at time step $t$ and $A_i^{(t)}$ is the corresponding output set. The probabilistic predicate evaluates the total ``score'' of the combined outputs of all the parties over all time steps given the combined inputs. $\pi$ describes the probability distribution of the different inputs. Finally, the latency function $\ell$ encodes the latency information. That is, $\ell(i,j)$ is the number of time steps required for party $i$ to transmit information to party $j$. We will assume this is always positive.\footnote{If $\ell(i,j) = \ell(j,i)=0$ then party $i$ and $j$ should be treated as the same party. If one direction is positive and the other zero, defining a strategy becomes messy. }
Moreover, we define for all $i \in [n]$
$$\ell(i,i) = 1$$
to more conveniently define strategies for $\tau$-step LC games below. For the sake of generality we make no additional assumptions on the latency function. For example, it could be asymmetric or violate the triangle inequality\footnote{However, we will add the triangle inequality to obtain some of the results below. } and therefore is not a distance function. \emph{In particular, when $\ell$ is the speed-of-light delay, then the corresponding realizable behaviors reflect fundamental physical limits}. Note here we are assuming that local operations consume zero time. The only operation that consumes time is communication between the parties, which is given by the latency function $\ell$. As mentioned in~\Cref{sec:intro}, this is justified by allowing the parties to use more energy in order to arbitrarily shorten their local operation times. 
Our assumption is in general a good approximation to real-world scenarios where latencies are much longer than local operation times, for example when the parties are physically distant from each other.  
We can define a \textbf{behavior} $p(\mathbf a\vert \mathbf s)$ and corresponding \textbf{winning probability} $p_\mathrm{win}$ for a $\tau$-step LC game analogously to the simple LC game case in~\Cref{sec:lc}.

In~\Cref{dfn:tau_lc} we seem to be assuming that time is discrete and that receiving inputs, producing outputs, and communication are done in a synchronous manner. However, the definition is more general than it may seem.
The time values we want to include in the definition of the LC game are the communication latencies and the times at which the parties need to receive inputs and produce outputs. To model any physical implementation with finite-precision clocks, we assume all of these time values are rational numbers. Let the set of all time values be
\begin{align*}
    V_T \subseteq \mathbb{Q}.
\end{align*}
This set is finite. We can then set one time step as a length of time 
\begin{align}
\label{eq:time_increment}
    t_0 \coloneqq \frac{1}{\mathrm{lcm}\{q: p/q \in V_T, \mathrm{gcd}(p,q)=1\}},
\end{align}
where $\mathrm{lcm},\mathrm{gcd}$ are short for least common multiple and greatest common divsor, respectively. It is clear that all time values in $V_T$ are integer multiples of $t_0$. For cases where parties are receiving inputs, producing outputs, and communicating asynchronously, we can simply let $S_i^{(t)}$ or $A_i^{(t)}$ be singletons when there is no input or output, respectively, at time step $t$ for party $i$. The parties can also choose to communicate nothing at a particular time step. This will depend on the specific strategy they choose to implement, which is what we will next discuss. Note that a $0$-step LC game is just a nonlocal game. A $1$-step LC game where $S_i^{(1)},A_i^{(0)}$ are singletons
is a simple LC game. 

Our $\tau$-step LC games framework can be used in real-world scenarios by specifying the points of time at which inputs are received and outputs are produced for a set of spatially separated parties with corresponding communication latencies. This determines the input sets $S_i^{(t)}$ and output sets $A_i^{(t)}$, as well as the latency function $\ell$. The probabilistic predicate $\mathcal V$ then specifies what is the desired time-series of outputs given the time-series of inputs, while $\pi$ is the prior on the inputs. By optimizing the winning probability using classical strategies versus quantum strategies, we can determine whether a quantum advantage is possible. This framework is a lot more natural than that of~\cite{ding2024coordinating} because real-world scenarios often involve a sequence of inputs and outputs and can involve communication (but with latency). Hence, we expect our framework to be useful for identifying real-world applications of quantum nonlocality.

We also mention that our $\tau$-step LC games framework goes beyond conventional nonlocal games \emph{even in the two-party setting}. This is because the two parties receive multiple inputs and produce multiple outputs, and they can also communicate, albeit with latency. For example, this can be used to analyze real-world scenarios where outputs in the past can affect the game being played in the present. That is, the game exhibits a ``memory effect.'' This may be applicable to scenarios such as~\cite{da2025entanglement}, in which the routing decisions for previous customer requests can affect the routing logic of later requests.  

For general $\tau$-step LC games we will mostly make definitions and prove some simple properties. Future work will be to establish a more comprehensive theory of $\tau$-step LC games. Such a theory could possibly make use of mathematical tools from works such as~\cite{gutoski2007toward}.

\subsection{Classical strategies} 

As we argued in~\Cref{subsec:classical_strategy}, we can always copy classical information, so in the deterministic setting the only nontrivial information that the parties can transmit to each other are their inputs. Hence, we can define a deterministic strategy for a $\tau$-step LC game as follows.
\begin{dfn}
    Let $(\mathcal V, \pi, \ell)$ be a $\tau$-step LC game. 
    Let 
    \begin{align*}
        H_i^{(t)} \coloneqq \prod_{t' \leq t} S_i^{(t')}\times\prod_{\substack{j \neq i\\ \ell(j,i) \leq  t}}\prod_{t' =0}^{t-\ell(j,i)} S_j^{(t')}
    \end{align*}
    be the \textbf{past light cone of party $i$ at time step $t$}.
    A \textbf{\boldmath deterministic strategy} is given by functions $\{f_i^{(t)}\}_{i \in [n], t\in [0:\tau]}$, where
    \begin{align*}
        f_i^{(t)}: H_i^{(t)} \to A_i^{(t)}.
    \end{align*}
\end{dfn}
\noindent That is, in a deterministic strategy, each party produces an output at time step $t$ given all his past inputs and the inputs of all parties that can be sent to him by time step $t$.

The behavior of a deterministic strategy is simply
\begin{align*}
    p(\mathbf a\vert \mathbf s)\coloneqq \prod_{t=0}^\tau \prod_{i=1}^n \delta_{a_i^{(t)}, f_i^{(t)}(h_i^{(t)})},
\end{align*}
where 
$h_i^{(t)}$ is the element of $H_i^{(t)}$ made up of corresponding elements in $\mathbf s$. 
Using the same proof as~\Cref{prp:num_det}, we can establish the following.
\begin{prp}
    Given a $\tau$-step LC game $(\mathcal V, \pi, \ell)$, the number of possible deterministic behaviors is 
    \begin{align}
    \label{eq:num_det_tau}
    \prod_{t=0}^\tau \prod_{i=1}^n \vert A_i^{(t)}\vert^{\vert H_i^{(t)}\vert }.
    \end{align}
\end{prp}
\noindent A \textbf{\boldmath classical strategy} is a probabilistic mixture of deterministic strategies where shared randomness can be used. Since the set of all classical behaviors is by definition the convex hull of the set of deterministic behaviors, it is a polytope with the number of vertices given in~\Cref{eq:num_det_tau}. Again, via the Minkowski-Weyl theorem, we can obtain the bounds on correlations between classical parties receiving multiple inputs and producing multiple outputs over time and each party having different communication latencies. These would be multi-round extensions of LC inequalities.

\subsection{Quantum strategies}
The core mathematical object describing a quantum strategy for a $\tau$-step LC game involves a set of  \textbf{interaction tensors} $W$. For a $\tau$-step LC game with $n$ parties, party $i$ at time step $t>0$ uses the interaction tensor
$W_i^{(t)}$, which has $(n+1)$ input legs and $(n+1)$ output legs. This is shown diagrammatically in~\Cref{fig:int_tensor}.
\begin{figure}[h!]
  \centering
  \begin{tikzpicture}[scale=.9]
    \pic (box) at (0,0) {tensorbox={}{}{}};
    \node[font=\scriptsize] at ($(box-in1)!.5!(box-out2)$) {$W^{(t)}_{i}$};
    \foreach \i in {1,2,3} {
      \draw (box-in\i) -- (box-inpin\i);
      \draw (box-out\i) -- (box-outpin\i);
    }
  \end{tikzpicture}
  \caption{The interaction tensor $W_i^{(t)}$ for party $i$ at time step
    $t$.}\label{fig:int_tensor}
\end{figure}
One input leg will take in the classical input from $S_i^{(t)}$, having dimension $\vert S_i^{(t)}\vert$. One output leg is a classical system with basis indexed by $A_i^{(t)}$. These legs correspond to classical systems. The other $n$ output legs are quantum systems sent to other parties (including party $i$ himself). That is, these connect with input legs of other interaction tensors in the future (larger $t$). Specifically, the $j$-th quantum output leg of $W_i^{(t)}$ will connect with the $i$-th quantum input leg of $W_j^{(t+ \ell(i,j))}$. 
Also, since we assumed $\ell(i,i)=1$ for all $i \in [n]$, the quantum system that party $i$ sends to himself can just be regarded as a local quantum register that is left behind and is sent to the interaction tensor for the next time step: $W_i^{(t+1)}$.

This specifies the connections of all the legs except for interaction tensors near the initial and final time steps. An ``early-time'' interaction tensor of party $i$ at time step $t$ may have an input leg corresponding to party $j$ that does not connect with an interaction tensor because $t- \ell(j,i) <0$. That is, no quantum system was transmitted to the interaction tensor via that leg because not enough time elapsed for a transmission from $j$ to $i$. We will assume these legs are trivial (dimension 1). 
Similarly, the ``late-time'' interaction tensors of party $i$ at time step $t$ may have an output leg corresponding to party $j$ that does not connect with an interaction tensor because $t+\ell(i,j) >\tau$. That is, no interaction tensor receives this transmission because the $\tau$-step LC game already ended by the time the transmission is completed. We will also assume these legs are trivial.
Lastly for $t=0$, $W_i^{(0)}$ has two input legs and $(n+1)$ output legs. The first input leg takes in the input from $S_i^{(0)}$ and the second input leg is connected to the $B_i$ subsystem of the initial shared quantum state $\vert \psi\rangle$. The output legs are connected in the same way as the other interaction tensors.

To keep everything physical, we require that for any interaction tensor $W_i^{(t)}$, for a fixed classical input leg index $s_i^{(t)}$, the map from the remaining (quantum) input legs to the remaining output legs (including the classical leg) is an isometry. 
With these interaction tensors and the initial shared quantum state, we can define a quantum strategy:

\begin{dfn}
    Let $(\mathcal V, \pi, \ell)$ be a $\tau$-step LC game. A \textbf{quantum strategy} is a tuple 
    $$(\vert \psi \rangle, \{W_i^{(t)}\}_{i \in [n], t\in [0:\tau]}).$$
\end{dfn}
\noindent Hence, the behavior of a quantum strategy is given by
\begin{align*}
    & p( \mathbf a \vert \mathbf s)
    =
    \Vert \vert \psi(\mathbf a \vert \mathbf s)\rangle \Vert^2
\end{align*}
where this time
\begin{align*}
    \mathbf{a} \coloneqq (a_i^{(t)})_{i \in [n], t \in [0:\tau]},\,
    \mathbf{s} \coloneqq (s_i^{(t)})_{i \in [n], t \in [0:\tau]}
\end{align*}
and
\begin{align}
\label{eq:tau_quantum}
    \vert \psi(\mathbf{a} \vert \mathbf{s})\rangle \coloneqq  \prod_{t=0}^{\tau} \left[\bigotimes_{i=1}^n \langle a_i^{(t)} \vert W_i^{(t)} \vert s_i^{(t)}\rangle \right]\vert \psi \rangle. 
\end{align}
The product in~\Cref{eq:tau_quantum} denotes the tensor-network contraction specified above: the $j$-th quantum output leg of $W_i^{(t)}$ is connected with the $i$-th quantum input leg of $W_j^{(t+\ell(i,j))}$.
The sandwiching of the $W_i^{(t)}$ interaction tensor in~\Cref{eq:tau_quantum} sets the classical input leg index to be $s_i^{(t)} \in S_i^{(t)}$ and the classical output leg index to be $a_i^{(t)} \in A_i^{(t)}$. Note that $\vert \psi(\mathbf{a} \vert \mathbf{s})\rangle$ is not a normalized vector.
A quantum strategy captures the most general quantum protocol that the parties can implement in $\tau$ time steps, where $\ell(i,j)$ is how long it takes for $i$ to signal to $j$. 
In~\Cref{fig:quantum_strategy}, we draw a diagram of a quantum strategy for three parties where $\ell(i,j) = \max\{1, \vert i-j\vert\}$. 
Note again that it is not strictly necessary for the parties to send quantum systems in a hardware implementation: they could simply perform quantum teleportation by using shared entanglement and communicating classical information. 
\begin{figure}[h]
  \centering

  \begin{tikzpicture}[scale=.9]
    \pgfmathsetlengthmacro{\xstep}{3.6cm}
    \pgfmathsetlengthmacro{\ystep}{2.7cm}

    \pic (box11) at (0,2*\ystep) {tensorbox={$W^{(0)}_{1}$}{$s^{(0)}_{1}$}{$a^{(0)}_{1}$}};
    \pic (box12) at (0,1*\ystep) {tensorbox={$W^{(0)}_{2}$}{$s^{(0)}_{2}$}{$a^{(0)}_{2}$}};
    \pic (box13) at (0,0*\ystep) {tensorbox={$W^{(0)}_{3}$}{$s^{(0)}_{3}$}{$a^{(0)}_{3}$}};

    \pic (box21) at (\xstep,2*\ystep) {tensorbox={$W^{(1)}_{1}$}{$s^{(1)}_{1}$}{$a^{(1)}_{1}$}};
    \pic (box22) at (\xstep,1*\ystep) {tensorbox={$W^{(1)}_{2}$}{$s^{(1)}_{2}$}{$a^{(1)}_{2}$}};
    \pic (box23) at (\xstep,0*\ystep) {tensorbox={$W^{(1)}_{3}$}{$s^{(1)}_{3}$}{$a^{(1)}_{3}$}};

    \pic (box31) at (2*\xstep,2*\ystep) {tensorbox={$W^{(2)}_{1}$}{$s^{(2)}_{1}$}{$a^{(2)}_{1}$}};
    \pic (box32) at (2*\xstep,1*\ystep) {tensorbox={$W^{(2)}_{2}$}{$s^{(2)}_{2}$}{$a^{(2)}_{2}$}};
    \pic (box33) at (2*\xstep,0*\ystep) {tensorbox={$W^{(2)}_{3}$}{$s^{(2)}_{3}$}{$a^{(2)}_{3}$}};

    \pic (box41) at (3*\xstep,2*\ystep) {tensorbox={$W^{(3)}_{1}$}{$s^{(3)}_{1}$}{$a^{(3)}_{1}$}};
    \pic (box42) at (3*\xstep,1*\ystep) {tensorbox={$W^{(3)}_{2}$}{$s^{(3)}_{2}$}{$a^{(3)}_{2}$}};
    \pic (box43) at (3*\xstep,0*\ystep) {tensorbox={$W^{(3)}_{3}$}{$s^{(3)}_{3}$}{$a^{(3)}_{3}$}};

    \pic (box51) at (4*\xstep,2*\ystep) {tensorbox={$W^{(\tau)}_{1}$}{$s^{(\tau)}_{1}$}{$a^{(\tau)}_{1}$}};
    \pic (box52) at (4*\xstep,1*\ystep) {tensorbox={$W^{(\tau)}_{2}$}{$s^{(\tau)}_{2}$}{$a^{(\tau)}_{2}$}};
    \pic (box53) at (4*\xstep,0*\ystep) {tensorbox={$W^{(\tau)}_{3}$}{$s^{(\tau)}_{3}$}{$a^{(\tau)}_{3}$}};

    \draw (box11-in1) -- ([xshift=-8pt]box11-inpin1);
    \draw (box12-in2) -- ([xshift=-8pt]box12-inpin2);
    \draw (box13-in3) -- ([xshift=-8pt]box13-inpin3);

    \draw (box11-out1) -- (box21-in1);
    \draw (box11-out2) to[out=0,in=180] (box22-in1);
    \wire{(box11-out3) to[out=0,in=180] (box33-in1)}{box22}

    \draw (box12-out1) to[out=0,in=180] (box21-in2);
    \draw (box12-out2) -- (box22-in2);
    \draw (box12-out3) to[out=0,in=180] (box23-in2);

    \wire{(box13-out1) to[out=0,in=180] (box31-in3)}{box22}
    \draw (box13-out2) to[out=0,in=180] (box22-in3);
    \draw (box13-out3) -- (box23-in3);

    \draw (box21-out1) -- (box31-in1);
    \draw (box21-out2) to[out=0,in=180] (box32-in1);
    \wire{(box21-out3) to[out=0,in=180] (box43-in1)}{box32}

    \draw (box22-out1) to[out=0,in=180] (box31-in2);
    \draw (box22-out2) -- (box32-in2);
    \draw (box22-out3) to[out=0,in=180] (box33-in2);

    \wire{(box23-out1) to[out=0,in=180] (box41-in3)}{box32}
    \draw (box23-out2) to[out=0,in=180] (box32-in3);
    \draw (box23-out3) -- (box33-in3);

    \draw (box31-out1) -- (box41-in1);
    \draw (box31-out2) to[out=0,in=180] (box42-in1);
    \clipwire{(box31-out3) to[out=0,in=180] (box53-in1)}{box42}{box31-out3}{box43-outpin3}

    \draw (box32-out1) to[out=0,in=180] (box41-in2);
    \draw (box32-out2) -- (box42-in2);
    \draw (box32-out3) to[out=0,in=180] (box43-in2);

    \draw (box33-out2) to[out=0,in=180] (box42-in3);
    \draw (box33-out3) -- (box43-in3);
    \clipwire{(box33-out1) to[out=0,in=180] (box51-in3)}{box42}{box31-out3}{box43-outpin3}

    \foreach \x in {0.2*\xstep,0.3*\xstep,.4*\xstep} {
      \fill ([xshift=\x]box41-outpin1) circle (.8pt);
      \fill ([xshift=\x]box42-outpin1) circle (.8pt);
      \fill ([xshift=\x]box43-outpin1) circle (.8pt);
    }

    \foreach \i in {1,2,3} {
      \foreach \j in {1,2,3} {
        \draw (box4\i-out\j) -- (box4\i-outpin\j);
        \draw (box5\i-in\j) -- (box5\i-inpin\j);
      }
    }

    \draw [decorate,decoration={brace,amplitude=5pt,mirror}]
    ($(box11-inpin1)+(-.4,0)$) -- ($(box13-inpin3)+(-.4,0)$)
    node[midway,xshift=-15pt] {$\ket{\psi}$};
  \end{tikzpicture}

  \caption{A quantum strategy for a three-party, $\tau$-step latency-constrained
    game, where the latency function is given by
    $\ell(i,j)= \max\{1,\vert i-j\vert\}$.
    The curved lines represent communicated information, with time flowing from left to right.
    For instance, the second quantum output leg takes one time step to reach
    party $v_{2}$, and the third quantum output leg takes two time steps to reach $v_{3}$.
    Trivial legs are omitted.}\label{fig:quantum_strategy}
\end{figure}

\indent Let $\mathcal{Q}$ be the set of all quantum behaviors in a $\tau$-step LC game $(\mathcal{V},\pi,\ell)$. Similar to~\Cref{prp_convex_LC}, we can prove the convexity of this set.
\begin{prp}
    For any $\tau$-step LC game $(\mathcal{V},\pi,\ell)$, the set $\mathcal{Q}$ is convex.
\end{prp}

\begin{proof}
    Consider a convex combination of two quantum behaviors in $\mathcal{Q}$:
    \begin{equation*}
    \begin{split}
        &\xi(\mathbf{a}\vert\mathbf{s})=\lambda p(\mathbf{a}\vert\mathbf{s})+(1-\lambda)p'(\mathbf{a}\vert\mathbf{s}),\\
    \end{split}
    \end{equation*}
    where
    \begin{equation*}
    \begin{split}
        \quad&p(\mathbf a \vert \mathbf s) = \Vert \vert \psi(\mathbf a \vert \mathbf s)\rangle \Vert^2,~~\vert \psi(\mathbf a \vert \mathbf s)\rangle =  \prod_{t=0}^{\tau} \left[\bigotimes_{i=1}^n \langle a_i^{(t)} \vert W_i^{(t)} \vert s_i^{(t)}\rangle \right]\vert \psi \rangle,\\
        &p'(\mathbf a \vert \mathbf s) = \Vert \vert \psi'(\mathbf a \vert \mathbf s)\rangle \Vert^2,~~\vert \psi'(\mathbf a \vert \mathbf s)\rangle =  \prod_{t=0}^{\tau} \left[\bigotimes_{i=1}^n \langle a_i^{(t)} \vert {W'_i}^{(t)} \vert s_i^{(t)}\rangle \right]\vert \psi' \rangle.
    \end{split}
\end{equation*}
By replacing the underlying quantum system with the direct sum of the original two systems, we can reconstruct $\xi(\mathbf{a}\vert\mathbf{s})$:
\begin{equation*}
    \begin{split}
        &\xi(\mathbf a \vert \mathbf s) = \Vert \vert \phi(\mathbf a \vert \mathbf s)\rangle \Vert^2,\\
    \end{split}
\end{equation*}
where 
\begin{equation*}
    \begin{split}
        \quad&\vert \phi(\mathbf a \vert \mathbf s)\rangle =  \prod_{t=0}^{\tau} \left[\bigotimes_{i=1}^n \langle a_i^{(t)} \vert \left( W_i^{(t)} \oplus {W'_i}^{(t)} \right)  \vert s_i^{(t)}\rangle \right] \left( \sqrt{\lambda}\vert \psi \rangle \oplus\sqrt{1-\lambda}\ket{\psi'} \right).
    \end{split}
\end{equation*}
The direct sums can be defined similarly to~\Cref{prp_convex_LC}. 

\end{proof}


\subsection{Single-input single-output (SISO) LC games}
\label{subsec:siso}
We would like to consider a special class of $\tau$-step LC games where there are multiple rounds of communication, but each party only has a single input and output:
\begin{dfn}
    A $\tau$-step LC game $(\mathcal V, \pi,\ell)$ is \textbf{single-input single-output (SISO)} if $S_i^{(t)}$ is a singleton for all $i \in [n]$ and $t > 0$, and $A_i^{(t)}$ is a singleton for all $i \in [n]$ and $t < \tau$.
\end{dfn}
\noindent That is, the parties only receive inputs in the very beginning at time step 0, and they only give outputs at the very end at time step $\tau$. 
In between they are able to transmit information to each other, but they do not receive inputs or give outputs. 
For conciseness, we will define
\begin{align*}
    S_i \coloneqq S_i^{(0)}, A_i \coloneqq A_i^{(\tau)}.
\end{align*}
Note that a $0$-step SISO game\footnote{For conciseness we will sometimes omit ``LC''. } is just a nonlocal game, and a $1$-step SISO game is just a simple LC game.
In~\Cref{fig:siso} we show a 2-step SISO game with 3 parties  where $\ell(1,2) = \ell(2,1) = 1$, $\ell(2,3) = \ell(3,2) = 2$, and $\ell(1,3), \ell(3,1) > 2$. 
\begin{figure}[htbp!]
  \centering
  \begin{tikzpicture}[scale=.9]
    \pgfmathsetlengthmacro{\xstep}{3.6cm}
    \pgfmathsetlengthmacro{\ystep}{2.7cm}

    \pic (box11) at (0,2*\ystep) {tensorboxin={$W_{1}^{(0)}$}{$s_{1}$}};
    \pic (box12) at (0,1*\ystep) {tensorboxin={$W_{2}^{(0)}$}{$s_{2}$}};
    \pic (box13) at (0,0*\ystep) {tensorboxin={$W_{3}^{(0)}$}{$s_{3}$}};

    \pic (box21) at (\xstep,2*\ystep) {tensorboxmid={$W_{1}^{(1)}$}};
    \pic (box22) at (\xstep,1*\ystep) {tensorboxmid={$W_{2}^{(1)}$}};
    \pic (box23) at (\xstep,0*\ystep) {tensorboxmid={$W_{3}^{(1)}$}};

    \pic (box31) at (2*\xstep,2*\ystep) {tensorboxout={$M_{1}$}{$a_{1}$}};
    \pic (box32) at (2*\xstep,1*\ystep) {tensorboxout={$M_{2}$}{$a_{2}$}};
    \pic (box33) at (2*\xstep,0*\ystep) {tensorboxout={$M_{3}$}{$a_{3}$}};

    \draw (box11-in1) -- ([xshift=-18pt]box11-inpin1);
    \draw (box12-in2) -- ([xshift=-18pt]box12-inpin2);
    \draw (box13-in2) -- ([xshift=-18pt]box13-inpin2);

    \draw (box11-out1) -- (box21-in1);
    \draw (box11-out2) to[out=0,in=180] (box22-in1);

    \draw (box12-out1) to[out=0,in=180] (box21-in2);
    \draw (box12-out2) -- (box22-in2);
    \draw (box12-out3) to[out=0,in=180] (box33-in0);

    \draw (box13-out2) -- (box23-in2);
    \draw (box13-out0) to[out=0,in=180] (box32-in3);

    \draw (box21-out1) -- (box31-in1);
    \draw (box21-out2) to[out=0,in=180] (box32-in1);

    \draw (box22-out1) to[out=0,in=180] (box31-in2);
    \draw (box22-out2) -- (box32-in2);
    \draw (box23-out2) -- (box33-in2);

    \draw [decorate,decoration={brace,amplitude=5pt,mirror}]
    ($(box11-inpin1)+(-.8,0)$) -- ($(box13-inpin2)+(-.8,0)$)
    node[midway,xshift=-15pt] {$\ket{\psi}$};
  \end{tikzpicture}
  \caption{A 2-step SISO game with $3$ parties where
    $\ell(1,2) = \ell(2,1) = 1$, $\ell(2,3) = \ell(3,2) =2$, and $\ell(1,3), \ell(3,1) >2$. 
    For clarity we denote the last interaction tensors as $M_i$ since they are measurement operations.
    }\label{fig:siso}
\end{figure}

A SISO game is naturally interpreted as multiple parties that receive inputs and need to produce their outputs within latency constraint $\tau$, given their pairwise communication latencies. If $\tau$ is too small, no communication is possible, and this becomes a conventional nonlocal game. On the other hand, if $\tau$ is sufficiently long for all the parties to receive all the inputs, then this becomes trivial. That is, we can prove:
\begin{prp}
    \label{prp:siso_trivial}
    Let $(\mathcal V,\pi,\ell)$ be a SISO, $\tau$-step LC game. Suppose $\tau \geq \ell_{\max}\coloneqq \max_{i,j\in[n]} \ell(i,j)$. Then, a classical strategy can realize all possible behaviors. In particular, it can achieve the maximum algebraic value $\omega_a$.
\end{prp}
\begin{proof}
    Since, $\tau \geq \ell_{\max}$, all parties have access to all inputs. The rest follows the proof of~\Cref{prp:trivial}.
\end{proof}
\noindent However, if $\tau$ is not sufficiently long, then only limited communication is possible. In fact, by computing the classical value $\omega_c(\tau)$ or quantum value $\omega_q(\tau)$ for the $\tau$-step SISO game $(\mathcal V, \pi, \ell)$ for different $\tau$, ceteris paribus, we are actually computing the highest possible winning probability achievable with classical or quantum resources \emph{as a function of time}, thereby answering the questions posed in~\Cref{sec:intro}. It is clear that these functions are monotonically non-decreasing.
We can then take the difference $\omega_q(\tau) -\omega_c(\tau)$ to evaluate the quantum advantage as a function of time.

We can also define a quantum advantage from the other direction. We can define $\tau_c(\alpha)$ for $\alpha \in [0,\omega_a]$ as the \textbf{$\alpha$-classical threshold time}:
the minimum number of time steps for a classical strategy to achieve a winning probability of $\alpha$ or higher. We can similarly define $\tau_q(\alpha)$ as the \textbf{$\alpha$-quantum threshold time}. If 
$$\tau_q(\alpha) < \tau_c(\alpha),$$
for some $\alpha$, then we have a \emph{mathematically provable time advantage using quantum resources} (as described in~\Cref{sec:intro}) for this SISO game and specific $\alpha$. Note that~\Cref{prp:siso_trivial} implies that for all $\alpha \in [0,\omega_a]$, 
\begin{align*}
    \tau_c(\alpha), \tau_q(\alpha)\leq \ell_{\max}.
\end{align*}

For a SISO game, we observe that both classical and quantum strategies lead to behaviors which satisfy a relaxation of the no-signaling condition. Namely, parties that are within each other's light cones can communicate, but the marginal behavior of parties outside of each other's light cones must be no-signaling with respect to each other. Just like simple LC games, information regarding which parties are in each other's light cones can be captured by a directed graph. The only condition we need to add is that the latency function satisfies the triangle inequality. Otherwise, going through an intermediate party can lower the total number of time steps taken to communicate between two parties.
\begin{dfn}
\label{dfn:siso_graph}
    Let $(\mathcal V, \pi, \ell)$ be a SISO, $\tau$-step LC game. 
    We say the SISO game is \textbf{path-consistent} if $\ell$ satisfies the triangle inequality:
    $$\forall i,j,k \in [n], \quad \ell(i,j) +\ell(j,k) \geq \ell (i,k).$$
    The \textbf{connectivity graph} $G =(V,E)$ is then defined as: 
     $$V \coloneqq [n]$$ 
     and
    \begin{align*}
        E \coloneqq \{(i,j) : i,j \in V, i\neq j, \ell(i,j) \leq \tau\}.
    \end{align*}
\end{dfn}
\noindent That is, the connectivity graph's vertices are the $n$ parties, and two vertices are connected in one direction if the first vertex can communicate to the second within $\tau$ time steps. Hence, the connectivity graph describes which parties can communicate with which other parties. 
We then define the generalization of no-signaling behaviors as follows: 
\begin{dfn}
Let $G = (V,E)$ be a directed graph.
We say a behavior $p(\mathbf a\vert \mathbf s)$ for a path-consistent SISO game is \textbf{\boldmath $G$-signaling} if for fixed $i \in [n]$, $s_1, s_2, \cdots, s_{i-1}, s_{i+1}, \cdots, s_n$, 
\begin{align*}
    \sum_{\substack{a_j \in A_j \\ j \in N_\mathrm{out}[i]}} p(a_1, a_2, \cdots, a_n \vert s_1, s_2, \cdots, s_{i-1}, s_i, s_{i+1},\cdots, s_n)
\end{align*}
is independent of $s_i \in S_i$. 
\end{dfn}
\noindent In words, this is saying that the marginal distribution of the outputs for parties that are not out-neighbors of some party is independent of the input of that party. It is clear that quantum and classical strategies for a path-consistent SISO game realize behaviors that are $G$-signaling, where $G$ is the connectivity graph. Note that $(V, \emptyset)$-signaling is equivalent to no-signaling. Furthermore, any marginal of a $G$-signaling behavior that only involves parties that are not adjacent (in either direction) of each other is no-signaling for every possible set of inputs of the parties that were traced out. That is, let $V' \subseteq V$ such that 
$$\forall i,j \in V', (i,j) \not \in E.$$
Then, for fixed $\{s_i\}_{v_i \not \in V'}$
\begin{align*}
    \sum_{\substack{a_i \in A_i \\v_i \not \in V'}} p(a_1, a_2,\cdots, a_n \vert s_1, s_2, \cdots, s_n)
\end{align*}
is a no-signaling behavior. 
Furthermore, we can optimize the winning probability over all $G$-signaling behaviors. This is a linear program as in the case of no-signaling behaviors~\cite{toner2009monogamy}. The corresponding maximum value is the $G$-signaling value:
\begin{align*}
    \omega_G \coloneqq \max_{G\text{-signaling behaviors}} p_\mathrm{win}.
\end{align*}

Now, we can prove an upper bound on the affine dimension of the set of $G$-signaling behaviors $\mathcal B_G$. Based on the notations introduced in \Cref{subsec:classical_strategy}, for any $X\subseteq V$, define
\begin{equation*}
    N_\text{in}[X]\coloneqq\bigcup_{i\in X}N_\text{in}[i].
\end{equation*}
\begin{prp}
    \begin{align*}
        \dim\mathcal{B}_G\leq\sum_{X\in\mathcal{P}^+(V)}\biggl(\prod_{w\in N_\mathrm{in}[X]}\vert S_w\vert\biggr)\biggl(\prod_{t\in X}(\vert A_t\vert-1)\biggr),
    \end{align*}
    where $\mathcal{P}^+(V)$ is the set of all nonempty subsets of $V$.
    \label{prp:dim}
\end{prp}
\begin{proof}
    For simplicity, in this proof, we assume $A_i=\{0,1,...,\vert A_i\vert-1\}$.
    The upper bound in this proposition can be proven by an explicit parametrization in \cite{lift_bell_inequalities} from an observation that all conditional probabilities involving the outcome 0 can be reconstructed using the normalization conditions and no-signaling conditions. For example, from the no-signaling condition 
    $$\sum_{a_j=0}^{\vert A_j\vert-1}p(a_i,a_j\vert s_i,s_j)=p(a_i\vert s_i),$$
    one may express 
    $$p(a_i,0\vert s_i,s_j)=p(a_i\vert s_i)-\sum_{a_j=1}^{\vert A_j\vert-1}p(a_i,a_j\vert s_i,s_j).$$
    In this parametrization scheme, for each nonempty subset $X\subseteq V$, there are $\prod_{w\in N_\text{in}[X]}\vert S_w\vert$ input possibilities and for each input scenario, there are $\prod_{t\in X}(\vert A_t\vert-1)$ specifiable output probabilities. The summation over all the subsets then gives an upper bound on the dimension.

\end{proof}

We next prove a surprising result. Although a SISO game appears to be more general than a simple LC game, using port-based teleportation~\cite{PhysRevLett.101.240501}, we can actually show that the behaviors achievable in a path-consistent SISO game are essentially the same as those of the simple LC game with the same connectivity graph. In other words, \emph{the intermediate interaction tensors $W_i^{(t)}$ for $0 < i < \tau$ can be removed}, with no effect on the possible classical and quantum behaviors up to closure. We first make a definition:
\begin{dfn}
    Let $\mathcal G = (\mathcal V,\pi,\ell)$ be a SISO, $\tau$-step LC game with $n$ parties that is path consistent. Let $G$ be its connectivity graph. The \textbf{induced simple LC game} $\mathcal G_\mathrm{sim}$ is a simple LC game with the tuple $(\mathcal V, \pi, G)$.
\end{dfn}

The classical statement is obvious:
\begin{prp}
\label{prp:siso_classical}
    Let $(\mathcal{V},\pi,\ell)$ be a SISO, $\tau$-step LC game with $n$ parties that is path-consistent.
    Let $\mathcal C$ be the set of classical behaviors and let $G$ be the connectivity graph. Let $\mathcal C_\mathrm{sim}$ be the set of classical behaviors for
    the induced simple LC game. Then,
    $$\mathcal C = \mathcal C_\mathrm{sim}.$$
\end{prp}
\begin{proof}
    The only nontrivial communication in a deterministic strategy is sending inputs. Hence in a SISO game, the intermediate interaction tensors are simply relaying inputs since any function the tensors implement can be implemented at the end by $W_i^{(\tau)}$. We can thus simply include the relayed information in the direct transmission from the initial sender to the final receiver. Since $\ell$ satisfies the triangle inequality, going through an intermediate party can only increase the number of time steps needed, and so direct transmission can only be faster. Hence, it is sufficient for the parties to share their inputs in accordance with the connectivity graph. That is, a deterministic strategy for the SISO game after replacing the intermediate interaction tensors with direct transmission of inputs becomes a deterministic strategy for the simple LC game. Adding shared randomness proves 
    $$\mathcal C \subseteq \mathcal C_\mathrm{sim}.$$
    The other direction is immediate. The conclusion follows.
\end{proof}

We now consider the quantum case. In SISO games, the intermediate interaction tensor $W_i^{(t)}$ can be regarded as a quantum channel
which serves as a ``middleman,'' responsible for relaying quantum information from the input parties to the output parties. 
Using port-based teleportation, we can effectively cut out the middleman and replace him with local quantum operations and direct communication from each of the input parties to each of the output parties. An example of this is shown in~\Cref{fig:interaction-tensor-elimination}. 
\begin{figure*}[htbp!]
  \centering
  \begin{tikzpicture}[scale=.9]
    \pgfmathsetlengthmacro{\xstep}{2.4cm}
    \pgfmathsetlengthmacro{\ystep}{2.4cm}

    \pgfmathsetlengthmacro{\arrowx}{2.8*\xstep}
    \pgfmathsetlengthmacro{\rightshift}{3.8*\xstep}


    \pic (box11) at (0,1*\ystep) {tensorboxin={$W_{1}^{(0)}$}{$s_{1}$}};
    \pic (box12) at (0,0*\ystep) {tensorboxin={$W_{2}^{(0)}$}{$s_{2}$}};

    \pic (box21) at (\xstep,0*\ystep) {tensorboxmid={$W_2^{(1)}$}};

    \pic (box31) at (2*\xstep,1*\ystep)  {tensorboxout={$M_{1}$}{$a_{1}$}};
    \pic (box32) at (2*\xstep,0*\ystep)  {tensorboxout={$M_{2}$}{$a_{2}$}};
    \pic (box33) at (2*\xstep,-1*\ystep) {tensorboxout={$M_{3}$}{$a_{3}$}};

    \draw (box11-in1) -- ([xshift=-18pt]box11-inpin1);
    \draw (box12-in2) -- ([xshift=-18pt]box12-inpin2);

    \draw (box11-out2) to[out=0,in=180] (box21-in1);
    \draw (box12-out2) -- (box21-in2);

    \draw (box21-out1) to[out=0,in=180] (box31-in2);
    \draw (box21-out2) -- (box32-in2);
    \draw (box21-out3) to[out=0,in=180] (box33-in2);



    \node at (\arrowx,0) {\Large $\Rightarrow$};


    \begin{scope}[xshift=\rightshift]

      \pic (box41) at (0,1*\ystep) {tensorboxin={$W_{1}'$}{$s_{1}$}};
      \pic (box42) at (0,0*\ystep) {tensorboxin={$W_{2}'$}{$s_{2}$}};

      \pic (box51) at (1.6*\xstep,1*\ystep)  {tensorboxout={$M_{1}'$}{$a_{1}$}};
      \pic (box52) at (1.6*\xstep,0*\ystep)  {tensorboxout={$M_{2}'$}{$a_{2}$}};
      \pic (box53) at (1.6*\xstep,-1*\ystep) {tensorboxout={$M_{3}'$}{$a_{3}$}};

      \draw (box41-in1) -- ([xshift=-18pt]box41-inpin1);
      \draw (box42-in2) -- ([xshift=-18pt]box42-inpin2);

      \draw (box41-out1) -- (box51-in1);
      \draw (box41-out2) to[out=0,in=180] (box52-in1);
      \draw (box41-out3) to[out=0,in=180] (box53-in1);

      \draw (box42-out1) to[out=0,in=180] (box51-in2);
      \draw (box42-out2) -- (box52-in2);
      \draw (box42-out3) to[out=0,in=180] (box53-in2);


    \end{scope}

  \end{tikzpicture}
  \caption{Cutting out the middleman $W_2^{(1)}$ in a three-party scenario. The interaction tensor is replaced by direct communication, along with additional local operations. }
  \label{fig:interaction-tensor-elimination}
\end{figure*}
We can then apply this iteratively to obtain the desired result.
\begin{thm}
\label{thm:siso_to_LC}
    Let $\tau \geq 2$ and let $(\mathcal{V},\pi,\ell)$ be a SISO, $\tau$-step LC game with $n$ parties that is path-consistent.
    Let $G$ be the connectivity graph. 
    Let $\mathcal Q$ be the set of quantum behaviors of the SISO game and $\mathcal Q_\mathrm{sim}$ that of the induced simple LC game. Then,
    \begin{align*}
        \mathcal Q_\mathrm{sim} \subseteq \mathcal Q \subseteq \overline{\mathcal Q_\mathrm{sim}}.
    \end{align*}
\end{thm}
\begin{proof}
    The detailed proof is given in Appendix~\ref{app:siso_to_LC_proof}.
\end{proof}
Together,~\Cref{prp:siso_classical} and~\Cref{thm:siso_to_LC} constitute a quite astounding result. They imply that the correlations achievable in an arbitrary SISO game, where there are multiple rounds of communication between different pairs of parties in a complicated network structure such as in~\Cref{fig:siso}, can also be achieved by a simple LC game, where there is only one round of communication. The only latency-associated information that is retained in this reduction is the connectivity graph $G$.\footnote{Note that this reduction is up to set closure, but this mathematical subtlety does not have any physical consequences since infinitesimal differences cannot be measured in experiments. }
\emph{This reveals an interesting fact about the nature of quantum mechanics that is not obvious a priori: entanglement, local operations, and a single round of communication are the only physical resources necessary to achieve a quantum correlation}.

\subsection{Auxiliary parties}
In our definition of quantum strategies in~\Cref{dfn:graph_quantum}, we assumed that the only physical resources available to the parties are local quantum registers that can perform gates and transmit quantum systems. Conceivably, the parties could also use \emph{remote quantum registers}, in particular ones located between a pair of parties. Such remote quantum registers can be interpreted as \textbf{auxiliary parties} whose input and output sets are singletons and whose sole purpose is to help the non-auxiliary parties win the game. As long as the latency constraint is still satisfied, this is an allowed physical resource.

For example, consider three collinear parties where $v_1$ and $v_2$ are closer together and the latency constraint is such that they can communicate, but $v_3$ is isolated. 
We model this as three-party SISO, $\tau$-step game where
\begin{align*}
    \ell(1,2) = 2, \, \ell(2,3) = 3, \, \ell(1,3) = \ell(1,2)+\ell(2,3)=5
\end{align*}
and $\tau=2$. We assume $\ell$ is a symmetric function. 
This is visualized in~\Cref{fig:auxiliary_before}. We choose these numbers so that an auxiliary party can be accommodated.
Then, we can add an auxiliary party $v_A$ located at the midpoint between $v_1$ and $v_2$. This is visualized in~\Cref{fig:auxiliary_after}.
\begin{figure}[htbp!]
    \centering
    \begin{subfigure}{0.8\textwidth}
    \centering
    \begin{tikzpicture}
    \node (A) at (0,0) {$v_1$};
    \node (B) at (2,0) {$v_2$};
    \node (D) at (5,0) {$v_3$};
    
    \draw[dashed] (A) -- node[midway, above] {$2$} (B);
    \draw[dashed] (B) -- node[midway, above] {$3$} (D);
    \end{tikzpicture}
    \caption{}
    \label{fig:auxiliary_before}
    \end{subfigure}
    \hfill
    \begin{subfigure}{0.8\textwidth}
    \centering
    \begin{tikzpicture}
    \node (A) at (0,0) {$v_1$};
    \node (B) at (2,0) {$v_2$};
    \node (C) at (1,0) {$v_A$};
    \node (D) at (5,0) {$v_3$};
    
    \draw[dashed] (A) -- node[midway, above] {$1$} (C);
    \draw[dashed] (C) -- node[midway, above] {$1$} (B);
    \draw[dashed] (B) -- node[midway, above] {$3$} (D);
    \end{tikzpicture}
    \caption{}
    \label{fig:auxiliary_after}
    \end{subfigure}
    
    \caption{(a) Three parties in a line where $v_1$ and $v_2$ are closer together and $v_3$ is farther away. (b) Three parties in a line with an additional auxiliary party $v_A$ in between $v_1$ and $v_2$. }
\end{figure}
    
\noindent This now becomes a four-party SISO game where $\tau=2$ and
$$\ell(1,A) = \ell(A,2) = 1, \ell(A,3) = 4.$$
In particular, $v_A$ can serve as a remote quantum register in between $v_1$ and $v_2$ that can facilitate their mutual transmissions. This is illustrated in~\Cref{fig:auxiliary_strategy}. 
\begin{figure}[htbp!]
  \centering
  \begin{tikzpicture}[scale=.9]
    \pgfmathsetlengthmacro{\xstep}{2.0cm}
    \pgfmathsetlengthmacro{\ystep}{2.7cm}


    \pic (box11) at (0,2*\ystep) {tensorboxin={$W_{1}^{(0)}$}{$s_{1}$}};
    \pic (box12) at (0,1*\ystep) {tensorboxin={$W_{2}^{(0)}$}{$s_{2}$}};
    \pic (box13) at (0,0*\ystep) {tensorboxin={$W_{3}^{(0)}$}{$s_{3}$}};

    \pic (box21) at (\xstep,1.5*\ystep) {tensorboxmid={$W_A$}};

    \pic (box31) at (2*\xstep,2*\ystep) {tensorboxout={$M_{1}$}{$a_{1}$}};
    \pic (box32) at (2*\xstep,1*\ystep) {tensorboxout={$M_{2}$}{$a_{2}$}};
    \pic (box33) at (2*\xstep,0*\ystep) {tensorboxout={$M_{3}$}{$a_{3}$}};


    \draw (box11-in1) -- ([xshift=-18pt]box11-inpin1);
    \draw (box12-in2) -- ([xshift=-18pt]box12-inpin2);
    \draw (box13-in2) -- ([xshift=-18pt]box13-inpin2);

    \coordinate (box21-inA) at ($(box21-in1)!0.5!(box21-in2)$);
    \coordinate (box21-inpinA) at ($(box21-inA)+(-2.2,0)$);
    \draw (box21-inA) -- ([xshift=-18pt]box21-inpinA);


    \draw (box11-out2) to[out=0,in=180] (box21-in1);
    \draw (box12-out1) to[out=0,in=180] (box21-in2);

    \draw (box21-out1) to[out=0,in=180] (box31-in2);
    \draw (box21-out2) to[out=0,in=180] (box32-in1);


    \draw (box13-out2) -- (box33-in2);


    \draw [decorate,decoration={brace,amplitude=5pt,mirror}]
    ($(box11-inpin1)+(-.8,0)$) -- ($(box13-inpin2)+(-.8,0)$)
    node[midway,xshift=-15pt] {$\ket{\psi}$};

  \end{tikzpicture}
  \caption{A quantum strategy with three parties plus one auxiliary party in between $v_1$ and $v_2$. Some elements of the quantum strategy are not drawn in order to simplify the figure. These elements can be omitted without loss of generality. }
  \label{fig:auxiliary_strategy}
\end{figure}

Such an auxiliary party can greatly simplify quantum strategies. For example, suppose we are playing the distributed CHSH game. We found in~\cref{subsubsec:dist_chsh} an optimal quantum strategy that uses a three-qubit entangled state between the three parties. However, if we add an auxiliary party $v_A$ in between $v_1$ and $v_2$, then a trivial quantum strategy is available. Namely, $v_1$ and $v_2$ send their inputs to $v_A$, who shares a Bell pair with $v_3$. Then, $v_A$ and $v_3$ use an optimal quantum strategy for the usual CHSH game, after which $v_A$ sends the measurement results back to $v_1, v_2$. This clearly achieves the quantum value $\cos^2 \frac \pi 8$. 

More generally, for a three-party LC game with the connectivity graph
$$ v_1 \leftrightarrows v_2 \quad v_3,$$
we can use port-based teleportation to aggregate $v_1$ and $v_2$ as per~\Cref{thm:aggregable} and then use an optimal quantum strategy for the two-party nonlocal game thus attained to achieve the quantum value of the post-aggregation two-party game. However, this may require a large amount of entanglement, as we saw in~\Cref{lem:measurement_channel_reduction}. Using an auxiliary party as above would obviate this entanglement overhead. Indeed, in this case there is no overhead at all since $v_A$ and $v_3$ can directly use the quantum strategy for the two-party nonlocal game.

An important question is whether the quantum value of the game can increase if we add auxiliary parties. This would imply that auxiliary parties are necessary physical resources in order to achieve the highest physically possible winning probabilities. Interestingly, we can prove that this is not true, by directly using the result of~\Cref{thm:siso_to_LC} that reduces a SISO game to its induced simple LC game. That is, the auxiliary parties' operations can all be replaced by direct communication between the non-auxiliary parties.
To state our result, we define what it means to have auxiliary parties. 
\begin{dfn}
    Let $(\mathcal V, \pi, \ell)$ be a SISO, $\tau$-step LC game with $n$ parties that is path-consistent.
    Let $n_A \in \mathbb Z^+$. An \textbf{$n_A$-auxiliary extension} $(\mathcal V_A, \pi_A, \ell_A)$ is a SISO, $\tau$-step LC game with $n+n_A$ parties that is path-consistent and also has the following properties. The input and output sets for the first $n$ parties are the same as those of the original game, while the input and output sets of the additional $n_A$ parties are singletons. $\mathcal V_A, \pi_A$ are the extensions of $\mathcal V$, $\pi$ to include the auxiliary parties (the additional input and output sets are singletons). Lastly,     $$ \ell_A(i,j) = \ell(i,j)$$
    for all $i,j \leq n$.
\end{dfn}
\noindent The latency function needs to satisfy the triangle inequality because otherwise information could travel from one party to another faster via an auxiliary party, which would contradict our claim.

We will prove the following. As the operations of auxiliary parties can be described by quantum channels, this is a direct corollary of~\Cref{thm:siso_to_LC}.
\begin{cor}
    Let $\mathcal Q$ be the set of quantum behaviors of a SISO game $\mathcal G$ and $\mathcal Q_\mathrm{aux}$ be that of an auxiliary extension $\mathcal G_\mathrm{aux}$. Then, letting $\mathcal Q_A$ be the extension of the behaviors in $\mathcal Q$ to include the auxiliary parties,
    \begin{align*}
        \mathcal Q_A \subseteq \mathcal Q_\mathrm{aux} \subseteq \overline{\mathcal Q_A}.
    \end{align*}
    In particular, the quantum values are the same.
\end{cor}
\begin{proof}
    The first containment is trivial.
    
We apply~\Cref{thm:siso_to_LC} to $\mathcal G_\mathrm{aux}$ to obtain
$$ \mathcal Q_\mathrm{aux} \subseteq \overline{(\mathcal Q_\mathrm{aux})_\mathrm{sim}},$$
where $(\mathcal Q_\mathrm{aux})_\mathrm{sim}$ is the set of quantum behaviors for the simple LC game induced by $\mathcal G_\mathrm{aux}$.

For a simple LC game, there is only a single round of communication. Now, the auxiliary parties' input and output sets are singletons. Furthermore, any quantum system received by an auxiliary party can be traced out, while any quantum system sent from an auxiliary party to a non-auxiliary party can be included into the initial quantum state shared between the non-auxiliary parties. Therefore, the presence of auxiliary parties does not affect the set of achievable quantum behaviors of the induced simple LC game. That is, 
$$(\mathcal Q_\mathrm{aux})_\mathrm{sim} = (\mathcal Q_{\mathrm{sim}})_A,$$
where $\mathcal (Q_{\mathrm{sim}})_A$ is the set of quantum behaviors for the simple LC game induced by $\mathcal G$ extended to include the auxiliary parties.
Hence,
$$\mathcal Q_\mathrm{aux} \subseteq \overline{(\mathcal Q_{\mathrm{sim}})_A}.$$
On the other hand, we can apply~\Cref{thm:siso_to_LC} to $\mathcal G$ to obtain
$$ \mathcal Q_\mathrm{sim} \subseteq\mathcal Q \subseteq \overline{\mathcal Q_\mathrm{sim}},$$
and so
$$ (\mathcal Q_{\mathrm{sim}})_A \subseteq\mathcal Q_A \subseteq \overline{(\mathcal Q_{\mathrm{sim}})_A},$$
This implies that 
$$\overline{\mathcal Q_A} = \overline{(\mathcal Q_\mathrm{sim})_A}.$$
We conclude
$$\mathcal Q_\mathrm{aux} \subseteq \overline{\mathcal Q_A}.$$

The equality of the quantum values follows directly from the set containments and the continuity of the winning probability as a function of the behavior with respect to the 1-norm.
\end{proof}
\noindent\emph{This again reveals an interesting fact about quantum mechanics: to realize the set of physically achievable correlations between spatially separated parties, we do not need remote registers located between the parties. }
The cost, again, is possibly an enormous amount of quantum entanglement. We note that from a hardware implementation standpoint, there is also a equipment cost to installing these remote quantum registers, and so the relative practicality of using remote registers would have to be carefully evaluated.

\subsection{Continuous-time LC games}
\label{subsec:cont_time}
Although our definition of $\tau$-step LC games can approximate physical latency-constrained scenarios arbitrarily well, it is still worthwhile to consider how an LC game should be formulated in the continuous-time setting where the parties are continuously receiving inputs and producing outputs. This would be more natural from a physics perspective. In the same spirit, we will also assume inputs and outputs are continuous-valued, in particular real-valued~\cite{aharon2013continuous}. We will describe in words how continuous-time LC games can be formulated. We leave a more formal mathematical formulation and subsequent analyses to future work.

For simplicity, we consider the two-party case as drawn in~\Cref{fig:cont_time}.
\begin{figure}
  \centering
  \begin{tikzpicture}[node distance=1.7cm,
    player/.style = {draw, circle, minimum width=.9cm, fill=gray!15},
    doublearrow/.style = {double arrow, draw, inner sep=0pt,
      minimum height=1.2cm,
      minimum width=4mm,
      double arrow head extend=.3pt,
    },
    singlearrow/.style = { single arrow, draw, inner sep=0pt,
      minimum height=.7cm,
      minimum width=4mm,
      single arrow head extend=.3pt,
      rotate=-90}]

    \node[player] (P1) at (0,0) {$1$};
    \node[player, right=of P1] (P2) {$2$};

    \node[doublearrow] (DA) at ($(P1)!.5!(P2)$) {};

    \node[singlearrow] at ([yshift=-1cm]P1) {};
    \node[singlearrow] at ([yshift=1cm]P1) {};
    \node[singlearrow] at ([yshift=-1cm]P2) {};
    \node[singlearrow] at ([yshift=1cm]P2) {};

    \node at ([yshift=-1.7cm]P1) {$a_1(t)$};
    \node at ([yshift=-1.7cm]P2) {$a_2(t)$};

    \node at ([yshift=1.7cm]P1) {$s_1(t)$};
    \node at ([yshift=1.7cm]P2) {$s_2(t)$};

    \node at ([yshift=0.7cm]DA) {$\ell$};

  \end{tikzpicture}
  \caption{A continuous-time LC game with two parties.}\label{fig:cont_time}
\end{figure}
Let $\tau \geq 0$ be a real number and let $\mathcal S_1(t), \mathcal S_2(t)$, where $t \in [0, \tau]$, be continuous-time stochastic processes that take values in $\mathbb R$. We interpret these as inputs to $v_1$ and $v_2$, respectively. Likewise, we let $\mathcal A_1(t), \mathcal A_2(t)$ be continuous-time stochastic processes that take values in $\mathbb R$. These will be the outputs of $v_1$ and $v_2$, respectively. Let $s_1(t), s_2(t)$ and $a_1(t), a_2(t)$ be real-valued functions on $[0,\tau]$.
We are interested in the \textbf{continuous-time behavior}
\begin{align*}
    p(\mathcal A_1(t)=a_1(t), \mathcal A_2(t) = a_2(t) \text{ for } t \leq \tau \vert \mathcal S_1(t)=s_1(t), \mathcal S_2(t) = s_2(t) \text{ for } t \leq \tau).
\end{align*}
That is, we are interested in the probability of a \emph{trajectory} of the outputs $a_1(t), a_2(t)$ for $t \in [0,\tau]$ given the trajectory of the inputs $s_1(t),s_2(t)$. Thus, unlike $\tau$-step LC games where the behavior $p(\textbf a\vert \textbf s)$ is a function whose domain is a Cartesian product of discrete sets $S_i, A_i$, in the continuous-time setting, the behavior is a function whose domain is a Cartesian product of the space of functions 
$$\{f: [0,\tau]\to \R\}.$$ 
Then, we can define the latency function $\ell$ in the continuous-time setting as real-valued:
$$\ell(v_1,v_2), \ell(v_2,v_1) \in \R^+,$$
where $\ell(v_1,v_2)$ is the latency from $v_1$ to $v_2$ and vice versa. 

We can strategies as involving continuous-time classical or quantum operations. Again we will assume local operations consume zero time. This is somewhat extreme in the continuous-time setting but as we discussed in~\Cref{sec:intro}, this can be justified by simply increasing the energy. In the end, we are concerned with fundamental physical limits: 
\begin{center}
    ``Given all possible classical/quantum resources (including arbitrarily large energies), what correlations can be exhibited?''
\end{center}

A classical strategy can include actions such as $v_1$ at time $t$ sending classical information to $v_2$, for example his input $s_1(t)$ or his output $a_1(t)$. $v_2$ will receive the information at time $t+ \ell(v_1,v_2)$. $v_2$ can apply a function on this information and other information that he has to produce his output $a_2(t+\ell(v_1,v_2))$. Similarly, quantum strategies can include actions such as $v_2$ locally preparing and sending a quantum system $\vert \psi\rangle_{B_2}$ at time $t$ to $v_1$. $v_1$ would receive this quantum system at time $t+\ell(v_2,v_1)$. He can locally perform a measurement $\{\Pi_{1, a_1}\}_{a_1}$ on $\vert \psi\rangle_{B_2}$ together with some quantum system on his local quantum register to obtain his output $a_1(t+\ell(v_2,v_1))$. In essence, these strategies are the continuous-time versions of classical and quantum strategies for $\tau$-step LC games.

The above two questions about fundamental physical limits then translate to the mathematical questions:
\begin{center}
    ``What is the set of continuous-time behaviors realized by classical/quantum strategies?''
\end{center}
We leave the rigorous formulation and answering of this question to future work.

\section{Numerical Techniques and Examples}
\label{sec:numerical}
We want to compute by how much quantum resources, which are the most general physical resources, can violate the bounds on classical correlations in the latency-constrained games setting. That is, we want to numerically compute the quantum value $\omega_q$ for arbitrary simple LC games and compare it to the classical value $\omega_c$. Via~\Cref{prp:siso_classical} and~\Cref{thm:siso_to_LC}, these numerical methods immediately apply to SISO games as well.

In general, optimizing conventional nonlocal games is already a computationally hard problem~\cite{kempe2011entangled,ji16classical,ji2020mip}. Hence, we can only hope to provide some heuristic algorithms. Standard methods for nonlocal games, like the see-saw algorithm~\cite{werner2001bell} or the Navascu\'es-Pironio-Ac\'in (NPA) algorithm~\cite{navascues2007bounding,navascues2008convergent}, may also apply directly to LC games with special connectivity graphs $G$. For example, this is possible when $G$ is the empty graph or we can aggregate parties in a simple LC game such that it reduces to a conventional nonlocal game. In general, however, such methods cannot be used directly on games that cannot be reduced to conventional nonlocal games.

For this reason, we provide in \Cref{subsec:seesaw} a generalization of the see-saw algorithm to compute lower bounds on the quantum values of arbitrary simple LC games. Upper bounds on the quantum value of LC games could possibly be obtained via a generalization of the NPA hierarchy along the lines of~\cite{klep2024state}. We can ultimately use the tools we develop to analyze the achievable correlations between different parties with a fixed spatial layout as a function of the latency constraint in~\Cref{subsec:pert_xor}.

\subsection{Lower bounds via a see-saw algorithm}
\label{subsec:seesaw}

We can compute lower bounds on the quantum value $\omega_q$ of a simple LC game via a generalization of the see-saw algorithm~\cite{werner2001bell} used to compute lower bounds for conventional nonlocal games.
We first rewrite the behavior~\Cref{eq:g_network} in the mixed quantum formalism: 
\begin{align*}
    p(\mathbf a \vert \mathbf s) = \tr \left[ \bigotimes_{i=1}^n \Pi_{i, a_i} \left(\bigotimes_{i=1}^n \mathcal W_i(s_i)\right) (\rho) \right],
\end{align*}
where, writing $\mathcal L(\mathcal H)$ the space of linear operators on a Hilbert space $\mathcal H$, 
$$\rho \in \mathcal L \Bigl(\bigotimes_{i=1}^n B_i\Bigr)$$
is a quantum state and where, for each party $i$,
$$\mathcal W_i(s_i): \mathcal L(B_i) \to \mathcal L( B^\text{out}_i)$$
is a quantum channel parametrized by $s_i$, and $\Pi_i = \{\Pi_{i, a_i}\}_{a_i \in A_i}$ is a projective measurement with projectors $\Pi_{i, a_i} \in \mathcal L \left(B^\text{in}_i\right)$.

Together, a parametrized quantum channel and a projective measurement define the ``local'' strategy of a party $i$. A key observation is that such a strategy is formally equivalent to a one-slot quantum comb~\cite{chiribella2008quantum, chiribella2099theoretical, gutoski2007toward}, which is a higher-order quantum operation~\cite{taranto2025higherorderquantumoperations} corresponding to the combination of two quantum operations connected by a memory link, which in this case is the space $B_{i \to i}$.

To represent the strategy of a party as a single higher-order operation, we use the ``link product''~\cite{chiribella2008quantum, chiribella2099theoretical} denoted `$*$', which represents the composition of operators at the level of their Choi matrices~\cite{Choi75linear, Jamiolkowski72Linear}. It is defined for any matrices $X_{CD} \in \mathcal L(C \otimes D)$ and $Y_{DE} \in \mathcal L(D\otimes E)$ by 
\begin{equation*} 
X_{CD} * Y_{DE} \coloneqq \tr_D \left[(X_{CD} \otimes I_E)^{\top[D]} \cdot (I_C \otimes Y_{DE}) \right] \in \mathcal L(C\otimes E), 
\end{equation*}
where $\cdot ^{\top[D]}$ and $\tr_D$ are respectively the partial transpose and the partial trace on the space $D$, and where $I_E$ is the 
operator on space $E$. We start by writing $\mathsf{W}_i(s_i)$ as the non-normalized Choi operator of $\mathcal{W}_i(s_i)$ defined as 
\begin{align*}
    \mathsf{W}_i(s_i) \coloneqq \sum_{j,k=1}^{\dim{B_i}} \vert j \rangle\langle k \vert  \otimes \mathcal{W}_i(s_i)( \vert j \rangle \langle k \vert) \in \mathcal L\left(B_i \otimes B_i^\text{out}\right).
\end{align*}
The local strategy of party $i$ can then be written as a collection of operators $\{ \mathsf{K}_{i, a_i}(s_i)\}_{a_i, s_i}$ with $\mathsf K_{i, a_i}(s_i) \in \mathcal L  \left(B_i \otimes B^\text{out}_{(i)} \otimes B^\text{in}_{(i)} \right)$ where
$$B^{\text{in}}_{(i)} := \bigotimes_{j\in {N_\text{in}}(i)} B_{j \to i},\,  B^{\text{out}}_{(i)} := \bigotimes_{j\in {N_\text{out}}(i)} B_{i \to j},$$
effectively removing the memory link $B_{i \to i}$ (see the definitions of $N_\text{in}(i)$ and $N_\text{out}(i)$ in \cref{subsec:classical_strategy}). These operators are defined as
\begin{align*}
    \mathsf{K}_{i, a_i}(s_i) &\coloneqq \mathsf{W}_i(s_i) * \Pi_{i, a_i}, \\
    &= \tr_{B_{i \to i}} \left[ \left(\mathsf{W}_i(s_i)^{\top[B_{i\to i}]} \otimes I_{B^\text{in}_{(i)}} \right) \cdot \left(I_{B_i} \otimes I_{B^\text{out}_{(i)}} \otimes \Pi_{i, a_i} \right)\right],
\end{align*}
By construction, because the link product of two positive matrices is positive~\cite{chiribella2008quantum, chiribella2099theoretical}, the $\mathsf{K}_{i, a_i}(s_i)$ are positive matrices. Furthermore, the property that they can be decomposed into two quantum operations $\mathsf{W}_i(s_i)$ and $\Pi_i$ can be characterized by a normalization constraint, two constraints enforcing the comb structure, and a positivity constraint.~\cite{chiribella2008quantum, chiribella2099theoretical, gutoski2007toward, taranto2025higherorderquantumoperations}. On fixed Hilbert spaces, these constraints characterize a channel followed by a POVM, i.e.\ a one-slot quantum comb, rather than a channel followed by a projective measurement. Since any POVM is a projective measurement on a dilated space (Naimark), and the supremum in~\Cref{eq:seesaw_combs} is over unbounded dimension, this is without loss of generality for the quantum value. More precisely, writing $\mathsf{K}_i(s_i) = \sum_{a_i \in A_i} \mathsf{K}_{i, a_i}(s_i)$, for any $s_i \in S_i$, the set $\{\mathsf{K}_{i, a_i}(s_i)\}_{a_i}$ defines a local (mixed state) strategy for party $i$ if and only if it satisfies 
\begin{align}
    \label{eq:comb_constraints}
    &\tr \mathsf K_i(s_i) = \dim B_i \dim B^\text{in}_{(i)} \\ 
    \label{eq:comb_constraints2}
    &\tr_{B^{\text{in}}_{(i)}} [ \mathsf{K}_i(s_i) ]\otimes I_{B^{\text{in}}_{(i)}} = \dim B^\text{in}_{(i)} \mathsf K_i(s_i) \\ 
    \label{eq:comb_constraints4}
    & \tr_{B^{\text{in}}_{(i)} \otimes B^{\text{out}}_{(i)}} \mathsf{K}_i(s_i) = \dim B^\text{in}_{(i)}I_{B_i}\\ 
    \text{and }  &\mathsf{K}_{i, a_i}(s_i) \geq 0,  \text{ for all } a_i \in A_i.
    \label{eq:comb_constraints3}
\end{align}
The normalization~\eqref{eq:comb_constraints} follows from~\eqref{eq:comb_constraints4} by taking the trace; we keep it as a separate condition only for clarity, since the initialization uses it directly. These are all permissible constraints for a semidefinite program (SDP).
We can now rewrite the behavior in~\Cref{eq:g_network} in the Choi picture as\footnote{Note the link product is associative.}  
\begin{align*}
        p(\mathbf a \vert \mathbf s) = \mathsf{K}_{1,a_1}(s_1) * \dots * \mathsf{K}_{n,a_n}(s_n) * \rho.
\end{align*}
The see-saw algorithm is then based on the bound
\begin{align}
    \label{eq:seesaw_combs}
    \omega_q \geq \sum_{s \in S} \pi(s) \sum_{a \in A} \mathcal V(a \vert s) \big(\mathsf{K}_{1,a_1}(s_1) * \dots * \mathsf{K}_{n,a_n}(s_n) * \rho \big).
\end{align}

The outline of the algorithm is as follows. First, we set the dimensions of all the quantum systems $B_i, B_{i \to j}$. 
Then, writing $D = \prod_{i=1}^n \dim B_i$, we generate a random density matrix $\rho$ on $\bigotimes_{i=1}^n B_i$ by taking a Haar-random pure state on a $D^2$-dimensional Hilbert space and tracing out a $D$-dimensional subsystem (equivalently, we sample $\rho$ from the Hilbert--Schmidt measure).
We also generate random positive semidefinite matrices $\{K_{i, a_i}(s_i)\}_{a_i, s_i}$ for each party $i$ by generating a random operator with an ancillary system $B_{|A_i|}$ of dimension $|A_i|$, from which the $|A_i|$ outcome operators are later recovered by measuring $B_{|A_i|}$ in the computational basis. This can be thought of as generating operators where each party's measurements have been purified. More precisely, we generate Haar random pure states of dimension $(\dim B_i\dim B^\text{out}_{(i)} \dim B^\text{in}_{(i)} \ |A_i|)^2$ of which we trace out a $(\dim B_i  \dim B^\text{out}_{(i)} \dim B^\text{in}_{(i)} \ |A_i|)$-dimensional system, defining operators $\overline{K}_i(s_i) \in \mathcal L(B_i \otimes B^\text{out}_{(i)} \otimes B^\text{in}_{(i)} \otimes B_{|A_i|})$. We then normalize these so as to satisfy \Cref{eq:comb_constraints} and project them onto the affine subspaces defined by \Cref{eq:comb_constraints2} and  \Cref{eq:comb_constraints4}, by tracing out, for each equation, the relevant subsystems of $\mathsf{K}_i(s_i)$ and replacing it with the value required by the constraint. Applied in sequence, these two orthogonal projections yield an operator $\widetilde{K}_{i}(s_i)$ satisfying both equations\footnote{The two projections are compatible: the trace-and-replace projection enforcing~\eqref{eq:comb_constraints4} alters $\mathsf{K}_i(s_i)$ only by a term of the form $M \otimes I_{B^\text{in}_{(i)} \otimes B^\text{out}_{(i)}}$, which lies in the affine subspace fixed by~\eqref{eq:comb_constraints2}. Hence applying them in sequence lands in the intersection, and the first constraint remains satisfied.} and such that 
\begin{equation*}
    \tr_{B_{|A_i|}} \widetilde{K}_{i}(s_i) = \tr_{B^{\text{in}}_{(i)} B_{|A_i|}} [\widetilde{K}_i(s_i)] \otimes \frac{1}{\dim B^\text{in}_{(i)} } I_{B_{(i)}^\text{in}}.
\end{equation*}
This last equation ensures that the system contained in $B_{|A_i|}$ has no influence on previous operations and corresponds to a purification of party $i$'s measurement.
Finally, if these operators are not positive semidefinite, we mix them with the identity to obtain positive eigenvalues and satisfy 
\Cref{eq:comb_constraints3}:
\begin{equation*}
    \mathsf{\tilde{K}}_{i}(s_i) := \frac{1}{1 + \dim B^\text{out}_{(i)} |A_i| \abs \lambda} \cdot \Big[ \widetilde{K}_{i}(s_i) + \abs\lambda \cdot I_{B_i B^\text{out}_{(i)} B^\text{in}_{(i)} B_{|A_i|}}\Big],
\end{equation*}
where $\lambda$ is the minimum eigenvalue of $\widetilde{K}_{i}(s_i)$. 

The set of operators $\{\mathsf{K}_{i, a_i}(s_i)\}_{a_i \in A_i, s_i \in S_i}$ are then obtained by measuring system $B_{|A_i|}$ in the computational basis, which can be expressed in the Choi picture as:
\begin{equation*}
    \mathsf{K}_{i, a_i}(s_i) = \mathsf{\tilde{K}}_{i}(s_i) * \ketbra{a_i}_{B_{|A_i|}}.
\end{equation*}
Once a random solution is initialized, we iteratively optimize the RHS of~\Cref{eq:seesaw_combs}, optimizing each of $\rho$ and $\{\mathsf{K}_{i, a_i}(s_i)\}_{a_i \in A_i, s_i \in S_i}$ satisfying the relevant constraints, ceteris paribus. Each iteration is clearly an SDP. This leads to a monotonically non-decreasing lower bound on $\omega_q$. We stop the algorithm when a round of optimization (of the state and each operator $\mathsf{K}_{i,a_i}(s_i)$) does not change the lower bound by more than a given tolerance, typically, $10^{-6}$. 

\subsection{Application to random XOR games}
\label{subsec:randomXORgames}
Using this generalized see-saw algorithm along with standard numerical tools used to analyze nonlocal games, we can analyze simple LC games and compare their classical and quantum values for different latency constraints and connectivity graphs. 

We apply the generalized see-saw algorithm to three-party XOR games. As we will see below, in this case we can trivially show that a see-saw algorithm for conventional nonlocal games suffices, but we perform the calculation using the generalized version to demonstrate its use and to finally motivate the example in~\Cref{subsec:pert_xor}.
We consider random three-party XOR games defined by a uniform distribution $\pi$ on their input and with probabilistic predicates of the form 
\begin{equation*}
    \mathcal{V}(a_1,a_2,a_3 \vert s_1,s_2,s_3) = \beta_{s_1, s_2, s_3} (a_1 \oplus a_2 \oplus a_3),
\end{equation*}
where for all $(s_1,s_2,s_3) \in \{0, 1\}^3$ we have 
\begin{equation}
\label{eq:3_party_xor_coefs}
    \beta_{s_1,s_2,s_3}(a_1 \oplus a_2 \oplus a_3) \coloneqq \max\{\hat{\beta}_{s_1,s_2,s_3} \cdot (-1)^{a_1 \oplus a_2 \oplus a_3}, 0\}
\end{equation} 
for some coefficients $ -1 \leq \hat{\beta}_{s_1,s_2,s_3} \leq 1$. We do this in order to fulfill the definition of an LC game in~\Cref{dfn:lc}.

Consider for instance the game defined by the coefficients
\begin{align}
\label{eq:coefs__}
\hat{\beta}_{0,0,0} &= 0.438 \qquad & \hat{\beta}_{1,0,0} &= 0.580 \\
\label{eq:coefs2__}
\hat{\beta}_{0,0,1} &= 0.610 \qquad & \hat{\beta}_{1,0,1} &= -0.502 \\
\label{eq:coefs3__}
\hat{\beta}_{0,1,0} &= 0.520 \qquad & \hat{\beta}_{1,1,0} &= -0.724 \\
\label{eq:coefs4__}
\hat{\beta}_{0,1,1} &= - 0.466 \qquad & \hat{\beta}_{1,1,1} &= -0.220.
\end{align}
These coefficients were randomly generated. We choose this instance because it exhibits some interesting properties, as summarized in \Cref{tab:randomXor}.
\begin{table}[htp!]
\renewcommand{\arraystretch}{1.4}
    \centering
    \begin{tabular}{|c|c|c|c|c|}
    \hline 
         Connectivity graph &  $\omega_c$ & Lower bound $\omega_q$ & Upper bound $\omega_q$ & Quantum violation?\\
    \hline
         $(1)$ \, $v_1 \quad v_2 \quad v_3$ & $0.42525$ & $0.43954$ & $0.43954$ (NPA) & Yes \\
    \hline
         $(2a)$ $v_1 \leftrightarrows v_2 \quad v_3$  & $0.42525$ & $0.44095$ &  $0.44095$ (NPA agg) & Yes\\
    \hline
        $(2b)$ $v_1 \quad v_2 \leftrightarrows v_3$  & $0.42525$ & $0.43985$ & $0.43985$ (NPA agg) & Yes\\
    \hline
        $(2c)$ $v_1 \leftrightarrows v_3 \quad v_2$  & $0.42525$ & $0.44049$ &  $0.44049$ (NPA agg) & Yes\\
    \hline
        $(3)$ $v_1 \leftrightarrows v_2 \leftrightarrows v_3$  & $0.50750$ & $0.5075$ & $0.5075$ (Algebraic) & No\\
    \hline
    \end{tabular}
    \caption{Computed classical and quantum values of the game defined by the coefficients given in \Cref{eq:coefs__,eq:coefs2__,eq:coefs3__,eq:coefs4__} with respect to different connectivity graphs. For the empty graph $(1)$, the upper bound on $\omega_q$ is obtained via the second level of the NPA hierarchy, augmented with the distinct-party triples ABC (products of one operator from each party). For the graphs $(2a)$, $(2b)$, and $(2c)$, the upper bounds on $\omega_q$ are obtained via the NPA hierarchy on the nonlocal game obtained by aggregating $(v_1,v_2)$, $(v_2, v_3)$, and $(v_1, v_3)$ respectively into a single party. Meanwhile the upper bound for the graph $(3)$ is given by $\omega_a = \frac{1}{8}\sum_{s_1, s_2, s_3} \abs{\hat{\beta}_{s_1, s_2, s_3}}$. The classical values are obtained by enumerating all the possible classical strategies as given by~\Cref{def:classical_strat}. The lower bounds on $\omega_q$ are obtained by running the generalized see-saw with all quantum systems $B_i, B_{i\to j}$ set to be qubits.}
    \label{tab:randomXor}
\end{table}

We first observe that in~\Cref{tab:randomXor} all connectivity graphs except $(3)$ exhibit quantum violations. This shows that whether or not an LC game has a quantum violation can depend on the connectivity graph, as in the case of the distributed CHSH game in~\Cref{subsubsec:dist_chsh}. 
The absence of a violation for $(3)$ is something generic to XOR games, as for such games, classical strategies can always achieve the algebraic bound on the graph $(3)$, by having parties $v_1$ and $v_3$ producing a fixed output $0$, and $v_2$ producing his output depending on the full input $(s_1, s_2, s_3)$ which has been communicated to him by $v_1$ and $v_3$.

The particularity of XOR games also appears in the results for the graphs $(2a), (2b)$ and $(2c)$ where the lower and upper bounds on $\omega_q$ match, meaning that for each graph, the see-saw converged to an optimal quantum strategy for the aggregated game. For XOR games, this is trivial: because the two connected parties can communicate their inputs, one of them can output a fixed value $0$, and the other one will produce his output depending on their joint input. 

Finally we observe in \Cref{tab:randomXor} that the classical values are the same for all connectivity graphs except $(3)$, where the maximum algebraic value is achieved. This means that for any physical arrangement of the parties (which pair of parties are physically closer), relaxing the latency constraints such that two physically close parties exchange their inputs does not help for this particular game. This in contrast to what happens in the quantum case, where the quantum value can be different depending on the physical arrangement. In fact, we can order the different graphs based on their quantum values 
$$\omega_q^{(1)} < \omega_q^{(2b)} < \omega_q^{(2c)} < \omega_q^{(2a)} < \omega_q^{(3)} .$$
Here, we stress that because the solutions of the different SDPs are obtained with a very small duality gap (typically on the order of $10^{-8}$), these results give genuine separations between these different graphs. 

We can also obtain analytic results for the three-party XOR games where $\hat \beta_{s_1,s_2,s_3} \in \{\pm 1\}$.
First consider the case where no parties can communicate. This is then a conventional nonlocal game.
We can write the winning probability of a three-party XOR game with uniformly distributed inputs as
\begin{equation*}
    p_{\mathrm{win}} = \frac{1}{8} \sum_{s_1,s_2,s_3}\sum_{a_1,a_2,a_3} \beta_{s_1,s_2,s_3}(a_1 \oplus a_2 \oplus a_3)p(a_1,a_2,a_3|s_1,s_2,s_3).
\end{equation*}
The quantum value $\omega_q$ can be upper bounded using the third level of the NPA hierarchy~\cite{navascues2007bounding,navascues2008convergent}, which we denote as $\omega_q^{\mathrm{upper}}$. 
And the classical value $\omega_c$ can be obtained via linear programming~\cite{zukowski1999strengthening,kaszlikowski2000violations}.

Now suppose $v_1$ and $v_2$ can communicate while $v_3$ cannot communicate with either. That is, the connectivity graph is
$$ G = v_1 \leftrightarrows v_2 \quad v_3.$$
This becomes an LC game.
In this case, $v_1$ and $v_2$ can be regarded as one aggregated party even without using port-based teleportation, as mentioned in~\Cref{subsec:randomXORgames}.
The winning probability becomes 
\begin{equation*}
    p_{\mathrm{win}} = \frac{1}{8} \sum_{s_1,s_2,s_3}\sum_{a_2,a_3} \beta_{s_1,s_2,s_3}(a_2 \oplus a_3)p(a_2,a_3|(s_1,s_2),s_3).
\end{equation*}
Thus, this is a two-party nonlocal game, for which the quantum value of the aggregated game $\omega_q^\mathrm{agg}$  can be computed exactly using an SDP~\cite{cleve2004consequences}. The classical value $\omega_c^\mathrm{agg}$ is again obtained via a linear program. 

Building on the above observation, we now investigate the behavior of three-party XOR games before and after aggregation.
We investigate the following examples:\footnote{Here we use the coefficients $\hat{\beta}_{s_1,s_2,s_3}$ in~\Cref{eq:3_party_xor_coefs} to define the probabilistic predicate. }
\begin{equation*}
    \begin{aligned}
        &\mathrm{XOR}_1: \hat{\beta}_{s_1,s_2,s_3} = (-1)^{s_1s_2s_3}, \\
        &\mathrm{XOR}_2: \hat{\beta}_{s_1,s_2,s_3} = (-1)^{s_1s_2}, \\
        &\mathrm{XOR}_3: \hat{\beta}_{s_1,s_2,s_3} = (-1)^{s_1s_3}, \\
        &\mathrm{XOR}_4: \hat{\beta}_{s_1,s_2,s_3} = (-1)^{s_1(s_2 \oplus s_3)}. 
    \end{aligned}
\end{equation*}
The quantum and classical values of these three-party XOR games before and after aggregating $v_1$ and $v_2$ are listed in \Cref{table:3-xor}. Note for the quantum value $\omega_q$, we can give explicit quantum strategies that match the upper bounds $\omega_q^{\mathrm{upper}}$ from the NPA algorithm within numerical precision.
Furthermore, for these special three-party XOR games, the quantum values $\omega_q^\mathrm{agg}$ can be verified analytically by transforming $\mathrm{XOR}_{i}$ into a generalized CHSH game~\cite{PhysRevLett.108.100402}. See~\Cref{appendix:3_party_xor} for details.
\begin{table}[htp!]
    \centering
    \begin{tabular}{|c|c|c|c|c|}
        \hline
        Game & $\omega_c$ & $\omega_q$ & $\omega_c^\mathrm{agg}$ & $\omega_q^\mathrm{agg}$  \\
        \hline
        $\mathrm{XOR}_1$ & $\frac{7}{8}$ & $0.8750$ & $\frac{7}{8}$ & $\frac{4+\sqrt{10}}{8}$ \\
        \hline
        $\mathrm{XOR}_2$ & $\frac{3}{4}$ & $0.8536$ & $1$ & $1$ \\
        \hline
        $\mathrm{XOR}_3$ & $\frac{3}{4}$ & $0.8536$ & $\frac{3}{4}$ & $\cos^2 \frac{\pi}{8}$ \\
        \hline
        $\mathrm{XOR}_4$ & $\frac{3}{4}$ & $0.7500$ & $\frac{3}{4}$ & $\cos^2 \frac{\pi}{8}$ \\
        \hline
    \end{tabular}
    \caption{Quantum and classical values of three-party XOR games $\mathrm{XOR}_{i}$ before and after aggregating $v_1$ and $v_2$. For $\omega_q$, we compute an upper bound via the NPA hierarchy and find matching lower bounds up to numerical precision. $\omega_q^\mathrm{agg}$ is obtained analytically. }
    \label{table:3-xor}
\end{table}
As shown in \Cref{table:3-xor}, for $\mathrm{XOR}_1$ and $\mathrm{XOR}_4$, there is no quantum advantage when no parties can communicate, but there is one when $v_1$ and $v_2$ can communicate. Meanwhile for $\mathrm{XOR}_2$, we see the opposite phenomenon: while there was a quantum advantage when no parties can communicate, it disappears when $v_1$ and $v_2$ can communicate. Finally, for $\mathrm{XOR}_3$, communication between $v_1$ and $v_2$ has no effect on the classical or quantum values. This shows that even for three-party XOR games where $\hat \beta_{s_1,s_2,s_3} \in \{\pm 1\}$, the classical and quantum values can exhibit diverse behaviors depending on the connectivity graph.


\subsection{A fixed spatial layout with varying latency constraints}
\label{subsec:pert_xor}
We are now ready to address one of the main goals of this paper: understanding how a fixed spatial layout of parties gives rise to different realizable correlations depending on the latency constraint. A spatial layout determines the communication latencies between the parties, and each choice of latency constraint then induces a corresponding connectivity graph. As the latency constraint is relaxed, this graph changes, leading to qualitatively different regimes: a nonlocal game when no communication is possible, a latency-constrained game when only some parties can communicate, and an unconstrained, fully communicating game when all pairs of parties can communicate. This is the phenomenon illustrated schematically in~\Cref{fig:ABC}.

Here we will analyze a simple spatial layout with three parties arranged in a line with the second party at the midpoint as in~\Cref{fig:line}.
\begin{figure}[htpb!]
  \centering
  \begin{tikzpicture}[scale=1]
   \coordinate (A) at (0,0);
   \coordinate (B) at (3,0);
   \coordinate (C) at (6,0);

   \draw[] (A) -- (B) node[midway, below] {$d$};
   \draw[] (B) -- (C) node[midway, below] {$d$};

    \node[above] at (A) {$v_{1}$};
    \node[above] at (B) {$v_{2}$};
    \node[above] at (C) {$v_{3}$};

    \filldraw (A) circle (2pt);
    \filldraw (B) circle (2pt);
    \filldraw (C) circle (2pt);
  \end{tikzpicture}
  \caption{Three parties arranged on a line, with $v_1$ and $v_3$ each at a distance $d$ from $v_2$.}
  \label{fig:line}
\end{figure}
\noindent We wish to derive fundamental physical limits on their correlations by setting the communication latency to be the distance divided by the speed of light $c$. This spatial layout then induces a family of SISO games parameterized by the latency constraint $t$ via the procedure given in~\Cref{sec:tau_lc}:\footnote{Note that using the continuous-time LC game framework in~\Cref{subsec:cont_time} would be much more natural. }
\begin{itemize}
    \item We assume all parties receive their inputs at time $0$ and must produce their outputs within time $t$, the latency constraint. 
    \item Assume the relevant time values are rational: $V_T = \{t,\frac d c,\frac{2d}{c}\} \subseteq \mathbb Q$. Define $t_0$ as in~\Cref{eq:time_increment}. 
    \item We will thereby consider a SISO, $\tau$-step LC game where $\tau = \frac{t}{t_0}$ and the latency function $\ell$ takes values in $\{1,\frac{d}{ct_0},\frac{2d}{ct_0}\}$.
\end{itemize}
Now, we observe that depending on the value of the latency constraint $t$, we end up with different special cases of SISO games. In particular, we have the following three regimes:
\begin{enumerate}
    \item For $t \in [0,\frac d c)$,\footnote{This should be an interval of rational numbers and not real numbers, but we will ignore this for now. } no parties can communicate and this is effectively a conventional nonlocal game.
    \item For $t \in [\frac d c, \frac{2d}{c})$, this is effectively a simple LC game where parties can communicate over a single round. We can make this simplification because multi-round communication would require more time steps than $\tau$. The connectivity graph is
    \begin{align}
    \label{eq:two_line_graph} 
     v_1 \leftrightarrows v_2 \leftrightarrows v_3.
    \end{align}
    \item For $t \geq \frac{2d}{c}$, this is an unconstrained game since all parties have access to all inputs.
\end{enumerate}
Now, we want to find a specific predicate $\mathcal V$ and input distribution $\pi$ for which we can numerically compute the classical value and bounds on the quantum value achievable with respect to this spatial layout \emph{as a function of time}. 

To make it more interesting, for the family of SISO games defined by $\mathcal V,\pi$ for different regimes of $t$, we ask that the classical and quantum values are different for the first two regimes of $t$, and are each strictly increasing over the three regimes. 
In addition, we ask that the LC game corresponding to the intermediate regime $t \in [\frac d c, \frac{2d}{c})$ has a quantum value that we cannot easily compute via direct observation or with algorithms designed for nonlocal games, such as the see-saw method and the NPA hierarchy. 
As we saw above, the distributed CHSH game, the extended CHSH game, and the three-party XOR games where $\hat \beta_{s_1,s_2,s_3} \in \{\pm 1\}$ all become bipartite nonlocal games after aggregation. The quantum values are therefore the same according to~\Cref{thm:aggregable}.
Furthermore, in~\Cref{subsec:randomXORgames}, we considered a randomly generated tripartite XOR game with respect to different connectivity graphs. However, because of the structural properties of XOR games, for any connectivity graph, the corresponding simple LC game is either just a nonlocal game (empty graph), an aggregable LC game which becomes a nonlocal game after aggregation (graph with one bidirected edge), or achieves the maximum algebraic value (two or more bidirected edges). 
Lastly, for many LC games, the quantum value is simply equal to the forwarding value, making the computation again equivalent to computing the quantum value for a nonlocal game.
Hence, we can use direct observation or algorithms for nonlocal games to compute the quantum values of all the LC games we have seen so far.
Is there an example of a LC game for which we genuinely need the generalized see-saw algorithm in~\Cref{subsec:seesaw} to compute bounds on the quantum value? 
We will informally call such an LC game ``nontrivial.''

To find $\mathcal V ,\pi$ with the desired properties (intermediate regime is nontrivial, quantum advantage for first two regimes, and strictly increasing classical and quantum values), we consider a more complicated tripartite game based on the randomly generated XOR game in \Cref{subsec:randomXORgames} but additionally mixed with bipartite games between $v_1$ and $v_3$. In particular, we will add a CHSH game and a ``very lazy'' guess-your-neighbor's-input game \cite{Almeida2010GYNI, Branciard2015Lazy}. This second game is defined by requiring $v_1$ to satisfy $a_1 = s_3$ whenever $s_1 = 1$, which requires signaling. There are no other requirements in this game; hence the name. 

The two bipartite games are carefully chosen to obtain the desired properties. The CHSH game between $v_1$ and $v_3$ ensures that in the intermediate regime, the quantum value will be strictly smaller than the maximum algebraic value of the game (see Tsirelson's bound~\cite{Cirelson1980}). Meanwhile, the ``very lazy'' guess-your-neighbor's-input game creates a gap between the maximum algebraic value of the game and its $G$-signaling value with respect to the connectivity graph in~\Cref{eq:two_line_graph}. Indeed, while the ``very lazy'' guess-your-neighbor's-input game has a maximum algebraic value of 1, achieving this requires signaling between $v_1$ and $v_3$. Note that these two bipartite games are compatible: a single (signaling) behavior can win both games with probability one. 

We let $\pi$ be the uniform distribution over the inputs and $\mathcal V$ be a probabilistic predicate that is a convex combination of these different games:
\begin{equation}
\label{eq:mixed_xor}
    \mathcal{V}(\mathbf a \vert  \mathbf s) = \lambda_ 1 \cdot \beta_{s_1, s_2, s_3}(a_1 \oplus a_2 \oplus a_3) + \lambda_2 \cdot (a_1 \oplus a_3 = s_1 \cdot s_3) + \lambda_3 \cdot [s_1 = 0 \lor (a_1=s_3 \land s_1=1)]
\end{equation}
with the coefficients $\lambda_1, \lambda_2, \lambda_3 \in [0, 1]$, such that $\lambda_1 + \lambda_2 + \lambda_3 = 1$ and $\beta_{s_1, s_2, s_3}$ are defined as in \Cref{eq:3_party_xor_coefs}.
The specific game we analyze is defined with $\lambda_1 = 0.2$, $\lambda_2 = 0.72$, $\lambda_3 = 0.08$ and with the coefficients $\hat{\beta}_{s_1,s_2,s_3}$ of \Cref{subsec:randomXORgames}:
\begin{align*}
\hat{\beta}_{0,0,0} &= 0.438 \qquad & \hat{\beta}_{1,0,0} &= 0.580 \\
\hat{\beta}_{0,0,1} &= 0.610 \qquad & \hat{\beta}_{1,0,1} &= -0.502 \\
\hat{\beta}_{0,1,0} &= 0.520 \qquad & \hat{\beta}_{1,1,0} &= -0.724 \\
\hat{\beta}_{0,1,1} &= - 0.466 \qquad & \hat{\beta}_{1,1,1} &= -0.220.
\end{align*}
These parameters were randomly generated. This specific instance is chosen so that the family of SISO games has the desired properties.

We run various optimization algorithms to compute the exact classical value of this game, as well as lower and upper bounds on $\omega_q$ for the three different regimes. These numerical results are summarized in \Cref{tab:PertXor}. Note that $\omega_c \leq \omega_f \leq \omega_q$ for all LC games.
We then give the numerical estimates of the quantum values of the possible aggregated games in~\Cref{tab:reduc}. Note that we do not consider aggregating $v_1$ and $v_3$ because these two parties cannot communicate in the intermediate regime.
\begin{table}[H]
\renewcommand{\arraystretch}{1.4}
    \centering
    \begin{tabular}{|c|c|c|c|c|c|}
    \hline
Latency constraint & Exact & Exact & Lower bound & Upper bound & Quantum \\
        $t$ & $\omega_c$ & $\omega_f$ & $\omega_q$ & $\omega_q$ & violation? \\
    \hline
         $t \in [0, \frac d c)$ & $0.68505$ & $0.74884$ & $0.74884$ & $0.74884$ (NPA) & Yes \\
    \hline
         $t \in [\frac d c, \frac{2d}{c})$  & $0.70150$ & $0.76136$ & $0.76277$ &  $0.88150$ ($G$-signaling) & Yes\\
    \hline
        $t \geq \frac{2d}{c}$ & $0.90150$ & $0.90150$ & $0.90150$ & $0.90150$ (Algebraic) & No\\
    \hline
    \end{tabular}
     \caption{Computed classical and quantum values for the tripartite game defined in~\Cref{eq:mixed_xor}.
     When $t \in [0, \frac d c)$, the upper bound on $\omega_q$ is obtained via the ``1 + |AB|'' level of the NPA hierarchy generalized to the multipartite setting, which includes the distinct-party triples ABC (products of one operator from each party), while the lower bound is computed using the see-saw algorithm. When $t \in [\frac d c, \frac{2d}{c})$ the upper bound on $\omega_q$ is obtained by computing the $G$-signaling value $\omega_G$ with respect to the graph in~\Cref{eq:two_line_graph}. The lower bound on $\omega_q$ for $t \in [\frac d c, \frac{2d}{c})$ is obtained with the generalized see-saw algorithm and corresponds to the best result obtained on 50 random initial solutions. Meanwhile, the classical and quantum values for $t \geq \frac{2d}{c}$ are given by the maximum algebraic value of the game as per~\Cref{prp:siso_trivial}. The values for forwarding strategies are obtained with matching lower and upper bounds obtained respectively via the see-saw algorithm and the NPA hierarchy. 
     Note that for $t \in [0,\frac d c)$, $\omega_f$ and $\omega_q$ coincide by definition. When computing all the lower bounds using the see-saw algorithm, the quantum strategies are optimized with quantum systems $B_i, B_{i\to j}$ set to be qubits. }
    \label{tab:PertXor}
\end{table}
\begin{table}[htp!]
\renewcommand{\arraystretch}{1.4}
    \centering
    \begin{tabular}{|c|c|c|c|c|c|}
    \hline
Aggregation & Lower bound & Upper bound \\
         & $\omega_q$ & $\omega_q$ \\
    \hline
         $(v_1 \quad v_2) \quad v_3$  & $0.75022$ & $0.75022$ (NPA) \\
    \hline
         $v_1 \quad (v_2 \quad v_3)$  & $0.74934$ &  $0.74934$ (NPA) \\
    \hline
    \end{tabular}
     \caption{Computed quantum values for the tripartite game defined in~\Cref{eq:mixed_xor}
     when aggregating $v_1$ and $v_2$ or $v_2$ and $v_3$. Lower bounds on the quantum value are obtained via the see-saw algorithm on the bipartite nonlocal game where each of the two parties (after aggregation) holds a qubit. The upper bounds are obtained via the NPA hierarchy. }
    \label{tab:reduc}
\end{table}
\noindent 
\newpage
\begin{figure}[htbp!]
  \centering
   \includegraphics[width=0.77\textwidth]{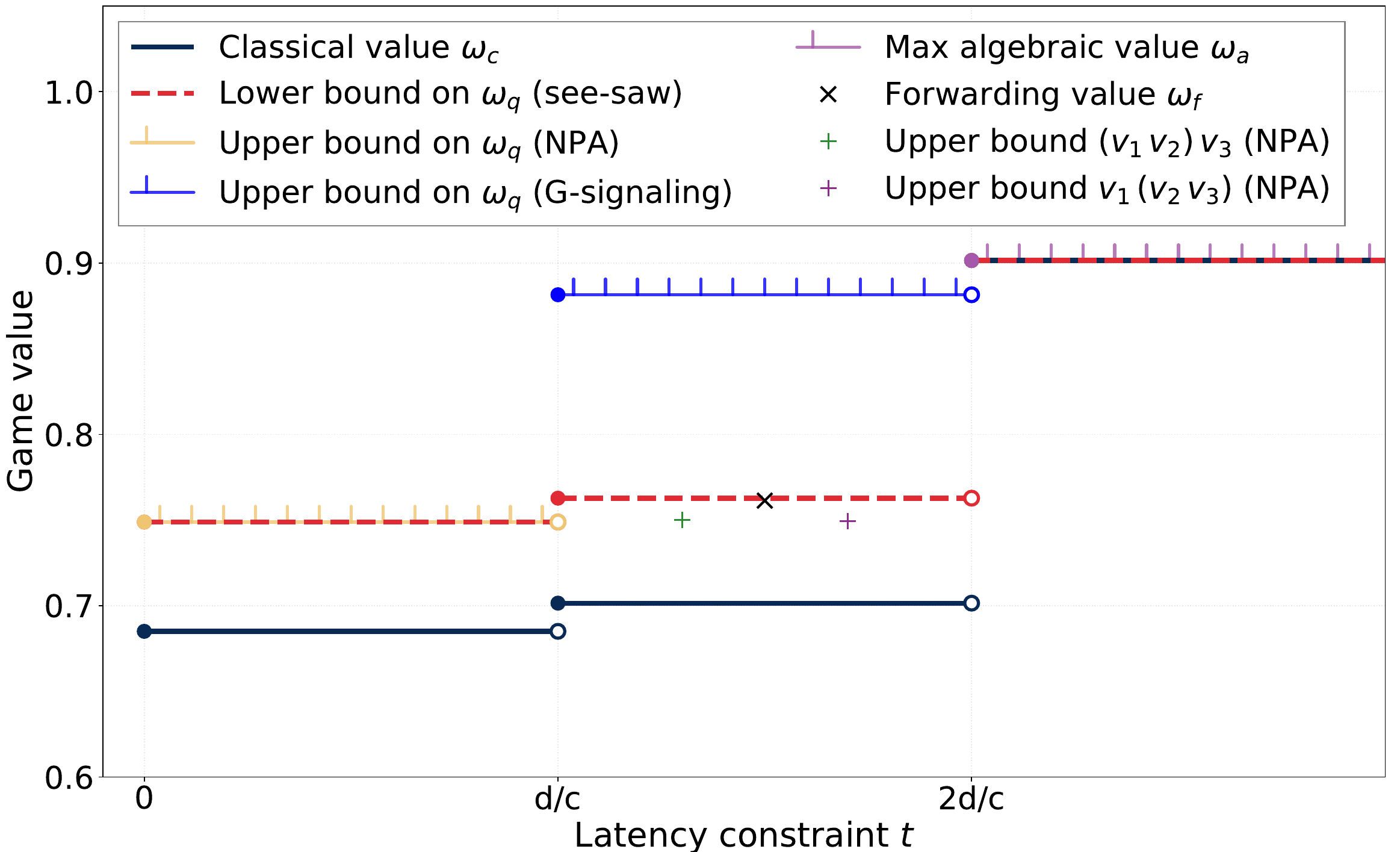}
   \caption{Graphic summarizing the results of~\Cref{tab:PertXor} and~\Cref{tab:reduc}. The three parties are arranged in a line as in~\Cref{fig:line}. The bounds certify progressively higher achievable classical and quantum values for the tripartite game defined in~\Cref{eq:mixed_xor} as the latency constraint is relaxed. The intermediate regime $t \in [\frac d c, \frac{2d}{c})$ corresponds to an nontrivial LC game. The lower bound on the quantum value in this regime is strictly greater than the forwarding value and upper bounds on the quantum values for the two possible aggregations $(v_1 \, \, v_2)$ and $(v_2 \, \, v_3)$. }
   \label{fig:threeXor}
\end{figure}
We summarize these results graphically in~\Cref{fig:threeXor}, where we plot the evolution of the classical and quantum values of the game as a function of time. We first observe that \emph{this is a step function}: the classical and quantum values are constant except at certain points where they jump. These jumps happen when the parties' light cones intersect. We also see from the graph that this family of SISO games has all the desired properties.
The classical and quantum values are strictly increasing over the regimes, and there is a quantum advantage for the first two. 
Furthermore, the game corresponding to the intermediate regime is nontrivial.
First, $\omega_f < \omega_q$. At a high level, this is because in order to win the CHSH game between $v_1$ and $v_3$ with high probability, they need to share a maximally entangled state. Hence, due to the monogamy of entanglement, the quantum state used cannot also be maximally entangled with $v_2$. Hence $v_2$ cannot know the outputs of $v_1$ and $v_3$ in a forwarding strategy. In comparison, a general quantum strategy allows for $v_1$ and $v_3$ to share their outputs with $v_2$, who can use them to win the tripartite XOR game. This reasoning is similar to the argument for why $\omega_f <\omega_q$ for the extended CHSH game in~\Cref{subsec:forwarding_strategies}.
Second, we see that the quantum values after aggregating $v_1$ and $v_2$ or $v_2$ and $v_3$ are strictly less than the lower bound obtained via the generalized see-saw algorithm.
This shows that our LC game is nontrivial. 

\section{Discussion}
\label{sec:discussion}
Bell inequality violation is one of the defining phenomena of quantum mechanics, taught in virtually every undergraduate physics curriculum and modern textbook on quantum mechanics. In this work, we give a physical interpretation of Bell inequalities by explicitly incorporating spacetime locality into their formulation. The core concept that we use is a \emph{latency constraint}. By setting the latency constraint sufficiently tight so that no parties can communicate, the inequalities that the parties' correlations must satisfy are exactly Bell inequalities. By slightly relaxing the latency constraint so that a subset of the parties can communicate, we obtain a new set of inequalities, LC inequalities, that generalize Bell inequalities in a very natural way. Like Bell inequalities, LC inequalities are bounds on local hidden variable theories: the correlations between classical parties must satisfy these bounds within the given latency constraint. 
Through examples of LC inequalities, such as that of the distributed CHSH game and the randomly generated nontrivial game, as well as their quantum violations, we showed that Bell inequality violations are actually a special case of a more general quantum phenomenon in which spacetime locality plays a central role.
In this sense, our work takes Bell's original insight~\cite{Bell1976LocalBeables} --- that Bell inequalities arise from \emph{locality} and their quantum violations are manifestations of quantum \emph{nonlocality} --- to its natural conclusion.

\paragraph{Experimental realization}
\begin{figure}[htbp!]
  \centering
  \begin{tikzpicture}[scale=.9]
    \pgfmathsetlengthmacro{\xstep}{6cm}
    \pgfmathsetlengthmacro{\ystep}{2.4cm}

    \pic (box11) at (0,2*\ystep) {tensorboxin={$W_{1}$}{$s_{1}$}};
    \pic (box12) at (0,1*\ystep) {tensorboxin={$W_{2}$}{$s_{2}$}};
    \pic (box13) at (0,0*\ystep) {tensorboxin={$W_{3}$}{$s_{3}$}};

    \pic (box21) at (\xstep,2*\ystep) {tensorboxout={$M_{1}$}{$a_{1}$}};
    \pic (box22) at (\xstep,1*\ystep) {tensorboxout={$M_{2}$}{$a_{2}$}};
    \pic (box23) at (\xstep,0*\ystep) {tensorboxout={$M_{3}$}{$a_{3}$}};

    \draw (box11-in1) -- ([xshift=-38pt]box11-inpin1);
    \draw (box12-in2) -- ([xshift=-38pt]box12-inpin2);
    \draw (box13-in3) -- ([xshift=-38pt]box13-inpin3);

    \draw (box11-out1) -- (box21-in1);
    \draw (box11-out2) to[out=0,in=180] (box22-in1);

    \draw (box12-out1) to[out=0,in=180] (box21-in2);
    \draw (box12-out2) -- (box22-in2);
    \draw (box12-out3) to[out=0,in=180] (box23-in2);

    \draw (box13-out2) to[out=0,in=180] (box22-in3);
    \draw (box13-out3) -- (box23-in3);

    \draw [decorate,decoration={brace,amplitude=5pt,mirror}]
    ($(box11-inpin1)+(-1.8,0)$) -- ($(box13-inpin3)+(-1.8,0)$)
    node[midway,xshift=-40pt] {$\ket{\psi} \otimes \vert\Phi^+\rangle^{\otimes N}$};

    \draw [dashed] (-1.7,-.8*\ystep) -- (-1.7,2.5*\ystep);
    \draw [dashed] (1.4,-.8*\ystep) -- (1.4,2.5*\ystep);
    \draw [dashed] (4.8,-.8*\ystep) -- (4.8,2.5*\ystep);
    \node [labelsty, text width=.7*\xstep, align=center] at (-3.2,-.7*\ystep) {Pre-share\\ entangled state};
    \node [labelsty, text width=.5*\xstep, align=center] at (-.2,-.7*\ystep) {Fast gates and measurements};
    \node [labelsty, text width=.5*\xstep, align=center] at (3.1,-.7*\ystep) {Send classical bits near speed of light};
    \node [labelsty, text width=.5*\xstep, align=center] at (6.5,-.7*\ystep) {Fast gates and measurements};
  \end{tikzpicture}
  \caption{Experimental operations corresponding to different parts of a quantum strategy for a simple LC game with connectivity graph
    $v_1 \leftrightarrows v_2 \leftrightarrows v_3$. The communication between different parties can be realized using quantum teleportation, whose implementation is included in the depicted operations. }\label{fig:experiment}
\end{figure}
The next natural step is to conduct experiments that can observe quantum violations in this new, relaxed latency constraint regime. In general, this will require, in addition to quantum measurements on a shared entangled state, \emph{real-time communication}. The parties will need to perform quantum operations on their quantum registers depending on the input they receive and send part of their quantum registers to each other as specified by the quantum strategy. As we mentioned above, this communication can be realized using quantum teleportation. The parties can pre-share Bell pairs and use them for quantum teleportation in real time. In particular, they only need to send each other classical bits as specified in the quantum teleportation protocol. This is summarized in~\Cref{fig:experiment}, where we label the experimental operations corresponding to different parts of the quantum strategy drawn in~\Cref{fig:g-quantum}.

Now, the main technical challenge, other than the already recognized challenges in quantum communication experiments such as realizing long-distance and long-lasting quantum entanglement, is the following: \emph{the entire quantum strategy needs to be completed within the latency constraint}. Fast quantum measurements in the computational basis at or below microsecond timescales are already possible~\cite{hensen2015loophole,shalm2015strong,giustina2015significant}. The quantum gates that the parties need to perform are usually even faster than the measurements for most physical platforms, assuming they can be efficiently decomposed into hardware-native gates.
For the communication step, to save time the transmitted classical bits ideally should be moving near the speed of light in vacuum. This can be approximately achieved via free-space electromagnetic communication, including radio transmission, microwave towers, or optical laser links. This is one reason quantum teleportation is preferred over physically sending quantum systems, as the latter usually requires an optical-fiber connection between the parties to achieve a high success probability. Light travels much slower in fiber than in free space.
Finally, we ideally should use quantum memories that can serve as local quantum registers and can store entanglement, both for measurement and for quantum teleportation purposes. Current physical platforms such as trapped ions, quantum dots, nitrogen-vacancy (NV) centers, and possibly neutral atoms used in emerging quantum networks can serve as such quantum memories~\cite{heshami2016quantum}.

\paragraph{Real-world applications}
The correlations realizable with quantum resources under tight latency constraints has a variety of real-world applications.
One example is high-frequency trading (HFT). It would be interesting to find LC games modeling HFT scenarios with a quantum advantage such as those of~\cite{brandenburger2015quantum,ding2024coordinating}. 

One immediate possibility is to consider the same HFT scenario in~\cite{ding2024coordinating}. A financial instrument $X$ is being traded at the NYSE (New York Stock Exchange) data center in Mahwah, NJ. This time, the same instrument $X$ is also being traded at Nasdaq's data center in Carteret, NJ. Meanwhile, instrument $Y$ is traded at the CME (Chicago Mercantile Exchange) data center far away in Aurora, Illinois. The price movements of instruments $X$ and $Y$ are sometimes positively correlated, sometimes negatively correlated. The trading servers at the different exchanges must quickly decide whether to issue a buy or sell order for their respective instrument so as to minimize the overall risk of their combined trading strategy.\footnote{Note that in HFT the ultra-fast trade decisions are actually order cancellations. We can consider the same scenario but instead of issuing orders we are canceling orders. A similar logic applies. Note also that~\cite{ding2024coordinating} instead considers which order (buy/sell) is issued first. } If the instruments are positively correlated, the servers trading $X$ issue the order opposite to the order issued by the server trading $Y$. For example, if the $X$ servers buy, the $Y$ server should sell. This would lower the overall risk of their combined trade. If the instruments are negatively correlated, the servers should issue the same order. 

Now, NYSE and Nasdaq's data centers are both in New Jersey, only 56 km apart. Meanwhile, the server in the CME data center in Illinois is about 1200 km away.\footnote{CME-NYSE is 1175 km, while CME-Nasdaq is 1177 km. } Thus, for latency constraints between about 0.2 and 4 ms, the trading servers at NYSE and Nasdaq can communicate with each other with one round of communication but neither can communicate with CME. The connectivity graph for the corresponding LC game is thus
\begin{align*}
    \text{CME} \quad \text{NYSE} \leftrightarrows \text{Nasdaq}.
\end{align*}
Since the HFT scenario in~\cite{ding2024coordinating} maps to the CHSH game, assuming we want to exactly copy the decision of the NYSE server at Nasdaq (otherwise there is no net trade of instrument $X$), \emph{this is exactly the extended CHSH game with communication between two parties considered in~\Cref{subsec:forwarding_strategies}}. 
That is, the task from the HFT perspective is two-fold:
\begin{enumerate}
    \item The three trading servers at the three exchanges make the correct decisions for hedging purposes.
    \item The NYSE and Nasdaq servers make the same decision.
\end{enumerate}
By~\Cref{prop:sep_forwarding_quantum}, quantum resources cannot achieve a higher probability of success in this task than with classical resources ($\omega_f = \frac 3 4$ is an upper bound on the quantum value when there is no communication) for latency constraints below $0.2$ ms, but they can for latency constraints between 0.2 and 4 ms. Hence, \emph{the existence of a quantum advantage in HFT depends on the timescale involved}.
Apart from HFT, it would also be interesting to consider other real-world scenarios such as ad hoc network routing~\cite{hasanpour2017quantum}, rendezvous on graphs~\cite{viola2024quantum}, and distributed systems~\cite{da2025entanglement} where a subset of the parties can communicate. 

With the framework of $\tau$-step LC games, we can also consider real-world scenarios where the parties are receiving multiple inputs and producing multiple outputs over the course of time. This should be a common characteristic in many real-world settings. For example, in HFT, this framework can be used to describe a real-time trading scenario where trading servers are receiving market data from the exchanges they are located at, communicating with each other, and making trades in real time. We can define the probabilistic predicate $\mathcal V$, input distribution $\pi$, and latency function $\ell$ for a specific HFT scenario, and then compute the classical and quantum values for the $\tau$-step LC game defined thereby. A separation between the two would entail a mathematically provable quantum advantage for this realistic trading scenario. In general, this framework applied to HFT allows us to answer the third question posed in~\Cref{sec:intro}.
Other latency-sensitive scenarios can be analyzed in a like manner.

\paragraph{Generalizations}
We can generalize LC games by considering more complicated physical scenarios. A natural generalization would be to accommodate \emph{parties in motion}. In this case, the distance between parties and therefore the latency function $\ell_L$, defined as the speed-of-light delay, itself will be time-dependent. 
Taking relativistic effects into account, for parties moving near the speed of light, $\ell_L$ will also depend on the parties' velocities. However, in general, we would need to specify a frame of reference by which we define time. Given this complication, it may be more desirable to define a coordinate-free definition of communication latency. That is, we define moving parties as world lines in spacetime. Then, the transmitting parties emit photons, whose worldlines are light-like curves. The receiving parties then receive the transmission when the light-like curves intersect their world lines, as shown in~\Cref{fig:worldlines} for a particular frame of reference. 
\begin{figure}[htbp!]
  \centering
  \begin{tikzpicture}[line cap=round, line join=round, scale=.9]

    \draw[->] (0,0) -- (7,0) node[below right] {$x$};
    \draw[->] (0,0) -- (0,6) node[above left] {$t$};

    \draw[->] (2,0)
      .. controls (2.4,1.0) and (3.0,2.0) .. (2.0,3)
      .. controls (1.2,4.0) and (2.5,4.8) .. (2.6,5.6)
      coordinate [pos=0.0] (eA)
      node[above] {$A$};

    \draw[->]
      (5.5,0)
      .. controls (4.9,1.5) and (5.1,2.8) .. (5.6,3.6)
      .. controls (6.0,4.4) and (5.6,5.0) .. (5.6,5.8)
      coordinate [pos=0.7] (eB)
      node[above] {$B$};

    \draw[dashed,->]  ($(eA)+(0.6,-0.9)$) -- ++(45:4.35) node[midway, above=2pt, sloped] {photon};
  \end{tikzpicture}
  \caption{Spacetime diagram of two moving parties where party $A$ emits a photon to party $B$. }
  \label{fig:worldlines}
\end{figure}
Furthermore, since LC games by construction involve both relativity (spacetime locality) and quantum mechanics (quantum entanglement), the natural next step is to develop a quantum field-theoretic formulation of LC games.

We can also include gravitational effects by considering an LC game with parties in curved spacetime. In particular, it would be interesting to consider a specific physical scenario that requires a general relativistic treatment and derive the fundamental limits on correlations the parties can exhibit using classical versus quantum resources. One exciting example is to consider two parties, where the first party is orbiting a black hole outside the event horizon while the second is falling into the black hole. In particular, when the second party crosses the event horizon, we have an intrinsically asymmetric communication constraint: the party outside the black hole can transmit information to the party inside, but not vice versa. \emph{This is actually a physical scenario modeled by an LC game with a connectivity graph $G$ where two vertices are only connected in one direction}.

Other possible generalizations include
\begin{itemize}
    \item Nonlinear objective functions. Bell inequalities are linear functions of the  behavior. We can consider nonlinear functions of the behavior~\cite{rosset2016nonlinear,klep2024state} and analyze corresponding fundamental limits for an LC game.
    \item Non-cooperative games. There is a body of literature on the usage of quantum entanglement to obtain more and better (in terms of social welfare) Nash equilibria in non-cooperative games~\cite{auletta2021belief,pappa2015nonlocality,bolonek2017three,khan2018quantum, abbott2024improving}. In the LC games setting, we can analyze what Nash equilibria can be reached if some parties can communicate.
    \item Photon loss. In many Bell experiments, entangled photon pairs are used as the shared quantum system between the different parties. These photons are readily absorbed by physical media, and the consequent effects on the behaviors realized can be described mathematically~\cite{gigena2024robustselftestingbellinequalities,ding2024coordinating,tba}. Compared to conventional nonlocal games, in the LC games setting, a new phenomenon appears: parties that can communicate can tell each other whether or not photon loss has occurred.
\end{itemize}

\paragraph{Open problems}
There is a plethora of open problems regarding LC games and their generalizations. 
One important direction is to mathematically analyze the threshold times $\tau_c(\alpha),\tau_q(\alpha)$ defined in~\Cref{subsec:siso}. This would shed more light on the time advantage that quantum resources can achieve in SISO games.
It will also be important to find more efficient numerical optimization algorithms for LC games. This is especially pressing since the numerical optimizers for conventional nonlocal games can already face scaling problems~\cite{mortimer2025bounding}. An efficient optimizer is crucial for designing experiments and for evaluating real-world applications. Another question is to derive an NPA-like upper bound on the quantum value of LC games via techniques such as those in~\cite{klep2024state}. Lastly, an interesting question is whether the enormous amount of quantum entanglement in port-based teleportation is necessary in order to achieve the reductions in~\Cref{thm:aggregable},~\Cref{thm:LC_to_broadcasting}, and~\Cref{thm:siso_to_LC}. Related to this question is whether there is a provable reduction in the amount of entanglement needed to achieve a certain winning probability if we can make use of auxiliary parties.
\\~\\
\noindent  \emph{Author contributions.} 
D.D.\ conceived the idea of explicitly integrating spacetime locality into the formulation of Bell inequalities, was in charge of the overall direction and planning, and was primarily responsible for the writing of the manuscript. Z.J.\ contributed to the definition of simple and multi-step LC games, led the effort in elucidating the structure and separations among classes of LC game strategies, and improved the presentation of the manuscript. P.P.\ devised the generalization of the see-saw method to LC games, was responsible for its implementation, analyzed the simulations for random, extended and perturbed XOR games, and designed the nontrivial game. M.X. contributed to the proofs of the basic properties of LC games. X.X. derived all the results that use port-based teleportation, contributed to the analysis of distributed games, and analyzed specific three-party XOR games. All authors discussed the results and commented on the manuscript.
\\~\\
\noindent \emph{Acknowledgments.}
We would like to thank Hamed Adami, Elliott Gesteau, Nafiz Ishtiaque, Ziwen Liu, Guorui Ma, Veronica Pasquarella, Leonardo Santilli, Richard Schoen, and Jie Wang for helpful discussions. We thank Simon H{\"o}fer for pointing out the connection between LC games and port-based teleportation. We thank Yuqing Li for writing one of the SDP algorithms that we used in our paper. DD would like to thank God for all of His provisions.
\\~\\
\noindent \emph{AI disclosure.}
ChatGPT 5.5 was used to proofread the manuscript and also used collaboratively to improve its language. 

\bibliographystyle{unsrt}
\bibliography{ref}

\appendix
\addtocontents{toc}{\protect\setcounter{tocdepth}{1}}

\section{Introduction to Nonlocal Games}
\label{app:prelim}
In this section, we give a brief introduction to nonlocal games. A nonlocal game typically involves multiple parties that cannot communicate with each other. This is shown schematically in~\Cref{Fig:nonlocal_game}.
\begin{figure}[htbp!]
    \centering
    \begin{tikzpicture}[node distance=1.7cm,
  player/.style = {draw, circle, minimum width=.9cm, fill=gray!15},
  singlearrow/.style = {single arrow, draw, inner sep=0pt,
    minimum height=.7cm, minimum width=4mm, single arrow head extend=.3pt, rotate=-90}]
  
  \node[player] (P1) at (0,0) {$1$};
  \node[player, right=of P1] (P2) {$2$};        
  \node[player] (Pn) at ($(P2)+(3.4,0)$) {$n$}; 

  \node at ($(P2)!0.5!(Pn)$) {$\cdots\cdots$};

  \node[singlearrow] at ([yshift=1cm]P1) {};
  \node[singlearrow] at ([yshift=1cm]P2) {};
  \node[singlearrow] at ([yshift=1cm]Pn) {};
  \node at ([yshift=1.7cm]P1) {$s_{1}$};
  \node at ([yshift=1.7cm]P2) {$s_{2}$};
  \node at ([yshift=1.7cm]Pn) {$s_{n}$};

  \node[singlearrow] at ([yshift=-1cm]P1) {};
  \node[singlearrow] at ([yshift=-1cm]P2) {};
  \node[singlearrow] at ([yshift=-1cm]Pn) {};
  \node at ([yshift=-1.7cm]P1) {$a_{1}$};
  \node at ([yshift=-1.7cm]P2) {$a_{2}$};
  \node at ([yshift=-1.7cm]Pn) {$a_{n}$};
\end{tikzpicture}
\caption{A nonlocal game with $n$ parties. Each party $i$ receives input $s_i$ and produces output $a_i$.}
\label{Fig:nonlocal_game}
\end{figure}
The formal definition is as follows.
\begin{dfn}
    Let $n\geq2$ be an integer. Let $S_i,A_i$ be finite sets for $i\in[n]$. We define a probabilistic predicate $\mathcal{V}$ which is a map
    \begin{align*}
        \mathcal V: \prod_{i = 1}^n A_i \times \prod_{i=1}^n  S_i \to [0,1].
    \end{align*}
    Furthermore, we define the input distribution $\pi$ which is a probability distribution on $\prod_{i=1}^nS_i$. A \textbf{nonlocal game} is described by the tuple $(\mathcal{V},\pi)$.
\end{dfn}
\noindent Informally speaking, in a nonlocal game, each party $i$ receives input $s_i$ and produces an output $a_i$. All parties have full knowledge of the input distribution $\pi$ and the winning predicate $\mathcal{V}$. The goal is to maximize the winning probability. Parties may agree on a joint strategy in advance based on their knowledge, but once the game starts, communication is forbidden. When playing the game, implementing the strategy realizes a \textbf{behavior}, defined as a conditional probability distribution
$$ p(a_1, a_2, \cdots, a_n \vert s_1, s_2, \cdots, s_n).$$
A behavior is interpreted as the distribution over the different outputs the parties produce, conditioned on their inputs.
By plugging the behavior into the predicate, we can compute the \textbf{winning probability}:
\begin{align*}
    p_\mathrm{win} \coloneqq \sum_{a_i \in A_i, s_i \in S_i} \pi(s_1, s_2,\cdots, s_n) \cdot p(a_1, a_2, \cdots, a_n \vert s_1, s_2, \cdots, s_n) \mathcal V(a_1,a_2, \cdots, a_n \vert s_1, s_2, \cdots, s_n).
\end{align*}

We proceed to define what strategies the parties can use with classical resources. We first introduce \textbf{deterministic strategies}.
\begin{dfn}
    Let $(\mathcal{V},\pi)$ be a nonlocal game. A \textbf{deterministic strategy} is given by functions ${\{f_i\}}_{i=1}^n$, where 
    $$f_i:S_i\rightarrow A_i.$$
    The behavior realized by a deterministic strategy is simply
    \begin{align*}
    p(\mathbf a\vert\mathbf s)=\prod_{i=1}^n\delta_{a_i,f_i(s_i)},
\end{align*}
where $\mathbf a\in\prod_{i=1}^nA_i,\mathbf s\in\prod_{i=1}^nS_i$.
\end{dfn}
\noindent More generally, for a \textbf{classical strategy}, parties can use shared randomness. In this case, a classical strategy can be expressed as a probabilistic mixture of deterministic strategies, and the behavior realized is therefore a convex combination of deterministic behaviors. A bound on the largest winning probability attainable by a classical strategy is known as a \textbf{Bell inequality}.

For a quantum strategy, parties can have access to a shared entangled state and perform local quantum measurements according to the inputs they receive. We therefore have the following definition.
\begin{dfn}
    Let $(\mathcal{V},\pi)$ be a nonlocal game. Let $B_i$ be a Hilbert space for $i\in[n]$. A \textbf{quantum strategy} is a tuple $(\ket{\psi},{\{M_i(s_i)\}}_{i\in [n], s_i \in S_i})$, where 
    $$\ket{\psi}\in\bigotimes_{i=1}^nB_i$$
    is a quantum state and 
    $$M_i(s_i)={\{\Pi_{i, a_i}(s_i)\}}_{a_i\in A_i}$$
    is a projective measurement on $B_i$ for each $s_i$. The behavior realized by a quantum strategy is given by
    \begin{align*}
    p(\mathbf a\vert\mathbf s)=\bigbra{\psi}\bigotimes_{i=1}^n\Pi_{i, a_i}(s_i)\bigket{\psi}.
\end{align*}
\end{dfn}
\noindent Classical behaviors form a subset of quantum behaviors. However, with quantum resources, parties have the potential to achieve a higher winning probability. This phenomenon is called \textbf{Bell inequality violation}.
Several nonlocal games are known for which this occurs~\cite{XOR,eisert1999,pappa2015conflicting_interest}, specific examples being the CHSH game~\cite{bell1964einstein,CHSH} and the magic square game~\cite{arkhipov2012extending}.

\section{Using Port-Based Teleportation in LC Games}
\label{app:teleportation}

\subsection{Quantum teleportation and port-based teleportation}
\label{app:pbt_dfn}

Quantum teleportation is a protocol that transfers an unknown quantum state from one party to another using shared entanglement and classical communication~\cite{PhysRevLett.70.1895}.

Let $A_0$ denote the $d$-dimensional Hilbert space that contains the quantum state to be teleported. Alice and Bob share the maximally entangled state
\begin{equation*}
\ket {\Phi^{+}_{AB}} = \sum_{i=0}^{d-1} \frac{1}{\sqrt d} |ii\rangle  \in A \otimes B,\qquad \Phi^{+}_{AB} = \ket {\Phi^{+}_{AB}} \bra{\Phi^{+}_{AB}} .
\end{equation*}
Given the quantum state to be teleported $\rho \in \mathcal{L}(A_0)$, Alice performs a joint measurement $M^{TP}$
\begin{equation*}
M^{\mathrm{TP}}\coloneqq \{ M_k^{\mathrm{TP}} \}_{k=1}^{d^2}, 
\end{equation*}
where $M_k^{\mathrm{TP}}$ are measurement operators on $A_0
A$.
In the standard quantum teleportation protocol, $M^{\mathrm{TP}}$ is chosen to be a projective measurement onto a maximally entangled orthonormal basis.
Conditioned on outcome $k$ of $M^{\mathrm{TP}}$, the unnormalized post-measurement state on Bob's system is
\begin{equation*}
\tilde{\rho}_k =
\mathrm{Tr}_{A_0 A}
\left[
(M_k^{\mathrm{TP}} \otimes I_B)
\left(
\rho_{A_0} \otimes \Phi^{+}_{AB}
\right)
\right].
\end{equation*}
Alice sends the outcome $k$ to Bob, who then applies a \textbf{correction unitary} $\sigma_k$ on his subsystem. The resulting state is
\begin{equation*}
\sigma_k \left( \frac{\tilde \rho_k}{\mathrm{Tr}(\tilde \rho_k)} \right) \sigma_k^\dagger.
\end{equation*}
Averaging over all outcomes, this protocol implements the quantum channel 
\begin{equation*}
\Lambda(\rho_{A_0}) = \sum_{k=1}^{d^2} \sigma_k 
\mathrm{Tr}_{A_0A}
\left[
(M_k^{\mathrm{TP}} \otimes I_B)
(\rho_{A_0} \otimes \Phi^{+}_{AB})
\right]
\sigma_k^\dagger.
\end{equation*}
We choose measurements and correction unitaries so that $\Lambda$ is the identity channel $\mathcal I_{A_0 \to B}$.

In~\cite{PhysRevLett.101.240501}, a teleportation scheme—referred to as port-based teleportation (PBT)—was introduced, which removes the need for Bob's correction unitaries, but at the cost of requiring unbounded entanglement to achieve vanishing error. Let the state to be teleported $\rho \in \mathcal{L}(A_0)$ be of dimension $d$. In this protocol,\footnote{Note here that we are using the ``standard'' PBT protocol described in~\cite{pirandola2019fundamental}. }
Alice and Bob share a resource state composed of $N$ maximally entangled pairs
\begin{equation*}
\Phi^{+}_{A_i B_i} \in \mathcal{L}(A_i \otimes B_i), 
\quad i=1,\dots,N.
\end{equation*}
The total resource state is
\begin{equation*}
\Phi^{\otimes N}_{AB} := \bigotimes_{i=1}^N \Phi^{+}_{A_i B_i}.
\end{equation*}
Alice performs a joint measurement
$$M^{\mathrm{PBT}}\coloneqq \{ M_k^{\mathrm{PBT}} \}_{k=1}^N$$
on the system $A_0 A_1 \cdots A_N$.
Conditioned on the outcome $k$, the unnormalized post-measurement state on Bob's systems is
\begin{equation*}
\tilde \rho_k = \mathrm{Tr}_{A_0 A_1 \cdots A_N}
\left[
(M_k^{\mathrm{PBT}} \otimes I_{B_1 \cdots B_N})
\left(
\rho_{A_0} \otimes \Phi^{\otimes N}_{AB}
\right)
\right].
\end{equation*}
Alice then sends the classical index $k$ to Bob. The subsystems 
$\{B_j\}_{j=1}^N$ are referred to as \textbf{ports}. Bob selects the port $B_k$ and discards the remaining ports, yielding
\begin{equation*}
\rho_k 
= \mathrm{Tr}_{\bigotimes_{j\neq k} B_{j}} \left [\frac{
\tilde \rho_k
}{
\mathrm{Tr}(\tilde \rho_k)
}\right].
\end{equation*}
This protocol induces a quantum channel from $A_0$ to Bob’s system $B_k$
\begin{equation*}
\Lambda_N : \mathcal{L}(A_0) \to \mathcal{L}(B).
\end{equation*}
where all ports $B_k$ are isomorphic to a system $B$.
This channel be can be written as
\begin{equation*}
\Lambda_N(\rho)
=
\sum_{k=1}^N p_k\rho_k,
\qquad
p_k:=\mathrm{Tr}(\tilde \rho_k).
\end{equation*} 
It is known that there exists a choice of measurement 
$\{M_k^{\mathrm{PBT}}\}$, namely the pretty good measurement~\cite{PhysRevLett.101.240501,HausladenWootters1994PrettyGood}, such that the resulting 
channel $\Lambda_N$ approximates the identity channel. 
More precisely, for any fixed dimension $d$, the entanglement 
fidelity for this protocol satisfies~\cite{christandl2021asymptotic}
\begin{equation}
\label{eq:fidelity_pbt}
F_e(\Lambda_N) 
\ge 1 - \frac{d^2-1}{N}.
\end{equation}


\subsection{Approximating a multipartite channel with port-based teleportation}
\label{app:box_reduction_proof}

We will provide here a multipartite extension of port-based teleportation. We first prove a result for general quantum channels and then use it to prove a channel decomposition for a measurement channel.

Consider a quantum channel $\mathcal N$ with
\(n_{\mathrm{in}}\) input systems and
\(n_{\mathrm{out}}\) output systems. These will respectively correspond to the quantum systems held by the sender and receiver parties in the LC game. Again, when we say ``input party'' or ``output party,'' we mean the party that holds the corresponding input or output quantum system. Note that the input parties and output parties are both, possibly overlapping,\footnote{This could mean that both halves of a Bell pair could be held by the same party. To make our proof concise, we allow for this possibility. } subsets of the parties playing the LC game.
We index the input parties by
\begin{equation*}
V_{\mathrm{in}}
=
\{i_1,\dots,i_{n_{\mathrm{in}}} \},
\end{equation*}
and the output parties by
\begin{equation*}
V_{\mathrm{out}}
=
\{o_1,\dots,o_{n_{\mathrm{out}}} \}.
\end{equation*}
For each input party \(i_a\), let
\(B_{i_a}\)
be a finite-dimensional Hilbert space for the corresponding input system.
The overall input to $\mathcal N$ is 
\begin{equation*}
    B_{\mathrm{in}}
    \coloneqq
    B_{i_1} \otimes \dots \otimes B_{i_{n_{\mathrm{in}}}}.
\end{equation*}
Similarly, for each output system \(o_b\), let
\(B_{o_b}\)
be a finite-dimensional Hilbert space, and define
\begin{equation*}
    B_{\mathrm{out}}
    \coloneqq
    B_{o_1} \otimes \dots \otimes B_{o_{n_{\mathrm{out}}}}.
\end{equation*}
This is the overall output of $\mathcal N$. Here we omit one-dimensional (trivial) input and output systems since they do not play any role in the quantum  strategy.

Let \(\mathcal L(B_{\mathrm{in}})\) and
\(\mathcal L(B_{\mathrm{out}})\) denote the sets of linear operators on
\(B_{\mathrm{in}}\) and \(B_{\mathrm{out}}\), respectively.
The quantum channel $\mathcal{ N }$ is a completely positive trace-preserving (CPTP) map

\begin{equation*}
    \mathcal{ N } :
    \mathcal L(B_{\mathrm{in}})
    \rightarrow
    \mathcal L(B_{\mathrm{out}}).
\end{equation*}
Let $\rho_E \in \mathcal L(E)$ be a quantum  state shared among the input and output parties, where
\begin{equation*}
    E \coloneqq E_{i_1} \otimes \dots \otimes E_{i_{n_{\mathrm{in}}}} \otimes E' \otimes E_{o_1} \otimes \dots \otimes E_{o_{n_{\mathrm{out}}}} = E_\mathrm{in} \otimes E_\mathrm{out}.
\end{equation*}
\(E_{i_a},E',E_{o_b}\) are finite-dimensional Hilbert spaces, and each input party $i_a$ holds subsystem $E_{i_a}$, except $i_{n_\mathrm{in}}$ who holds $E_{i_{n_\mathrm{in}}}\otimes E' $. Here we defined
\begin{align*}
E_\mathrm{in} &\coloneqq E_{i_1} \otimes \dots \otimes E_{i_{n_{\mathrm{in}}} } \otimes E'\\ E_\mathrm{out} &\coloneqq E_{o_1} \otimes \dots \otimes E_{o_{n_{\mathrm{out}}}}
\end{align*}
Each output party $o_b$ holds subsystem $E_{o_b}$.

With the help of $\rho_E$,
the quantum channel $\mathcal N$ can be simulated arbitrarily well using local measurements performed by the input parties,
followed by one round of classical communication from the
input parties to the output parties, and finally local quantum
channels at the output parties.
More explicitly, we consider a quantum channel
\begin{equation*}
    \mathcal E :
    \mathcal L(B_{\mathrm{in}})
    \rightarrow
    \mathcal L(B_{\mathrm{out}})
\end{equation*}
of the form
\begin{equation}
    \mathcal E(\rho)
    =
    \sum_{\substack{x_{i_a}\\ i_a \in V_\mathrm{in}}}
    \left(
        \bigotimes_{o_b\in V_{\mathrm{out}}}
        \mathcal{N}_{o_b}^{\vec x}
    \right)
    \tr_{B_\mathrm{in} E_\mathrm{in}}\left[
    \left(
        \bigotimes_{i_a\in V_{\mathrm{in}}}
        \Lambda_{i_a,x_{i_a}} \otimes I_{E_\mathrm{out}}
    \right)
    (\rho\otimes\rho_E)\right],
    \label{eq:op_single_round}
\end{equation}
where $\{\Lambda_{i_a,x_{i_a}}\}_{x_{i_a}}$ is a POVM on $B_{i_a} \otimes E_{i_a}$, except $i_\mathrm{in}$ for whom it is on $B_{i_{n_\mathrm{in}}} \otimes E_{i_{n_{\mathrm{in}}}} \otimes E'$
and 
\begin{equation*}
    \begin{aligned}
        & \mathcal{N}_{o_b}^{\vec x}: \mathcal{L} (E_{o_b}) \rightarrow \mathcal{L} (B_{o_b})
    \end{aligned}
\end{equation*}
is a quantum channel for all
$
    \vec x
    \coloneqq
    (x_{i_a})_{i_a\in V_{\mathrm{in}}}.
$
Namely, $\Lambda_{i_a,x_{i_a}}$ is the measurement performed by the input parties, the measurement result $\vec x$
is the classical message communicated to the output parties, and
$\mathcal{N}_{o_b}^{\vec x}$ is the quantum channel applied by the output parties depending on the message received $\vec x$. This decomposition will be useful to prove the reductions we claim below. We will find $\mathcal E \approx \mathcal N$ for appropriate choices of $\rho_E, \Lambda_{i_a}, \mathcal N_{o_b}^{\vec x}$.
We illustrate this decomposition for the case of three input and three output parties in~\Cref{fig:tripartite-U-one-round}.
\begin{figure}[htbp!]
  \centering
  \begin{tikzpicture}[
    scale=.9,
    bigU/.style={
      draw,
      minimum width=2.15cm,
      minimum height=4.15cm,
      fill=gray!15,
      rounded corners=2pt,
      align=center,
      font=\normalsize
    },
    localbox/.style={
      draw,
      minimum width=1.05cm,
      minimum height=1cm,
      fill=gray!15,
      rounded corners=2pt,
      align=center,
      font=\scriptsize,
      inner sep=1pt
    }
  ]
    \pgfmathsetlengthmacro{\ystep}{1.8cm}


    \node[bigU] (U) at (0,0) {$\mathcal N$};

    \foreach \i/\y in {1/1*\ystep,2/0*\ystep,3/-1*\ystep}{
      \coordinate (Uin\i)  at ($(U.west)+(0,\y)$);
      \coordinate (Uout\i) at ($(U.east)+(0,\y)$);

      \draw ([xshift=-1.00cm]Uin\i) -- (Uin\i);

      \draw (Uout\i) -- ([xshift=.55cm]Uout\i);
    }
    \draw[decorate,decoration={brace,amplitude=5pt,mirror}]
      ($(Uin1)+(-1.35cm,0)$) -- ($(Uin3)+(-1.35cm,0)$)
      node[midway,xshift=-15pt] {$\ket{\psi}$};


    \node at (2.80cm,0) {$\approx$};


    \begin{scope}[xshift=7.15cm]

      \node[localbox] (W10) at (0,1*\ystep)  {$\Lambda_{i_1}$};
      \node[localbox] (W20) at (0,0*\ystep)  {$\Lambda_{i_2}$};
      \node[localbox] (W30) at (0,-1*\ystep) {$\Lambda_{i_3}$};

      \node[localbox] (W11) at (4.35cm,1*\ystep)  {$\mathcal N_{o_1}^{\vec x}$};
      \node[localbox] (W21) at (4.35cm,0*\ystep)  {$\mathcal N_{o_2}^{\vec x}$};
      \node[localbox] (W31) at (4.35cm,-1*\ystep) {$\mathcal N_{o_3}^{\vec x}$};

      \foreach \boxname in {W10,W20,W30,W11,W21,W31}{
        \foreach \port/\t in {1/.125,2/.375,3/.625,4/.875}{
          \coordinate (\boxname in\port) at ($(\boxname.north west)!\t!(\boxname.south west)$);
          \coordinate (\boxname out\port) at ($(\boxname.north east)!\t!(\boxname.south east)$);
        }
        \coordinate (\boxname outC) at (\boxname.east);
      }

      \coordinate (psi10) at ($(W10in1)+(-1.695cm,0)$);
      \coordinate (psi20) at ($(W20in1)+(-1.695cm,0)$);
      \coordinate (psi30) at ($(W30in1)+(-1.695cm,0)$);
      \coordinate (Ein10) at ($(W10in4)+(-.825cm,0)$);
      \coordinate (Ein20) at ($(W20in4)+(-.675cm,0)$);
      \coordinate (Ein30) at ($(W30in4)+(-.525cm,0)$);
      \draw (psi10) -- (W10in1);
      \draw (psi20) -- (W20in1);
      \draw (psi30) -- (W30in1);

      \draw[decorate,decoration={brace,amplitude=5pt,mirror}]
        ($(psi10)+(-.35cm,0)$) -- ($(psi30)+(-.35cm,0)$)
        node[midway,xshift=-15pt] {$\ket{\psi}$};

      \coordinate (Ein1collect) at (-2.12,-2.92);
      \coordinate (Ein2collect) at (-2.12,-3.25);
      \coordinate (Ein3collect) at (-2.12,-3.58);
      \coordinate (Ein1drop) at (Ein10 |- Ein1collect);
      \coordinate (Ein2drop) at (Ein20 |- Ein2collect);
      \coordinate (Ein3drop) at (Ein30 |- Ein3collect);
      \draw[rounded corners=3pt] (W10in4) -- (Ein10) -- (Ein1drop) -- (Ein1collect);
      \draw[rounded corners=3pt] (W20in4) -- (Ein20) -- (Ein2drop) -- (Ein2collect);
      \draw[rounded corners=3pt] (W30in4) -- (Ein30) -- (Ein3drop) -- (Ein3collect);
      \draw[decorate,decoration={brace,amplitude=4pt,mirror}]
        (-2.22,-2.87) -- (-2.22,-3.63)
        node[midway,xshift=-23pt,font=\scriptsize] {$E_\mathrm{in}$};

      \coordinate (Eout11) at (W11in4);
      \coordinate (Eout21) at (W21in4);
      \coordinate (Eout31) at (W31in4);
      \coordinate (Eout1turn) at ($(Eout11)+(-.475cm,0)$);
      \coordinate (Eout2turn) at ($(Eout21)+(-.345cm,0)$);
      \coordinate (Eout3turn) at ($(Eout31)+(-.215cm,0)$);
      \coordinate (Eout1collect) at (2.80,-2.92);
      \coordinate (Eout2collect) at (2.80,-3.25);
      \coordinate (Eout3collect) at (2.80,-3.58);
      \coordinate (Eout1drop) at (Eout1turn |- Eout1collect);
      \coordinate (Eout2drop) at (Eout2turn |- Eout2collect);
      \coordinate (Eout3drop) at (Eout3turn |- Eout3collect);
      \draw[rounded corners=3pt] (Eout11) -- (Eout1turn) -- (Eout1drop) -- (Eout1collect);
      \draw[rounded corners=3pt] (Eout21) -- (Eout2turn) -- (Eout2drop) -- (Eout2collect);
      \draw[rounded corners=3pt] (Eout31) -- (Eout3turn) -- (Eout3drop) -- (Eout3collect);
      \draw[decorate,decoration={brace,amplitude=4pt,mirror}]
        (2.70,-2.87) -- (2.70,-3.63)
        node[midway,xshift=-23pt,font=\scriptsize] {$E_\mathrm{out}$};

      \draw (W11outC) -- ([xshift=.80cm]W11outC);
      \draw (W21outC) -- ([xshift=.80cm]W21outC);
      \draw (W31outC) -- ([xshift=.80cm]W31outC);

      \draw[wire/classical] (W10out1) -- (W11in1);
      \draw[wire/classical] (W20out2) -- (W21in2);
      \draw[wire/classical] (W30out3) -- (W31in3);

      \draw[wire/classical] (W10out2) to[out=0,in=180] (W21in1);
      \draw[wire/classical] (W10out3) to[out=0,in=180] (W31in1);

      \draw[wire/classical] (W20out1) to[out=0,in=180] (W11in2);
      \draw[wire/classical] (W20out3) to[out=0,in=180] (W31in2);

      \draw[wire/classical] (W30out1) to[out=0,in=180] (W11in3);
      \draw[wire/classical] (W30out2) to[out=0,in=180] (W21in3);

      \node at (2.175,-2.35) {\scriptsize $\vec x$};
    \end{scope}

  \end{tikzpicture}
  \caption{A tripartite quantum channel $\mathcal N$ can be approximated arbitrarily well with an ancillary entangled state $\rho_E =\rho_{E_\mathrm{in} E_\mathrm{out}}$, local measurements $\Lambda_{i_a}$, a single round of classical communication of measurement results $\vec x$ between every pair of parties, and post-processing channels $\mathcal N_{o_b}^{\vec x}$ that depend on $\vec x$.}
  \label{fig:tripartite-U-one-round}
\end{figure}

Let $d_\mathrm{in} = \dim (B_{\mathrm{in}})$ be the total input dimension of $\mathcal{N}$, and $d_\mathrm{out} = \dim (B_{\mathrm{out}})$ be the total output dimension. We can prove the following. 
\begin{lem}\label{lemma:box_reduction}
    Let $n_\mathrm{in}, n_\mathrm{out} \in \mathbb{Z}^+$ and $n_\mathrm{in} \geq 2$. For any quantum channel $\mathcal{N}$ with $n_\mathrm{in}$ input systems and $n_\mathrm{out}$ output systems as defined above, and any $\varepsilon >0$, there exists a quantum state $\rho_E$ on a Hilbert space of dimension  $d_\varepsilon$, where  
\begin{equation*}
\begin{aligned}
&d_\varepsilon
\le d_{\mathrm{in}}^{2 + 2\sum_{a=1}^{n_\mathrm{in}-1} \left(\prod_{r=1}^{a} n_r \right) } d_{\mathrm{out}}^{2 \prod_{a=1}^{n_{\mathrm{in}}-1} n_a},\\
    &n_a = 2(n_{\mathrm{in}}-1)d_{\mathrm{in}}^{2\prod_{r=1}^{a-1} n_r } m_\varepsilon,
\end{aligned}
\end{equation*}
the empty product being equal to 1 and $m_\varepsilon \coloneqq \lceil 1/\varepsilon\rceil$,
as well as POVMs $\{\Lambda_{i_a,x_{i_a}}\}_{x_{i_a}}$ for all $i_a\in V_\mathrm{in}$, and quantum channels $\{\mathcal{N}_{o_b}^{\vec{x}}\}_{o_b \in V_\mathrm{out},\vec x = (x_{i_a})_{i_a \in V_\mathrm{in}}}$, such that
    \begin{equation*}
       \left\| \mathcal{E} - \mathcal{N} \right\|_\diamond < \varepsilon,
    \end{equation*}
$\mathcal E$ being defined by
\begin{equation*}
        \mathcal{E}(\rho) = \sum_{\substack{x_{i_a} \\ i_a \in V_\mathrm{in}}} 
    \left(
    \bigotimes_{o_b\in V_\mathrm{out}} \mathcal{N}_{o_b}^{\vec x}
    \right)
    \tr_{B_\mathrm{in} E_\mathrm{in}}\left[
    \left(
    \bigotimes_{i_a\in V_\mathrm{in}}  \Lambda_{i_a,x_{i_a}} \otimes I_{E_\mathrm{out}}
    \right) (\rho \otimes \rho_E)\right].
    \end{equation*}
\end{lem}

\begin{proof}
    Let 
    \begin{equation*}
        |\Phi^+_d\rangle = \frac{1}{\sqrt{d}} \sum_{i=0}^{d-1} |ii\rangle 
    \end{equation*}
    be a maximally entangled state with local dimension $d$. 
    Now we give the explicit ancillary entangled state $\rho_E$ used to approximate $\mathcal{N}$. First, each input party $i_a \in V_{\mathrm{in}}$ ($a \ne 1$) shares a maximally entangled state $|\Phi_{d_{i_a}}^+\rangle$ with input party $i_1$, where $d_{i_a} \coloneqq \dim (B_{i_a})$. This leads to an overall state
    \begin{equation*}
\rho_1^{\mathrm{TP}}
=
\bigotimes_{a=2}^{n_{\mathrm{in}}}
\left(
|\Phi_{d_{i_a}}^+\rangle
\langle\Phi_{d_{i_a}}^+|
\right)_{i_1,i_a} .
    \end{equation*}
    Meanwhile, for $1\le a \le n_{\mathrm{in}}-1$, 
    input party $i_a$ and $i_{a+1}$ share $n_{a}$ (given in the lemma statement) pairs of maximally entangled state $|\Phi_{d_a}^+\rangle$, denoted by
    \begin{equation*}
        \rho^{i_{a}\rightarrow i_{a+1}} = \left( |\Phi_{d_a}^+\rangle \langle\Phi_{d_a}^+|^{\otimes  n_a } \right)_{i_a,i_{a+1}},
    \end{equation*}
    where the $d_a$ is defined by  
    \begin{equation*}
    \begin{aligned}
        d_a \coloneqq d_{\mathrm{in}}^{\prod_{r=1}^{a-1} n_r}.
    \end{aligned}
    \end{equation*}
    This defines the quantum state on the ancillary system $\bigotimes_{a=1}^{n_\mathrm{in}} E_{i_a}$:
    $$\rho^{\mathrm{TP}}_1
    \otimes
    \left(\bigotimes_{a=1}^{n_{\mathrm{in}}-1}
    \rho^{i_a\rightarrow i_{a+1}}
    \right) \in \mathcal L(\bigotimes_{a=1}^{n_\mathrm{in}} E_{i_a})$$
    Each output party $o_b \in V_{\mathrm{out}}$ shares $\prod_{a=1}^{n_{\mathrm{in}}-1} n_a$ pairs of 
    maximally entangled states $|\Phi_{d_{o_b}}^+\rangle$ with input party $i_{n_{\mathrm{in}}}$, that is
    \begin{equation*}
\rho^{\mathrm{TP}}_2
=
\bigotimes_{b=1}^{n_{\mathrm{out}}}
\left(
|\Phi_{d_{o_b}}^+\rangle
\langle\Phi_{d_{o_b}}^+|^{\otimes \left(\prod_{a=1}^{n_{\mathrm{in}}-1} n_a\right)}
\right)_{i_{n_{\mathrm{in}}},o_b} \in\mathcal L(E'\otimes E_\mathrm{out}),
    \end{equation*}
    where $d_{o_b} \coloneqq \dim (B_{o_b})$.
This is the quantum state on the ancillary system $E' \otimes E_\mathrm{out}$.
The overall ancillary entangled resource state is given by
\begin{equation*}
    \rho_E
    =
    \rho^{\mathrm{TP}}_1
    \otimes
    \left(\bigotimes_{a=1}^{n_{\mathrm{in}}-1}
    \rho^{i_a\rightarrow i_{a+1}}
    \right)
    \otimes
    \rho^{\mathrm{TP}}_2
    .
\end{equation*}
    
    Let $\ket{\psi} \in B_{\mathrm{in}}$ be the input state of $\mathcal{N}$. The protocol proceeds as follows:
\begin{enumerate}[label=Step~\arabic*.]
    \item {[}Teleportation to first party, without applying correction unitary{]}\\
    For each $a\neq 1$, input party $i_a$ performs a standard teleportation measurement $M^{\mathrm{TP}}_{i_a\rightarrow i_{1}}$ on subsystem $B_{i_a}$ of $\vert \psi\rangle$ and his share of $\rho^{\mathrm{in}}$.
    Let $t_{i_a}$ denote the outcome of $M^{\mathrm{TP}}_{i_a\rightarrow i_{1}}$. 
    After all such measurements, party $i_1$ holds a state
    \begin{equation*}
        \ket{\psi_1} 
        = \left(I_{B_{i_1}} \otimes \left(\bigotimes_{a=2}^{n_{\mathrm{in}}} \sigma_{i_a}(t_{i_a})\right)^\dagger \right)\ket{\psi},
    \end{equation*}
    where $\sigma_{i_a}(t_{i_a})$ are the $t_{i_a}$-dependent correction unitaries on the corresponding subsystem. We denote the Hilbert space of
    $\ket{\psi_1}$ by $\mathcal H_1$.

    \item {[}PBT from party $i_{1}$ to $i_2$, without discarding other ports{]}\\
    Party $i_1$ performs port-based teleportation measurement $M^{\mathrm{PBT}}_{i_1 \rightarrow i_2}$ with $n_1$ ports on system $\mathcal H_1$ and his part of the entangled state $\rho^{i_1 \rightarrow i_2}$. 
    The measurement result of $M^{\mathrm{PBT}}_{i_1 \rightarrow i_2}$ is $m_1 \in [n_1]$.  
Party $i_2$ now holds $n_1$ ports. Denote by
$\mathcal H_2$ the Hilbert space associated with these
ports and denote the joint state of all the ports by 
$|\psi_2\rangle \in \mathcal H_2$.
In particular, the reduced state of the $m_1$-th port is
\begin{equation*}
    \rho_2^{(m_1)}
    =
    \operatorname{tr}_{\overline{m_1}}
    \bigl(
        |\psi_2\rangle\langle\psi_2|
    \bigr),
\end{equation*}
where $\operatorname{tr}_{\overline{m_1}}$
denotes the partial trace over all ports except the
$m_1$-th port.
    By~\Cref{eq:fidelity_pbt}, the protocol implements the identity channel as $n_1 \to \infty$, that is 
    \begin{equation*}
        \lim_{\substack{n_1 \to \infty}}  \bra{\psi_1}\rho_{2}^{(m_1)} \ket{\psi_1}  =1.
    \end{equation*}

    \item {[}Apply correction unitary{]} \\
    Party $i_2$ applies the correction unitary 
$ \sigma_{i_2}(t_{i_2})$ for all $n_1$ ports and the resulting state is denoted as $|\psi^c_2\rangle$ ($c$ stands for `corrected'). 
The selected port $m_1$ now contains a state 
that approximates the state 
    \begin{equation*}
        \left(I_{B_{i_1}}  \otimes I_{B_{i_2}} \otimes \left( \bigotimes_{a=3}^{n_{\mathrm{in}}} \sigma_{i_a}(t_{i_a})\right)^\dagger  \right) \ket{\psi}.
    \end{equation*}
    Thus, compared to $\ket{\psi_1}$, $\ket{\psi^c_2}$ has one less correction unitary needed. 

    \item {[}Iterated port-based teleportation to apply all correction unitaries{]}\\
When $n_{\mathrm{in}} \ge 3$, for $2\le a \le n_{\mathrm{in}}-1$, party $i_{a}$ performs a port-based teleportation (PBT) measurement
$M_{i_{a}\to i_{a+1}}^{\mathrm{PBT}}$
with $n_a$ ports on the state $\ket{\psi_{a}^c}$ and his share of the entangled state
$\rho^{i_{a}\to i_{a+1}}$, obtaining an outcome
$m_a\in[n_a]$.
Consequently, party $i_{a+1}$ holds
$\prod_{r=1}^{a} n_r$
ports, whose joint Hilbert space is denoted by
$\mathcal H_{a+1}$.
Let
$\ket{\psi_{a+1}}\in\mathcal H_{a+1}$
be the state stored in these ports.
We further define the multi-index 
\begin{equation*}
    \vec l^{(a)}
\coloneqq
(l_1,\ldots,l_a)
\in
[n_1]\times\cdots\times[n_a],
\end{equation*}
which labels the ports held by party $i_{a+1}$.
    Party $i_{a+1}$ applies the correction unitary 
        $\sigma_{i_{a+1}}(t_{i_{a+1}})$
    for all his ports, and obtains the state $\ket{\psi^c_{a+1}}$ in these ports.
    The reduced state in the port $\vec{m}^{(a)} \coloneqq (m_1,\ldots,m_a)$ approximates the state
    \begin{equation*}
        \left(I_{B_{i_1}} \otimes\cdots \otimes I_{B_{i_{a+1}}} \otimes \vec{\sigma}_{a+2}^\dagger(\vec{t}_{a+2}) \right)
        \ket{\psi}
    \end{equation*}
    as all $n_a \to \infty$, where 
    \begin{equation*}
        \vec{\sigma}_{a+2}(\vec{t}_{a+2})=
        \begin{cases}
            \sigma_{i_{a+2}}(t_{i_{a+2}})\otimes \dots \otimes \sigma_{i_{n_{\mathrm{in}}}}(t_{i_{n_\mathrm{in}}}) ,& a\le n_{\mathrm{in}}-2 \\
            1, &a=n_{\mathrm{in}}-1
        \end{cases}.
    \end{equation*}
    At each step, one additional correction operation $\sigma_{i_{a+1}}(t_{i_{a+1}})$ is applied. 
    

    \item {[}Applying the quantum channel{]}\\
    Party $i_{n_{\mathrm{in}}}$ now holds a state $\ket{\psi^c_{n_{\mathrm{in}}}}$ in a system consisting of 
    $\prod_{a=1}^{n_{\mathrm{in}}-1} n_a$ ports. On the port labeled by $\vec{m}^{(n_{\mathrm{in}}-1)}\coloneqq  (m_1,m_2,\dots,m_{n_{\mathrm{in}}-1})$, the reduced state is 
\begin{equation*}
    \rho_{\vec{m}^{(n_{\mathrm{in}}-1)}}
    =
    \operatorname{tr}_{\overline{\vec{m}^{(n_{\mathrm{in}}-1)}}}
    \left(
        |\psi^c_{n_{\mathrm{in}}}\rangle
        \langle\psi^c_{n_{\mathrm{in}}}|
    \right),
\end{equation*}
where $\operatorname{tr}_{\overline{\vec{m}^{(n_{\mathrm{in}}-1)}}}$ denotes the partial trace over all ports except the
port labeled by $\vec{m}^{(n_{\mathrm{in}}-1)}$.
    This state approximates the input state $\ket{\psi}$, i.e.,
    \begin{equation*}
        \lim_{n_1,\dots,n_{n_{\mathrm{in}}-1} \to \infty} 
        \langle \psi | \rho_{\vec{m}^{(n_{\mathrm{in}}-1)}}  | \psi \rangle = 1.
    \end{equation*}
    Party $i_{n_{\mathrm{in}}}$ applies the quantum channel $\mathcal{N}$ independently to each port. 

    \item {[}Teleportation to output parties, without applying correction unitaries{]}\\
    For every port $\vec{l}^{(n_{\mathrm{in}}-1)}$ of party $i_{n_{\mathrm{in}}}$ and output party $o_b\in V_{\mathrm{out}}$, party $i_{n_{\mathrm{in}}}$ performs standard teleportation measurements $M_{i_{n_{\mathrm{in}}} \rightarrow o_b}^{\mathrm{TP}}$ on the subsystem corresponding to $B_{o_b}$ in this port and his part of one copy of $|\Phi_{d_{o_b}}^+\rangle_{i_{n_{\mathrm{in}}},o_b}$ in $\rho^{\mathrm{out}}$.
    The measurement results are $t_{o_b}(\vec{l}^{(n_{\mathrm{in}}-1)})$. 
    On the half of the entangled state corresponding to $\vec{m}^{(n_{\mathrm{in}}-1)}$ for every output party, the overall state is 
    \begin{equation*}
        \left( \bigotimes_{o_b\in V_{\mathrm{out}}} \sigma_{o_b}(t_{o_b}(\vec{m}^{(n_{\mathrm{in}}-1)})) \right)^\dagger \mathcal N(\rho_{\vec{m}^{(n_{\mathrm{in}}-1)}}) \left( \bigotimes_{o_b\in V_{\mathrm{out}}} \sigma_{o_b}(t_{o_b}(\vec{m}^{(n_{\mathrm{in}}-1)})) \right)
    \end{equation*}
    with $\sigma_{o_b}(t_{o_b}(\vec{m}^{(n_{\mathrm{in}}-1)}))$ being the  
    correction unitary corresponding to measurement result $t_{o_b}(\vec{m}^{(n_{\mathrm{in}}-1)})$. 

    \item {[}Classical communication and correction{]}\\
    Finally, input party $i_{n_{\mathrm{in}}}$ sends $t_{o_b}(\vec{l}^{(n_{\mathrm{in}}-1)})$ to output party $o_b$ for every $\vec l^{(n_\mathrm{in}-1)}$. 
    Every other input party $i_a$ sends $m_a$ to all output parties. After communication, each output party $o_b$ discards all ports except the one labeled by $\vec{m}^{(n_{\mathrm{in}}-1)}$, and applies $\sigma_{o_b}(t_{o_b}(\vec{m}^{(n_{\mathrm{in}}-1)}))$ to recover the overall state $\mathcal{N}(\rho_{\vec{m}^{(n_{\mathrm{in}}-1)}})$. 
\end{enumerate}
In the limit of large 
numbers of ports $n_a$, the protocol approximates the ideal quantum channel $\mathcal N$ arbitrarily well.
Importantly, \emph{all measurements and correction operations from Step 1 to Step 6 are local operations on different subsystems performed by the input parties that can be implemented in parallel}. Indeed, letting $\mathcal{E}(\cdot)$ be the overall quantum channel describing the above protocol, we see that it has the required form. $\Lambda_{i_a,x_{i_a}}$ is the effective POVM induced by
\begin{itemize}
    \item the teleportation measurements from other input parties to input party $i_1$,
    \item the $n_{\mathrm{in}}-1$ port-based teleportation measurements corresponding to a PBT protocol from party $i_{a}$ to party $i_{a+1}$,
    \item the correction unitaries applied by each party, 
    \item the quantum channel $\mathcal N$ performed by party $i_{n_\mathrm{in}}$ on every port,
    \item and the teleportation measurements from party $i_{n_\mathrm{in}}$ to the output parties.
\end{itemize}
 $\vec x$ corresponds to the measurement outcomes $\vec m^{(n_\mathrm{in}-1)}$ and $t_{o_b}(\vec l^{(n_\mathrm{in}-1)})$. The quantum channels $\mathcal N_{o_b}^{\vec x}$ are the correction unitaries and the tracing out of irrelevant ports implemented by the output parties that depend on the measurement outcomes $\vec x$. 

Let $\Lambda^{i_a\to i_{a+1}}$ be the channel induced by the PBT from party $i_a$ to $i_{a+1}$. 
The input state of channels $\Lambda_{a}^{i_{a} \rightarrow i_{a+1}}$ is in the Hilbert space $\mathcal{H}_{a}$ of dimension 
\begin{equation*}
    d_a = d_{\mathrm{in}}^{\prod_{r=1}^{a-1} n_r }
\end{equation*}
by construction. Let $\varepsilon > 0$. According to \Cref{eq:fidelity_pbt}, let 
$$\varepsilon'\coloneqq \frac{\varepsilon}{2(n_{\mathrm{in}}-1)},\, n_a =  2(n_{\mathrm{in}}-1)d_a^2 m_{\varepsilon}$$
such that the PBT channel $\Lambda_{a}^{i_{a}\rightarrow i_{a+1}}$ with $n_a$ ports satisfies
\begin{equation*}
    F_e(\Lambda_{a}^{i_{a}\rightarrow i_{a+1}}) 
\ge 1 - \frac{d_{a}^2-1}{n_a} > 1-\varepsilon'.
\end{equation*}
Hence, by the relation between entanglement fidelity and diamond norm for the port-based teleportation quantum channel~\cite{pirandola2019fundamental,christandl2021asymptotic}, we have
\begin{equation*}
            \left \| \Lambda_{a}^{i_{a}\rightarrow i_{a+1}} - \mathcal{I}^{i_{a}\rightarrow i_{a+1}}\right\|_\diamond = 2\left(1-F_e(\Lambda_{a}^{i_a\rightarrow i_{a+1}}) \right) < 2\varepsilon'.
\end{equation*}
By the triangle inequality and contractivity of the diamond norm under composition~\cite{watrous2018theory},
for the channel $\Lambda_{1}^{i_1\rightarrow i_2} \circ \Lambda_{2}^{i_2\rightarrow i_3} \circ \cdots \circ \Lambda_{n_{\mathrm{in}}-1}^{i_{n_{\mathrm{in}}-1}\rightarrow i_{n_{\mathrm{in}}}}$, we have
\begin{equation*}
                    \left \| \Lambda_{1}^{i_1\rightarrow i_2} \circ \Lambda_{2}^{i_2\rightarrow i_3} \circ \cdots \circ \Lambda_{n_{\mathrm{in}}-1}^{i_{n_{\mathrm{in}}-1}\rightarrow i_{n_{\mathrm{in}}}} - \mathcal I \right\|_\diamond   
                   \le \sum_{a=1}^{n_{\mathrm{in}}-1} \left \| \Lambda_{a}^{i_{a}\rightarrow i_{a+1}} - \mathcal{I}^{i_{a}\rightarrow i_{a+1}} \right\|_\diamond 
                    < 2(n_{\mathrm{in}}-1)\varepsilon'\coloneqq \varepsilon. 
\end{equation*}
In the first inequality we use an appropriate telescoping argument.

As the remaining channels in $\mathcal E$ are not approximate, the overall diamond norm satisfies
\begin{equation*}
    \| \mathcal E -\mathcal N \|_\diamond \leq \left \| \Lambda_{1}^{i_1\rightarrow i_2} \circ \Lambda_{2}^{i_2\rightarrow i_3} \circ \cdots \circ \Lambda_{n_{\mathrm{in}}-1}^{i_{n_{\mathrm{in}}-1}\rightarrow i_{n_{\mathrm{in}}}} - \mathcal I \right\|_\diamond   < \varepsilon.
\end{equation*}
Moreover, in Step~1 and Step~6,  the ancillary entangled state $\rho^\mathrm{TP}_1$ is of dimension 
\begin{equation*}
     \left(\frac{d_{\mathrm{in}}}{d_{B_{i_1}}} \right)^2 \le d_{\mathrm{in}}^2,
\end{equation*}
where $d_{B_{i_1}} \ge 1$ is the dimension of subsystem $B_{i_1}$, and   $\rho^\mathrm{TP}_2$ of dimension $d_{\mathrm{out}}^{2 \prod_{a=1}^{n_{\mathrm{in}}-1} n_a}$, respectively. Thus the dimension of all ancillary systems is 
\begin{equation*}
d_\varepsilon
\le d_{\mathrm{in}}^{2 + 2\sum_{a=1}^{n_\mathrm{in}-1} \left(\prod_{r=1}^{a} n_r \right) } d_{\mathrm{out}}^{2 \prod_{a=1}^{n_{\mathrm{in}}-1} n_a}
\end{equation*}
This completes the proof.
    
\end{proof}

We can prove~\Cref{lem:measurement_channel_reduction} as a special case of~\Cref{lemma:box_reduction} where the quantum channel is a measurement channel $\mathcal M$. We restate the lemma here for convenience. 

\noindent \textbf{Lemma~\ref{lem:measurement_channel_reduction}.}\ %
\textit{
Let $n_\mathrm{in} \in \mathbb{Z}^+$ and $n_\mathrm{in} \geq 2$. For any joint measurement channel $\mathcal{M}$ on $n_\mathrm{in}$ input parties as defined above, and any $\varepsilon >0$, there exists a quantum state $\rho_E$ defined on a Hilbert space of dimension $d_\varepsilon$, 
    with
\begin{equation*}
\begin{aligned}
&d_\varepsilon
\le d_{\mathrm{in}}^{2 + 2\sum_{a=1}^{n_\mathrm{in}-1} \left(\prod_{r=1}^{a} n_r \right) },\\
    &n_a = 2(n_{\mathrm{in}}-1)d_{\mathrm{in}}^{2\prod_{r=1}^{a-1} n_r }m_{\varepsilon}, 
\end{aligned}
\end{equation*}
where the empty product is equal to $1$ and $m_\varepsilon \coloneqq \lceil 1/\varepsilon\rceil$, as well as
POVMs $\{M_{i_a,x_{i_a}}\}_{x_{i_a}}$, and a function $f$, 
such that 
\begin{equation*}
    \| \mathcal{E}_M - \mathcal{M}\|_\diamond <\varepsilon, 
\end{equation*}
$\mathcal E_M$ being defined by
\begin{equation*}
        \mathcal{E}_M(\rho) \coloneqq  \sum_{\substack{x_{i_a} \\ i_a \in V_\mathrm{in}}}   
    \tr \left[\left(
    \bigotimes_{i_a\in V_\mathrm{in}}   M_{i_a,x_{i_a}}
    \right) (\rho \otimes \rho_E)\right]\vert  f(\vec x) \rangle\langle f(\vec x)\vert
\end{equation*}
and $\vec x \coloneqq (x_{i_a})_{i_a\in V_{\mathrm{in}}}.$
}
\begin{proof}
The proof follows the same construction as in Lemma~\ref{lemma:box_reduction}. The only difference is that a measurement channel has no quantum output systems. Therefore, the standard teleportation step to the output parties is unnecessary, and the ancillary state $\rho^{\mathrm{out}}$ can be omitted.
The dimension of the ancillary entangled system is reduced to 
\begin{equation*} 
    d_\varepsilon \le d_{\mathrm{in}}^{2 + 2\sum_{a=1}^{n_\mathrm{in}-1} \left(\prod_{r=1}^{a} n_r \right) }. 
\end{equation*}
The function $f$ takes as input the PBT port indices and the measurement outcomes of $\mathcal M$ implemented by $i_{n_\mathrm{in}}$ on all candidate ports, and returns the measurement outcome associated with the selected port.
\end{proof}


\subsection{Proof of Theorem \ref{thm:siso_to_LC}}
\label{app:siso_to_LC_proof}
For completeness, we restate the theorem before giving its proof. 

\noindent \textbf{Theorem~\ref{thm:siso_to_LC}.}\ %
\textit{
    Let $\tau \geq 2$ and let $(\mathcal{V},\pi,\ell)$ be a SISO, $\tau$-step LC game with $n$ parties that is path-consistent.
    Let $G$ be the connectivity graph. 
    Let $\mathcal Q$ be the set of quantum behaviors of the SISO game and $\mathcal Q_\mathrm{sim}$ that of the induced simple LC game. Then,
    \begin{align*}
        \mathcal Q_\mathrm{sim} \subseteq \mathcal Q \subseteq \overline{\mathcal Q_\mathrm{sim}}.
    \end{align*}
}

\begin{proof}



First, consider a quantum strategy for the $\tau$-step SISO game $(\mathcal{V},\pi,\ell)$. Let $\varepsilon >0$.
    By iterative application of Lemma~\ref{lemma:box_reduction}, each intermediate interaction tensor, which are just isometries from the input quantum systems to the output quantum systems since the classical legs are trivial, can be replaced\footnote{For intermediate interaction tensors with only one input party, we can directly move it to the interaction tensor that connects to it from the past. } with local operations and direct communication with error $\frac {\varepsilon}{N_\text{int} \prod_{i=1}^n \vert S_i\vert}$, where $N_\text{int}$ is the number of intermediate interaction tensors. The POVMs in~\Cref{eq:op_single_round} can be absorbed into the preceding interaction tensors and the quantum channels in~\Cref{eq:op_single_round} can be absorbed into the succeeding interaction tensors via Naimark or Stinespring dilation.
    After all intermediate interaction tensors are eliminated this way, we are left with a quantum channel $\mathcal E_\textrm{int}$, which involves only local operations of each party at time $t=0$, a single round of communication, and local operations of each party after communication, such that 
    \begin{equation*}
        \left \| \mathcal E_\mathrm{int}-\mathcal{N} \right\|_\diamond < N_\text{int} \frac{\varepsilon}{N_\text{int}\prod_{i=1}^n \vert S_i\vert} =  \frac{\varepsilon}{\prod_{i=1}^n \vert S_i\vert},
    \end{equation*}
    where $\mathcal{N}$ is the overall quantum channel induced by all the intermediate interaction tensors and the inequality follows from the triangle inequality for the diamond norm and using an appropriate telescoping series.
    
    We claim the initial interaction tensors $W_i^{(0)}$, the quantum channel $\mathcal E_\mathrm{int}$ and final measurements $M_i$ together constitute a quantum strategy for the simple LC game $(\mathcal{V},\pi,G)$. This is because we can absorb the POVMs in $\mathcal E_\mathrm{int}$ into the initial interaction tensors and the quantum channel into the measurements via Naimark dilation. Also, the communication involved in this strategy can only be from party $i$ to party $j$ when $(i,j)$ is an edge of the connectivity graph $G$ because the latency function $\ell$ satisfies the triangle inequality. That is, if party $i$ can communicate to party $j$ via an intermediate party $k$ within $\tau$ time steps,
    \begin{align*}
        \ell(i,j) \leq \ell(i,k) +\ell(k,j) \leq \tau.
    \end{align*}
    Hence, it is a valid quantum strategy for the simple LC game.
    
    Let $\rho_{\mathcal{N}}(\mathbf{s})$ and $\rho_{\mathcal E_\text{int}}(\mathbf{s})$ be the states before the measurements of the two quantum strategies for the SISO game involving $\mathcal{N}$ and $\mathcal E_\mathrm{int}$ respectively. Let $p_\mathcal{N}, p_{\mathcal E_\mathrm{int}}$ be the realized behaviors. By the monotonicity of the trace distance under measurements and the definition of diamond norm, we have 
    \begin{equation*}
        \begin{aligned}
            \sum_{\mathbf{a},\mathbf s} |p_\mathcal{N}(\mathbf a|\mathbf s) -p_{\mathcal E_\mathrm{int}}(\mathbf a|\mathbf s) |   &\le  \sum_{\mathbf{s}}  \|\rho_{\mathcal{N}}(\mathbf{s}) - \rho_{\mathcal E_\mathrm{int}}(\mathbf{s}) \|_1  \\
            &=\prod_{i=1}^n \vert S_i\vert\left   \|   \left(   
        \mathcal{E}_\mathrm{int} - \mathcal{N} \right) \circ \bigotimes_{i=1}^n \mathcal W_i^{(0)} \left(|\mathbf{s}\rangle \langle \mathbf{s}| \otimes \vert \psi\rangle\langle \psi \vert\right)
         \right  \|_1    \\
        &\le     \prod_{i=1}^n \vert S_i\vert\left\|    
        \mathcal{E}_\mathrm{int}  - \mathcal{N} 
        \right\|_\diamond\\
        & < \varepsilon,
        \end{aligned}
    \end{equation*}
    where $\mathcal W_i^{(0)}$ is the quantum channel associated with the initial interaction tensor $W_i^{(0)}$. Note here we are using the convention for multi-step LC games by separating the input dependence from the isometry $W_i$.
    
    Conversely, any quantum strategy for the simple LC game $(\mathcal{V},\pi,G)$
    can be trivially implemented as a quantum strategy for the SISO game by only using legs between the initial $t=0$ interaction tensors and the final $t=\tau$ interaction tensors. 
    
    Hence,
    \begin{align*}
        \mathcal Q_\mathrm{sim} \subseteq \mathcal Q \subseteq \overline{\mathcal Q_\mathrm{sim}}.
    \end{align*}
\end{proof}

\section{Extended XOR Games}
\label{app:extendedXOR}




In this Appendix, we examine whether the separation of~\Cref{thm:q_neq_f} between forwarding behaviors and quantum behaviors observed for the extended CHSH game can occur for other LC games. In particular, can we obtain such a separation by extending other nonlocal games?

For any $n, m \in \mathbb{Z^+}$, an extended XOR game is defined by a Boolean function $f: S_2 \times S_3\to \{0, 1\}$ with $S_2 = \{0, \dots, n-1\}$ and $S_3 = \{0,\dots, m-1\}$ and a uniform prior $\pi$ on the inputs. The first party receives a fixed input $S_1=\{0\}$, while the two others receive an input from $S_2$ and $S_3$ respectively. They all produce an output in $A_1, A_2, A_3 =\{0,1\}$. The probabilistic predicate is given by
\begin{equation*}
    \mathcal{V}(a_1,a_2,a_3 \vert s_1,s_2,s_3) = 
    \begin{cases}
        1 & a_1=a_2 ,\, f(s_2, s_3) = a_1\oplus a_3\\
        0 & \text{otherwise.}
    \end{cases}
\end{equation*}
We consider such games on the connectivity graph $G$ given by
\begin{equation}
\label{eq:graph}
    v_1 \leftrightarrows v_2 \quad v_3,
\end{equation}

To find a separation between forwarding and quantum strategies, the upper bounds on $\omega_f$ are obtained using the NPA hierarchy on the transformed nonlocal game as given in \Cref{prp:forwarding_is_q}. The lower bounds on $\omega_q$ are obtained for each game with the see-saw algorithm described in~\Cref{subsec:seesaw}, with a two-dimensional quantum channel between $v_1$ and $v_2$, and a four-dimensional quantum system shared between $v_2$ and $v_3$, while $v_1$ initially does not hold any quantum system. Since each semidefinite program converged with a duality gap well below $1\mathrm{e}{-6}$, any separation of the order of $1\mathrm{e}{-3}$ translates to a meaningful separation between the different type of strategies.

By considering $50$ extended XOR games for Boolean functions selected uniformly at random with $n=m=3$, we find a separation for $32$ games.\footnote{All the code and numerical solutions are accessible on the online repository \url{https://github.com/pierrepocreau/Latency-constrained-games}.} This confirms that the separation we observed in the extended CHSH game applies naturally to a wide range of games. We also observe that for these 50 random extended XOR games, the upper bound on $\omega_f$ obtained via the NPA are all reached by classical solutions, just as in the case of the extended CHSH game in~\Cref{subsec:forwarding_strategies}. This may be because XOR games with a quantum-classical separation and a unique optimal quantum strategy (up to isometry)~\cite{miller2013Optimal} may be sufficient for this to hold via a proof similar to that of~\Cref{prop:sep_forwarding_quantum}. 

These results are summarized in Table~\ref{table:results}, where the ID is obtained by taking the truth table of a Boolean function $f: S_2 \times S_3 \to \{0, 1\}$ defining an extended XOR game and interpreting it as a binary representation. More explicitly, writing $S_2 = \{0, \dots, n-1\}$ and $S_3 = \{0, \dots, m-1\}$, a Boolean function is defined by a length $n \cdot  m$ truth table, such that for any $(i,j) \in S_2 \times S_3$, $f(i,j)$ is evaluated to the bit at position $m \cdot  i + j$. For example, for $n=m=3$ the truth table $[0, 0, 0, 0, 0, 0, 1, 0, 0]$ corresponds to the ID $4$. 

Each value is taken as the best among that obtained for $20$ random initial strategies and the semidefinite programs are solved with MOSEK~\cite{mosek}. We also include the classical values $\omega_c$ and $\omega_c(\emptyset)$ of each game, with $\omega_c(\emptyset)$ being the classical value obtained on the empty graph
\begin{equation*}
    v_1 \quad v_2 \quad v_3.
\end{equation*}
$\omega_c, \omega_c(\emptyset)$ are obtained by enumerating all possible classical strategies. We can observe in Table~\ref{table:results} that the best classical strategies on the graph $G$ always match the upper-bound $\omega_f$, and thus that the best forwarding strategies for such games are essentially classical. 

\captionsetup{justification=centering}
\setlength{\LTcapwidth}{\textwidth}
\begin{longtable}{|c|c|c|c|c|c|}
\caption{Classical values $\omega_c(\emptyset)$ and $\omega_c$ as well as lower and upper bounds on $\omega_q$ and $\omega_f$ for 50 random extended XOR games, with $n=m=3$. ``Gap'' refers to the difference between the lower bound on $\omega_q$ and the upper bound on $\omega_f$. Small negative gaps are numerical artifacts and should be interpreted as zero within solver precision. }

\label{table:results} \\
\hline
\endfirsthead

\hline
ID & $\omega_c(\emptyset)$ & $\omega_c$ & Upper $\omega_f$ & Lower $\omega_q$ & Gap \\ 
\hline
\endhead

\hline
\multicolumn{6}{r}{Continued on the next page} \\
\endfoot

\hline
\endlastfoot

ID & $\omega_c(\emptyset)$ & $\omega_c$ &  Upper $\omega_f$ & Lower $\omega_q$ & Gap \\ \hline
11 & 7.778e-01 & 7.778e-01 & 7.778e-01 & 8.333e-01 & 5.556e-02 \\
20 & 7.778e-01 & 7.778e-01 & 7.778e-01 & 8.333e-01 & 5.556e-02 \\
40 & 8.889e-01 & 8.889e-01 & 8.889e-01 & 8.928e-01 & 3.948e-03 \\
44 & 7.778e-01 & 7.778e-01 & 7.778e-01 & 8.333e-01 & 5.556e-02 \\
51 & 7.778e-01 & 7.778e-01 & 7.778e-01 & 8.333e-01 & 5.556e-02 \\
54 & 8.889e-01 & 8.889e-01 & 8.889e-01 & 8.889e-01 & 8.164e-08 \\
58 & 8.889e-01 & 8.889e-01 & 8.889e-01 & 8.889e-01 & 1.260e-06 \\
69 & 7.778e-01 & 7.778e-01 & 7.778e-01 & 7.778e-01 & 3.902e-10 \\
72 & 8.889e-01 & 8.889e-01 & 8.889e-01 & 8.928e-01 & 3.947e-03 \\
78 & 1.000e+00 & 1.000e+00 & 1.000e+00 & 1.000e+00 & -1.928e-09 \\
93 & 7.778e-01 & 7.778e-01 & 7.778e-01 & 7.778e-01 & -9.849e-10 \\
107 & 7.778e-01 & 7.778e-01 & 7.778e-01 & 7.778e-01 & -1.689e-10 \\
122 & 7.778e-01 & 7.778e-01 & 7.778e-01 & 8.333e-01 & 5.556e-02 \\
135 & 8.889e-01 & 8.889e-01 & 8.889e-01 & 8.928e-01 & 3.948e-03 \\
145 & 8.889e-01 & 8.889e-01 & 8.889e-01 & 8.928e-01 & 3.948e-03 \\
155 & 8.889e-01 & 8.889e-01 & 8.889e-01 & 8.928e-01 & 3.948e-03 \\
178 & 8.889e-01 & 8.889e-01 & 8.889e-01 & 8.889e-01 & -9.142e-10 \\
181 & 8.889e-01 & 8.889e-01 & 8.889e-01 & 8.928e-01 & 3.948e-03 \\
208 & 7.778e-01 & 7.778e-01 & 7.778e-01 & 8.333e-01 & 5.556e-02 \\
221 & 8.889e-01 & 8.889e-01 & 8.889e-01 & 8.889e-01 & 2.156e-09 \\
237 & 8.889e-01 & 8.889e-01 & 8.889e-01 & 8.889e-01 & -2.465e-09 \\
244 & 8.889e-01 & 8.889e-01 & 8.889e-01 & 8.928e-01 & 3.948e-03 \\
249 & 7.778e-01 & 7.778e-01 & 7.778e-01 & 8.333e-01 & 5.556e-02 \\
253 & 7.778e-01 & 7.778e-01 & 7.778e-01 & 8.333e-01 & 5.556e-02 \\
293 & 8.889e-01 & 8.889e-01 & 8.889e-01 & 8.889e-01 & -8.412e-10 \\
295 & 8.889e-01 & 8.889e-01 & 8.889e-01 & 8.889e-01 & 2.214e-06 \\
298 & 8.889e-01 & 8.889e-01 & 8.889e-01 & 8.928e-01 & 3.948e-03 \\
309 & 7.778e-01 & 7.778e-01 & 7.778e-01 & 8.333e-01 & 5.556e-02 \\
315 & 8.889e-01 & 8.889e-01 & 8.889e-01 & 8.889e-01 & -1.182e-08 \\
330 & 8.889e-01 & 8.889e-01 & 8.889e-01 & 8.928e-01 & 3.948e-03 \\
350 & 7.778e-01 & 7.778e-01 & 7.778e-01 & 8.333e-01 & 5.556e-02 \\
355 & 8.889e-01 & 8.889e-01 & 8.889e-01 & 8.928e-01 & 3.948e-03 \\
368 & 7.778e-01 & 7.778e-01 & 7.778e-01 & 8.333e-01 & 5.556e-02 \\
378 & 8.889e-01 & 8.889e-01 & 8.889e-01 & 8.928e-01 & 3.948e-03 \\
383 & 8.889e-01 & 8.889e-01 & 8.889e-01 & 8.889e-01 & -3.923e-10 \\
387 & 7.778e-01 & 7.778e-01 & 7.778e-01 & 8.333e-01 & 5.556e-02 \\
393 & 1.000e+00 & 1.000e+00 & 1.000e+00 & 1.000e+00 & 1.935e-09 \\
402 & 8.889e-01 & 8.889e-01 & 8.889e-01 & 8.928e-01 & 3.948e-03 \\
405 & 8.889e-01 & 8.889e-01 & 8.889e-01 & 8.928e-01 & 3.948e-03 \\
418 & 7.778e-01 & 7.778e-01 & 7.778e-01 & 8.333e-01 & 5.556e-02 \\
422 & 8.889e-01 & 8.889e-01 & 8.889e-01 & 8.889e-01 & 4.037e-09 \\
425 & 8.889e-01 & 8.889e-01 & 8.889e-01 & 8.928e-01 & 3.948e-03 \\
431 & 7.778e-01 & 7.778e-01 & 7.778e-01 & 8.333e-01 & 5.556e-02 \\
438 & 1.000e+00 & 1.000e+00 & 1.000e+00 & 1.000e+00 & -3.621e-10 \\
463 & 8.889e-01 & 8.889e-01 & 8.889e-01 & 8.928e-01 & 3.948e-03 \\
476 & 8.889e-01 & 8.889e-01 & 8.889e-01 & 8.928e-01 & 3.948e-03 \\
477 & 7.778e-01 & 7.778e-01 & 7.778e-01 & 8.333e-01 & 5.556e-02 \\
492 & 7.778e-01 & 7.778e-01 & 7.778e-01 & 7.778e-01 & 3.168e-09 \\
501 & 7.778e-01 & 7.778e-01 & 7.778e-01 & 8.333e-01 & 5.556e-02 \\
510 & 8.889e-01 & 8.889e-01 & 8.889e-01 & 8.889e-01 & -3.970e-10 \\
\end{longtable}

\section{Three-Party XOR Games with Aggregated Parties}
\label{appendix:3_party_xor}
For the three-party XOR games $\{\mathrm{XOR}_i\}_{i=1}^4$ in \Cref{subsec:randomXORgames},
\begin{equation*}
   \begin{aligned}
       &\mathrm{XOR}_1: \hat{\beta}_{s_1,s_2,s_3} = (-1)^{s_1s_2s_3}, \\
       &\mathrm{XOR}_2: \hat{\beta}_{s_1,s_2,s_3} = (-1)^{s_1s_2}, \\
       &\mathrm{XOR}_3: \hat{\beta}_{s_1,s_2,s_3} = (-1)^{s_1s_3}, \\
       &\mathrm{XOR}_4: \hat{\beta}_{s_1,s_2,s_3} = (-1)^{s_1(s_2+s_3)}. 
   \end{aligned}
\end{equation*}
the quantum upper bound $\omega_q^\mathrm{upper}$ is sufficiently tight, in the sense that we can find quantum strategies that achieve a winning probability that matches this upper bound within numerical precision. 
For $\mathrm{XOR}_1$ and $\mathrm{XOR}_4$, the strategy in which every party always outputs $0$ respectively achieves the quantum values $\frac{7}{8}$ and $\frac{3}{4}$ trivially. For $\mathrm{XOR}_2$, let $v_3$ always output $0$. Then the winning probability of $\mathrm{XOR}_2$ becomes
\begin{equation*}
    \begin{aligned}
        p_{\mathrm{win}} =& \frac{1}{8} \sum_{s_1,s_2,s_3}\sum_{a_1,a_2} \beta_{s_1,s_2,s_3}(a_1 \oplus a_2)p(a_1,a_2|s_1,s_2)p(0|s_3) \\
         = & \frac{1}{4} \sum_{s_1,s_2}\sum_{a_1,a_2} \beta_{s_1,s_2}^\prime(a_1 \oplus a_2)p(a_1,a_2|s_1,s_2),
    \end{aligned}
\end{equation*}
with 
\begin{equation*}
    \beta_{s_1,s_2}^\prime(a_1\oplus a_2) \coloneqq \max[(-1)^{s_1s_2} \cdot (-1)^{a_1 \oplus a_2}, 0].
\end{equation*}
This is nothing but the winning probability of a CHSH game between the parties $v_1$ and $v_2$. Hence there exists a quantum strategy that achieves the value $\cos^2 \frac{\pi}{8}$ for $\mathrm{XOR}_2$, which matches the upper bound $0.8536$ computed via the third level of the NPA hierarchy. 
Similarly, for $\mathrm{XOR}_3$, let $v_2$ always output $0$. The three-party XOR game then reduces to a CHSH game between $v_1$ and $v_3$. Then they can win this game with probability $\cos^2 \frac{\pi}{8}$.

When we aggregate $v_1,v_2$, the winning probability is given by 
\begin{equation*}
\label{eq:agg_win_condition}
    p_{\mathrm{win}}^\mathrm{agg} = \frac{1}{8} \sum_{s_1,s_2,s_3}\sum_{a_2,a_3} \beta_{s_1,s_2,s_3}(a_2 \oplus a_3)p(a_2,a_3|(s_1,s_2),s_3).
\end{equation*}
For the game $\mathrm{XOR}_1$, its probabilistic predicate is 
\begin{equation*}
    \beta_{s_1,s_2,s_3}( a_2 \oplus a_3) = \max[(-1)^{s_1s_2s_3} \cdot (-1)^{a_2 \oplus a_3}, 0]
\end{equation*}
We can let $\tilde{s}_2 \coloneqq s_1\cdot s_2$
and define a binary-input, binary-output 2-party nonlocal game, whose probabilistic predicate is 
\begin{equation*}
    \beta_{\tilde{s}_2,s_3}^\prime( a_2 \oplus a_3) = \max[(-1)^{\tilde{s}_2s_3} \cdot (-1)^{a_2 \oplus a_3}, 0]
\end{equation*}
and the four inputs $\{(\tilde{0},0),(\tilde{0},1),(\tilde{1},0),(\tilde{1},1)\}$ are chosen with probability $\{\frac{3}{8},\frac{3}{8},\frac{1}{8},\frac{1}{8}\}$, respectively. Its winning probability is 
\begin{equation*}
    p_{\mathrm{win}}^\prime =  \frac{1}{8}\sum_{\tilde{s}_2,s_3}\sum_{a_2,a_3} (2+(-1)^{\tilde{s}_2}) \beta_{\tilde{s}_2,s_3}^\prime( a_2 \oplus a_3) p(a_2,a_3|\tilde{s}_2,s_3)
\end{equation*}
For the game $\mathrm{XOR}_1$ after aggregation, the parties can achieve the behavior 
$$p(a_2,a_3|(0,0),s_3) = p(a_2,a_3|(0,1),s_3) = p(a_2,a_3|(1,0),s_3) = p(a_2,a_3|\tilde{0},s_3)$$
by producing the same output for $(s_1,s_2) \in \{(0,0),(0,1),(1,0)\}$, so they win with the same probability as $p_{\mathrm{win}}^\prime$. Meanwhile, in the binary-input, binary-output 2-party nonlocal game, the parties can realize the behavior 
$$p(a_2,a_3|\tilde{0},s_3) = \frac{1}{3}\left[ p(a_2,a_3|(0,0),s_3) + p(a_2,a_3|(0,1),s_3) + p(a_2,a_3|(1,0),s_3) \right]$$
to win with the same probability as $p_{\mathrm{win}}$ in $\mathrm{XOR}_1$, 
by randomly choosing a strategy (using shared randomness) from three candidates when receives input $\tilde{s}_2 =\tilde{0}$. That is, the three-party XOR game $\mathrm{XOR}_1$ is equivalent to the binary-input, binary-output 2-party nonlocal game above, when $v_1$ and $v_2$ can communicate. For the nonlocal game, if we map the outputs $\{0,1\}$ respectively to outputs $\{+1,-1\}$ of some binary observables $K_{i}^{s_i}$ respectively for each party, we can express the winning probability $p_{\mathrm{win}}^\prime$ of a quantum strategy in operator form as
\begin{equation*}
    p_{\mathrm{win}}^\prime = \frac{1}{2} + \frac{1}{16}\langle B_\text{CHSH}^{\alpha=3}\rangle,
\end{equation*}
where $\langle \cdot \rangle$ denotes the expectation with some quantum state $\vert \psi \rangle$, with 
\begin{equation*}
     B_\text{CHSH}^{\alpha=3}  = 3  K_{2}^{\tilde{0}}K_3^0 + 3  K_{2}^{\tilde{0}}K_3^1 + K_{2}^{\tilde{1}}K_3^0 - K_{2}^{\tilde{1}}K_3^1.
\end{equation*}
Here, $B_\text{CHSH}^{\alpha=3}$ is a generalized CHSH expression 
whose quantum value is $2\sqrt{1+3^2},$ derived analytically in~\cite{PhysRevLett.108.100402}. Consequently, the quantum value $\omega_q^\mathrm{agg}$ of the aggregated three-party XOR game $\mathrm{XOR}_1$ is $\frac{4+\sqrt{10}}{8}$. 

Similarly, for $\mathrm{XOR}_3$, 
we can take the map $\{(0,0),(0,1)\} \mapsto \tilde{0}$, $\{(1,0),(1,1)\} \mapsto \tilde{1}$ on $(s_1,s_2)$, which transforms $\mathrm{XOR}_3$ into a CHSH game between $v_1$ and $v_3$ when $v_1$ and $v_2$ are aggregated. So the quantum value is $\cos^2 \frac{\pi}{8}$.

For $\mathrm{XOR}_4$, the winning condition is 
\begin{equation*}
    a_2 \oplus a_3 = s_1(s_2\oplus s_3) 
\end{equation*}
by \Cref{eq:agg_win_condition}. Let $a_2' = a_2 \oplus s_1s_2$, we can write the winning condition as 
\begin{equation*}
    a_2' \oplus a_3 =  s_1(s_2\oplus s_3) \oplus s_1s_2 = s_1s_3.
\end{equation*}
This is precisely the winning condition of the CHSH game, with inputs $s_1$ and $s_3$ and outputs $a_2'$ and $a_3$, respectively. So the quantum value is $\cos^2 \frac{\pi}{8}$.


Finally, for $\mathrm{XOR}_2$, let $v_3$ always output $0$, then the game reduces to a CHSH game between $v_1$ and $v_2$ as discussed in~\Cref{subsec:randomXORgames}. Since $v_1$ and $v_2$ can win with certainty if they can communicate, so the quantum value $\omega_q^\mathrm{agg}$ of the aggregated $\mathrm{XOR}_2$ game is $1$.

\end{document}